\documentclass[aap,MSNbibl,seceqn,citesort,dvips]{arximspdf}
\usepackage{accents}

\doi{10.1214/11-AAP805} 
\volume{23}
\issue{1}
\pubyear{2013}
\firstpage{1}
\lastpage{65}

\makeatletter

\renewcommand{\underbar}{\underaccent{\bar}}

\newtheorem{theorem}{Theorem}[section]
\newtheorem{proposition}[theorem]{Proposition}
\newtheorem{lemma}[theorem]{Lemma}
\newtheorem{corollary}[theorem]{Corollary}

\newproclaim{definition}[theorem]{Definition}
\newproclaim{assumption}[theorem]{Assumption}
\newproclaim{remark}[theorem]{Remark}

\newcommand{\mc}[1]{{\mathcal #1}}
\newcommand{\bb}[1]{{\mathbb #1}}

\renewcommand{\epsilon}{\varepsilon}
\newcommand{\wwidetilde}{\tilde}
\renewcommand{\widehat}{\hat}

\newcommand{\llangle}{\langle\!\langle}
\newcommand{\rrangle}{\rangle\!\rangle}

\newcommand{\spec}{\operatorname{spec}}
\newcommand{\Av}{\mathop{\operatorname{Av}}}

\newcommand{\bbI}{{\mathbb I}}
\newcommand{\bbP}{{\mathbb P}}
\newcommand{\bbR}{{\mathbb R}}
\newcommand{\bbT}{{\mathbb T}}
\newcommand{\bbZ}{{\mathbb Z}}

\renewcommand{\a}{\alpha}
\newcommand{\e}{\varepsilon}

\newcommand{\g}{\gamma}
\renewcommand{\k}{\kappa}

\newcommand{\p}{\pi}

\newcommand{\s}{\sigma}
\renewcommand{\t}{\tau}

\newcommand{\z}{\zeta}

\newcommand{\D}{\Delta}
\renewcommand{\L}{\Lambda}

\renewcommand{\O}{\Omega}

\newcommand{\upbar}[1]{\bar{#1}}
\newcommand{\downbar}[1]{\underbar{#1}}

\makeatother

\begin{document}
\begin{frontmatter}

\title{Large deviation principles for nongradient weakly asymmetric stochastic lattice gases}
\runtitle{LDP for weakly asymmetric stochastic lattice gases}

\begin{aug}
\author[A]{\fnms{Lorenzo} \snm{Bertini}\ead[label=e1]{bertini@mat.uniroma1.it}},
\author[A]{\fnms{Alessandra} \snm{Faggionato}\thanksref{t1}\ead[label=e2]{faggionato@mat.uniroma1.it}}\\
\and
\author[B]{\fnms{Davide} \snm{Gabrielli}\corref{}\thanksref{t2}\ead[label=e3]{gabriell@univaq.it}}
\runauthor{L. Bertini, A. Faggionato and D. Gabrielli}
\affiliation{Universit\`a di Roma La Sapienza, Universit\`a di Roma La
Sapienza and~Universit\`a~dell'Aquila}
\address[A]{L. Bertini\\
A. Faggionato\\
Dipartimento di Matematica\\
Universit\`a di Roma La Sapienza\\
P.le Aldo Moro 2, 00185 Roma\\
Italy\\
\printead{e1}\\
\phantom{E-mail: }\printead*{e2}}
\address[B]{D. Gabrielli\\
Dipartimento di Matematica\\
Universit\`a dell'Aquila\\
67100 Coppito, L'Aquila\\
Italy\\
\printead{e3}} 
\end{aug}

\thankstext{t1}{Supported by the European Research Council through the
``Advanced Grant'' PTRELSS 228032.}

\thankstext{t2}{Supported by the PRIN 20078XYHYS.}

\received{\smonth{11} \syear{2010}}

%
\begin{abstract}
We consider a lattice gas on the discrete $d$-dimensional torus
$(\bb Z/N \bb Z)^d$ with a generic translation invariant, finite
range interaction satisfying a uniform strong mixing condition. The
lattice gas performs a Kawasaki dynamics in the presence of a weak
external field $E/N$. We show that, under diffusive rescaling, the
hydrodynamic behavior of the lattice gas is described by a nonlinear
driven diffusion equation. We then prove the associated dynamical
large deviation principle. Under suitable assumptions on the
external field (e.g., $E$ constant), we finally analyze the
variational problem defining the quasi-potential and characterize
the optimal exit trajectory. From these results we deduce the
asymptotic behavior of the stationary measures of the stochastic
lattice gas, which are not explicitly known. In particular, when the
external field $E$ is constant, we prove a stationary large
deviation principle for the empirical density and show that the rate
function does not depend on~$E$.
\end{abstract}

%

\end{frontmatter}

\section{Introduction}\label{intro}

A classical topic in nonequilibrium statistical mechanics is the
analysis of stationary measures (steady states) for interacting
particle systems with driving fields. Here we focus on driven
diffusive systems, a typical example being the ionic conduction. As
microscopic model we consider high temperature stochastic lattice
gases with short range and translation invariant interaction
\mbox{\cite{ELS,GLMS,KLS,SZ,S}}. More precisely, let $\Lambda$ be a box in
$\bb Z^d$ that we consider with periodic boundary conditions. Each
site $x\in\Lambda$ can be either occupied or empty, the particle
configuration is therefore described by the occupation numbers
$\eta_x\in\{0,1\}$, $x\in\Lambda$. Consider a translation invariant
Gibbs measure $\mu_\Lambda$ with short range interactions on the
configuration space $\Omega_\Lambda=\{0,1\}^\Lambda$ and let
$H_\Lambda$ be the corresponding Hamiltonian so that
$\mu_\Lambda(\eta)\propto\exp\{-H_\Lambda(\eta)\}$. Note
that we
included the temperature in the definition of $H_\Lambda$. The
(symmetric) Kawasaki dynamics is then defined as a Markov chain on
$\Omega_\Lambda$ in which the allowed transitions are the exchanges
of the occupation numbers between nearest neighbor sites. The jump
rate $c^0_{x,y}$ associated to the bond $\{x,y\}$ satisfies the
\textit{detailed balance} condition with respect to the Hamiltonian
$H_\Lambda$, that is,
%
%
\begin{equation}
\label{detbal}
c^0_{x,y}(\eta^{x,y}) =
c^0_{x,y}(\eta) \exp\{\nabla_{x,y}H_\Lambda(\eta)\},
\end{equation}
where $\eta^{x,y}$ is the configuration obtained from $\eta$ by
exchanging the occupation numbers in $x$ and $y$ and $\nabla_{x,y}
H_\Lambda(\eta) = H_\Lambda(\eta^{x,y}) -H_\Lambda(\eta)$.

We regard the symmetric Kawasaki dynamics as the reference system and
model the effect of a driving field $E$ by replacing the reference
rates $c^0$ with the (asymmetric) rates $c^E$ satisfying the
\textit{local detailed balance} condition. In the case of a constant
driving field $E$ this condition reads
%
%
\begin{eqnarray}
\label{1cE}
c_{x,y}^E (\eta^{x,y} ) &=&
c_{x,y}^E (\eta) \exp\{ W_{x,y}(\eta) \},
\nonumber\\[-8pt]\\[-8pt]
W_{x,y} (\eta) &=&
\nabla_{x,y} H_\Lambda(\eta) + (\eta_y-\eta_x) E\cdot(y-x),\nonumber
\end{eqnarray}
where $\cdot$ is the inner product in $\bb R^d$. Observe that
$W_{x,y}$ is the total work done in the exchange of $\eta_x$ and
$\eta_y$. When the driving field $E$ is not constant, the right-hand
side of the second equation in (\ref{1cE}) has to be properly
modified. We remark that, in view of the periodic boundary
conditions, a nonvanishing constant field is not conservative and
therefore (\ref{1cE}) does not lead to a Gibbsian form of the invariant
measures. We assume that the rates $c^E$ are strictly positive.

The total number of particles $N_\Lambda=\sum_{x\in\Lambda}\eta_x$
is conserved by the Kawasaki dynamics. In view of the strict
positivity of the transition rates, for each integer
$K=0,\ldots,|\Lambda|$ the chain is irreducible on the subset
$\Omega_{\Lambda,K}$ of the configuration space with $K$ particles.
Therefore, on $\Omega_{\Lambda,K}$ there exists a unique invariant
measure that we denote by $\nu_{\Lambda,K}^E$. If $E=0$, by the
detailed balance condition (\ref{detbal}), $\nu_{\Lambda,K}^0$~is
the canonical measure corresponding to the Hamiltonian $H_\Lambda$,
that is, it is the measure $\mu_\Lambda$ conditioned to $\{N_\Lambda
=K\}$. For nonvanishing\vspace*{1pt} driving fields $E$, a main issue is to
understand the behavior of the measure $\nu_{\Lambda,K}^E$ in the
\textit{thermodynamic limit} $\Lambda\nearrow\bb Z^d$, $K\to\infty$
with $K/|\Lambda|\to\upbar{\rho}\in[0,1]$. About this problem there
are only few rigorous results and not much is known. In the case of
constant driving field, there are, however, some quite interesting
conjectures that we next briefly recall.

Let $\tau_x\dvtx\Omega_\Lambda\to\Omega_\Lambda$ be the translation by
$x$, the symmetric rates $c^0$ satisfy the \textit{gradient condition}
if for each bond $\{x,y\}$
%
%
\begin{equation}
\label{1grad}
c^0_{x,y}(\eta) (\eta_x-\eta_y)= h( \tau_x\eta) -h(\tau_y \eta)\vadjust{\goodbreak}
\end{equation}
for some local function $h\dvtx\Omega_\Lambda\to\bb R $.
As shown in~\cite{KLS}, if the symmetric rates $c^0$ satisfy the
gradient condition, then $\nu_{\Lambda,K}^E$ does not depend on the
driving field and therefore coincides with the canonical Gibbs measure
associated to the Hamiltonian $H_\Lambda$. In the case of the
exclusion process, for which $H_\Lambda=0$, the previous statement
corresponds to the fact that the uniform measure on
$\Omega_{\Lambda,K}$ is reversible in the symmetric case and invariant
in the asymmetric one. On the other hand, the gradient condition is
quite restrictive (\cite{S}, Section~II.2.4), and the generic picture is
believed to be qualitatively different. In particular, as conjectured
in~\cite{GLMS} and~\cite{S}, Section~II.1.4, for nongradient models the
following behavior is expected (recall we are only concerned with the
high temperature regime):
\begin{longlist}[(iii)]
\item[(i)]
for each density $\upbar{\rho}\in[0,1]$ there exists a unique translation
invariant thermodynamic limit of the sequence
$\{\nu_{\Lambda,K}^E\}$ that we denote by $\nu_{\upbar{\rho}}^E$;
\item[(ii)]
in dimension $d=1$ the measure $\nu_{\upbar{\rho}}^E$ has exponentially
decaying correlations;
\item[(iii)]
in dimension $d\ge2$ the pair correlation of $\nu_{\upbar{\rho}}^E$ decays
as a power law.
\end{longlist}
As far as we know, there are no clear expectations whether the measure
$\nu_{\upbar{\rho}}^E$ is Gibbsian or not (see, however, the result in
\cite{Am}).

We here analyze the asymptotic behavior of the sequence
$\{\nu_{\Lambda,K}^E\}$ in a \textit{scaling limit} setting. Given the
$d$-dimensional torus $\bb T^d =\bb R^d/\bb Z^d$ (which we regard as
the macroscopic domain) and a scaling parameter $N$, we take as
microscopic domain the box in $\bb Z^d$ with side length $N$ and
periodic boundary conditions that we denote by $\bb T^d_N$. In view of
the natural embedding $x\mapsto x/N$, the set $\bb T^d_N$ can be
regarded as a discrete approximation of $\bb T^d$. We then fix a
macroscopic field $E$ on $\bb T^d$ and let $E_N=E/N$ be its microscopic
counterpart. In this setting, the corresponding Kawasaki dynamics is
called \textit{weakly asymmetric}. To each configuration
$\eta\in\Omega_{\bb T^d_N}$ we associate the piecewise constant
function $\pi^N(\eta)$ on $\bb T^d$ which is equal to $\eta_x$ on the
cube $x/N +[0,1/N)^d$, $x\in\bb T^d_N$. The map $\pi^N$ from
$\Omega_{\bb T^d_N}$ to the set of functions on $\bb T^d$ is called
\textit{empirical density}. Given $\upbar{\rho}\in[0,1]$ and a sequence
$\{K_N\}$ such that $K_N/N^d\to\upbar{\rho}$, we let $P^E_N$ be the law
of the empirical density when the configuration $\eta$ is sampled
according to $\nu_{\bb T^d_N,K_N}^{E_N}$, namely,\vspace*{1pt} $P^E_N=
\nu_{\bb T^d_N,K_N}^{E_N} \circ(\pi^N)^{-1}$. The original question is
then formulated in terms of the asymptotic behavior of the sequence
$\{P^E_N\}$ as $N\to\infty$. In this paper, we describe this behavior
by proving the corresponding large deviation principle. In the case of
constant driving field, the rate functional can be directly expressed
in terms of the thermodynamic free energy of the reference system. In
particular, it does not depend on the driving field and coincides with
the one associated to the sequence of canonical Gibbs measures
$\{\nu_{\bb T^d_N,K_N}^0\}$. This result shows that, as far as
stationary large deviations of the empirical density are concerned,
weakly asymmetric nongradient stochastic lattice gases behave as
gradient models. We obtain an explicit formula for the rate function
also for nonconstant driving field provided a suitable orthogonality
condition holds. We emphasize that the choice of the periodic boundary
conditions is crucial for the above result. Indeed, as shown in
\cite{Ba,BDGJLn+1,BGL,ED}, for one-dimensional (gradient) weakly
asymmetric boundary driven stochastic lattice gases the presence of a
driving field, even in the weakly asymmetric regime, does effect the
stationary rate function.

The basic strategy of the proof follows the dynamical/variational
approach introduced in~\cite{BDGJL2}. This amounts to first analyzing
the dynamical behavior of the weakly asymmetric Kawasaki process in a
fixed macroscopic time interval. The dynamical law of large numbers
for the empirical density is called the hydrodynamic scaling
limit and it is described as follows. If at time $t=0$ the
empirical density converges to some function $\gamma\dvtx\bb T^d \to[0,1]$,
then at later time it converges to the solution $u\equiv u_t(r)$,
$(t,r)\in\bb R_+\times\bb T^d$ of the nonlinear driven diffusion
equation
%
%
\begin{equation}
\label{hyeq}
\partial_t u + \nabla\cdot[ \sigma(u) E ]
= \nabla\cdot[ D(u) \nabla u ]
\end{equation}
with initial datum $u_0=\gamma$. In the above equation, the
\textit{diffusion coefficient} $D$ and the \textit{mobility} $\sigma$
are $d\times d$ matrices which are characterized in terms of the
symmetric dynamics. The proof of the hydrodynamic limit extends the
one given in~\cite{VY} for $E=0$. Given $\upbar{\rho}\in[0,1]$ we
denote by $\gamma^E_{\upbar{\rho}}\dvtx\bb T^d\to[0,1]$ the stationary
solution to (\ref{hyeq}) with total mass equal to $\upbar{\rho}$ and
observe that for constant $E$ we simply have
$\gamma^E_{\upbar{\rho}}=\upbar{\rho}$. Of course, as $N\to\infty
$ the
sequence $\{P^E_N\}$ weakly converges to the Dirac measure
concentrated in $\gamma^E_{\upbar{\rho}}$.

The next step is to prove the dynamical large deviation principle
associated to the hydrodynamic limit, that is, to compute the
asymptotic probability that the empirical density follows some trajectory
different from the solution to (\ref{hyeq}).
For gradient stochastic lattice gases, this has been proven for several
models (see, e.g.,~\cite{KL,KOV}). For nongradient models, the proof of
the dynamical large deviation principle is technically much more
involved and it has been achieved in~\cite{Q} for one-dimensional
Ginzburg--Landau models (see also~\cite{QRV}).
The basic approach to prove such a large deviation principle is the
one set forth in~\cite{Vld} which requires us to construct a suitable
perturbation of the original measure. For gradient lattice gases this
perturbation is obtained by modifying the driving field in such a way
that the fluctuation becomes the typical behavior. In the nongradient
case this is not enough and an additional nonlocal correction is
needed~\cite{Q}. Since our model is not restricted to one dimension
and its invariant measures are not product, we have new technical
issues with respect to the case studied in~\cite{Q}. The conclusion
is that the law of the empirical density in the macroscopic time
interval $[T_1,T_2]$ satisfies a large deviation principle with some
rate function $I^E_{[T_1,T_2]}(\cdot|\gamma)$ (here $\gamma\dvtx\bb
T^d\to
[0,1]$ is the macroscopic density at time $T_1$).

The final step is the analysis of the \textit{quasi-potential}~\cite{FW}
associated to the dynamical rate function $I^E_{[T_1,T_2]}(\cdot
|\gamma)$.
Given $\upbar{\rho}\in[0,1]$, this is the functional on the set of
functions $\rho\dvtx\bb T^d\to[0,1]$ defined by
\begin{eqnarray*}
V^E_{\upbar{\rho}}(\rho) &=& \inf_{T>0} \inf
\bigl\{I^E_{[-T,0]} (\pi| \gamma_{\upbar{\rho}} ),
\pi\dvtx [-T,0]\times\bb T^d\to[0,1]
\\
&&\hspace*{91pt} \mbox{such that }
\pi_{-T} =\gamma_{\upbar{\rho}}, \pi_0 =\rho\bigr\}.
\end{eqnarray*}
In particular, $V^E_{\upbar{\rho}}(\rho)$ is the minimal cost to
produce the fluctuation $\rho$ starting from the stationary solution
$\gamma_{\upbar{\rho}}$. In view of the conservation of the total
number of particles, $V^E_{\upbar{\rho}}(\rho)$ is finite only if the
total mass of $\rho$ is $\upbar{\rho}$.
As proven in~\cite{FW} for diffusion processes on $\bb R^n$ and in
\cite{BG,F} in the present case of stochastic lattice gases, the
quasi-potential $V^E_{\upbar{\rho}}$ is the large deviations rate
function of the sequence $\{P^E_N\}$.

We here show that the quasi-potential can be expressed in terms of
the thermodynamic free energy associated to the Hamiltonian
$H_\Lambda$ and characterize explicitly the optimal path realizing a
given fluctuation. The key observation is the following. Let
$\chi(\rho)$ be the \textit{compressibility} of the system (this is a
thermodynamic quantity which coincides with the reciprocal of the
second derivative of the free energy). Then the transport
coefficients in the hydrodynamic equation (\ref{hyeq}) satisfy the
\textit{Einstein relationship} $\sigma(\rho)=D(\rho) \chi(\rho)$
(\cite{S}, Section II.2.5); observe that while $D$ and $\sigma$ are
matrices, $\chi$ is a scalar. The Einstein relationship implies
that the vector field describing the flow given by the hydrodynamic
equation (\ref{hyeq}) admits an orthogonal decomposition with
respect to the metric associated to the dynamical large deviation
rate function. The characterization of the quasi-potential is then
achieved by using an argument analogous to the one for diffusion
processes in $\bb R^n$ (see~\cite{FW}, Theorem 4.3.1).

\section{The model}
\label{s2}

In this section we fix the notation (recall some basic concepts about
Gibbs measures) and define the weakly asymmetric Kawasaki dynamics.

\subsection{The lattice and the configuration space}
\label{notazione}

On $\bb R^d$ and on the $d$-dimen\-sional cubic lattice $\bb Z^d$ we
consider the norm $|x|:=|x|_\infty=\max_{i=1,\ldots,d} |x_i|$; we
denote by $d(\cdot,\cdot)$ the associated distance. The diameter of a
set $V\subset\bb Z^d$ with respect to $d(\cdot,\cdot)$ is denoted by
$\operatorname{diam} (V)$. Given $\ell\geq0$ and $x\in\bb Z^d$, we set
$\Lambda_{x,\ell} =\{y \in\bb Z^d\dvtx |y-x|\leq\ell\}$ and write
simply $\Lambda_\ell$ if $x=0$. The canonical basis, both in $\bb Z^d$
and in $\bb R^d$, is denoted by $e_1,\ldots,e_d$. If $\Lambda$ is a
finite subset of $\bb Z^d$, we write $\Lambda\subset\subset\bb Z^d$
and denote by $|\Lambda|$ the cardinality of $\Lambda$. The collection
of all finite subsets of $\bb Z^d$ is denoted by $\bb F$. Given an
integer $N$, we let $\bb T_N:= \bb Z / N \bbZ=\{0,\ldots, N-1\}$ so
that $\bb T_N^d $ is the discrete\vadjust{\goodbreak} $d$-dimensional torus of side length
$N$. Given $\Lambda\in\bb F$ and $\phi\dvtx\Lambda\to\bb R$ we let
$\Av_{x\in\Lambda} \phi(x):= |\Lambda|^{-1}\sum_{x\in\Lambda}
\phi(x)$ be the \textit{average} of $\phi$. The average over $\bb T_N^d$
is simply denoted by $\Av_{x}$. The \textit{bonds} in $\bb Z^d$ are the
(unordered) pairs $\{x,y\}$ with $x,y\in\bb Z^d$ such that $y=x\pm
e_i$ for some $i=1,\ldots,d$. The collection of all bonds in $\bb
Z^d$ is denoted by $\bb B$. Given $\Lambda\subset\bb Z^d$, we let
$\bb
B_\Lambda:=\{ b \in\bb B\dvtx b\subset\Lambda\}$ be the
collection of bonds in $\Lambda$ and denote by $\bb B_N$ the
collection of bonds in~$\bb T_N^d$.

Given $\Lambda\subset\bb Z^d$, the \textit{configuration space} in
$\Lambda$ is the set $\Omega_\Lambda:=\{0,1\}^\Lambda$; we also let
$\Omega:= \Omega_{\bb Z^d}$ and $\Omega_N:=\Omega_{\bb T_N^d}$. For
$V\subset\Lambda\subset\bb Z^d$ and $\eta\in\Omega_\Lambda$, the
natural projection of $\Omega_\Lambda$ to $\Omega_V$ is denoted by
$\eta_V$; we also write $\eta_x$ for $\eta_{\{x\}}$, $x\in\Lambda$. A
configuration $\eta\in\O_\L$ describes the microscopic state of the
lattice gas; a site $x \in\L$ is occupied by a particle if and only if
$\eta_x =1$. We consider the single spin space $\{0,1\}$ endowed with
the discrete topology and $\Omega_\Lambda$ with the product topology.
Given $\Lambda\subset\bb Z^d$, we let $\mc F_\Lambda$ be the
$\sigma$-algebra on $\Omega$ generated by the one-dimensional
projections $\eta_x$, $x\in\Lambda$. We also set $\mc F:=\mc F_{\bb
Z^d}$ and note it coincides with the Borel $\sigma$-algebra associated
to the product topology. If $V_1,V_2\subset\bb Z^d$ are disjoint, we
denote by $\eta_{V_1}\eta_{V_2}$ the configuration in $\Omega_{V_1\cup
V_2}$ equal to $\eta_{V_i}$ in $V_i$, $i=1,2$. For $V\subset\Lambda$,
$V\in\bb F$, the \textit{number of particles} $N_V\dvtx\Omega_\Lambda \to\bb
Z_+$ is the function $N_V(\eta):=\sum_{x\in V} \eta_x$, while the
\textit{density} $\bar{\eta}_V\dvtx\Omega_\Lambda\to[0,1]$ is $\bar{\eta}_V:=
\Av_{x\in V} \eta_x$. If $V=\Lambda_{x,\ell}$ for some $x\in\bb Z^d$
and $\ell\in\bb N$, the density in $\Lambda_{x,\ell}$ is simply denoted
by $\bar{\eta}_{x,\ell}$ omitting the subscript $x$ when $x=0$. The
same notation holds when referred to the discrete torus $\bb T ^d_N$.

Given $x\in\bb Z^d$, respectively, $x\in\bb T_N^d$, we define the
\textit{shift} $\tau_x\dvtx \Omega\to\Omega$, respectively, $\tau_x\dvtx
\Omega_N \to\Omega_N$ by $(\tau_x\eta)_y:= \eta_{y+x}$. The map
$\tau_x$ is naturally lifted to functions by setting $(\tau_x f)
(\eta):= f( \tau_{x} \eta)$.
Given $i,j=1,\ldots,d$, $i\neq j$, we denote by $R^{i,j}$ the rotation
by $\pi/2$ in the plane spanned by $e_i,e_j$,
that is, $R^{i,j}(\ldots,x_i,\ldots,x_j,\ldots) =
(\ldots,-x_j,\ldots,x_i,\ldots)$. We denote by $\mc R$ the collection
of all such rotations. Given $R\in\mc R$, the map $x\mapsto R x$ is
naturally lifted to configurations and functions by setting $ ( R
\eta)_x:= \eta_{R x}$ and $R f (\eta):= f(R \eta)$.
Given a function $f\dvtx \Omega\rightarrow\bbR$, its so-called
\textit{support} $\Delta_f$ is the smallest subset $V\subset\bb Z^d$
such that $f$ depends on $\eta$ only through the projection $\eta_V$.
If $\Delta_f\in\bb F$, the function $f$ is called \textit{local}.
Given a local function $f$, we let $\downbar{f}$ be the formal series
%
%
\begin{equation}
\label{2gvec}
\downbar{f}:= \sum_{x\in\bb Z^d} \tau_x f.
\end{equation}

\subsection{Gibbs measures}
\label{amicogibbs}

In this paper, by an \textit{interaction}, we mean a finite range,
translation invariant interaction as defined below.
%
%
\begin{definition}
\label{tint}
An \textit{interaction} $\Phi$ is a collection of real-valued local
function $\{\Phi_V\dvtx \Omega\to\bb R, V\in\bb F, |V|\ge2
\}$ such that:
\begin{longlist}[(iii)]
\item[(i)]for each $V\in\bb F$ with $|V|\ge2$
the support of $\Phi_V$ is $V$;
\item[(ii)]there exists $r_0\in\bb N$ called \textit{range} such
that $\Phi_{V}=0$ if $\operatorname{diam}(V) > r_0$;
\item[(iii)]for each $V\in\bb F$ with $|V|\ge2$ and
$x\in\bb Z^d$ we have $\tau_x\Phi_V=\Phi_{V+x}$.\vadjust{\goodbreak}
\end{longlist}
In some statements we also assume that the interaction is
\textit{isotropic}, that is, it satisfies:
\begin{longlist}[(iv)]
\item[(iv)]{for each $V\in\bb F$ with $|V|\ge2$ and each $R\in\mc
R$ we have $R \Phi_V=\Phi_{R V}$.}
\end{longlist}
\end{definition}

Given an interaction $\Phi$, a parameter $\lambda\in\bbR$ (called
\textit{chemical potential}), and a set $\Lambda\in\bb F$, we define the
\textit{Hamiltonian} $H_\Lambda^\lambda\dvtx\Omega\to\bb R$ by
%
%
\begin{equation}
\label{hamilton}
H_\L^\lambda( \eta):=
\sum_{V\dvtx V\cap\Lambda\neq\varnothing} \Phi_V(\eta)
+ \lambda\sum_{x \in\L} \eta_x,
\end{equation}
dropping the superscript in the case $\lambda=0$.
Given $\sigma\in\Omega$, called \textit{boundary condition}, we also set
$H_\Lambda^{ \lambda,\sigma}(\eta):= H_\Lambda^{ \lambda}(\eta
_\Lambda
\sigma_{\Lambda^\complement})$.
To the Hamiltonian $H_\L^\lambda$ and the boundary condition $\s$ we
associate the finite volume (grand-canonical) \textit{Gibbs measure} in
$\Lambda$, defined as the probability measure on $(\Omega,\mc F)$
given by
%
%
\begin{equation}
\label{2fvgm}
\mu_\Lambda^{ \lambda,\sigma} (\eta):=
\cases{
(Z_\Lambda^{ \lambda,\sigma})^{-1} \exp\{ -
H_\Lambda^{ \lambda,\sigma}(\eta) \},
&\quad if
$\eta_{\Lambda^\complement}= \sigma_{\Lambda^\complement}$,\vspace*{1pt}\cr
0, &\quad otherwise,}
\end{equation}
where the constant $Z_\Lambda^{ \lambda,\sigma}$, called
\textit{partition function}, is the proper normalization.
In addition, the \textit{canonical Gibbs measure} associated to the
interaction $\Phi$, boundary condition $\s$ and particle number $K\in
\{0,1, \ldots, |\L|\}$, is the probability measure on $(\Omega,\mc F)$
given by
%
%
\begin{equation}
\label{misurecanoniche}
\nu_{\L,K} ^\s(\cdot): =
\mu^{\lambda,\s} _\L( \cdot| N_\L=K),
\end{equation}
noticing that this measure does not depend on the chemical potential
$\lambda$. In the case of periodic boundary conditions, $\Lambda=\bb
T_N^d$, we denote the Hamiltonian, which has, of course, no boundary
condition, as $H_N^{\lambda}$ and by $Z_N^{ \lambda}$ the corresponding
partition function. The associated grand-canonical and canonical
Gibbs measures are denoted by $\mu_N^{\lambda}$ and $\nu_{N,K}$,
respectively. Finally, we write $\mu_N$, $H_N$ instead of
$\mu_N^0$,~$H_N^0$, respectively.

Given a probability measure $\mu$ and bounded measurable functions
$f,g$, we denote by $\mu(f)$ the expectation of $f$ with respect to
$\mu$ and by $\mu(f;g):= \mu(f g) - \mu(f) \mu(g)$ the
covariance, or pair correlation, between $f$ and $g$. Given a bounded
measurable function $f\dvtx\Omega\to\bb R$ and a set $\Lambda\in\bb F$,
we denote by $\mu_\Lambda^{\lambda,\cdot}(f)$ the real function
$\Omega\ni
\sigma\mapsto\mu_\Lambda^{\lambda,\sigma}(f)$. As simple to check,
the finite volume Gibbs measures defined in (\ref{2fvgm}) satisfy the
compatibility conditions
\[
\mu_\Lambda^{\lambda,\sigma}
( \mu_{\Lambda'}^{\lambda,\cdot}(f) ) =
\mu_\Lambda^{\lambda,\sigma} (f)\qquad
\forall\mbox{ local } f, \forall
\Lambda'\subset\Lambda\in\bb F.
\]

The definition of infinite volume Gibbs measure is then given in terms
of the so-called DLR equations as follows.
%
%
\begin{definition}
\label{t22}
Given $\lambda\in\bb R$, a probability measure $\mu$ on $(\Omega
,\mc
F)$ is called an \textit{infinite volume Gibbs measure} with chemical
potential $\lambda$ iff
%
%
\begin{equation}
\label{2dlr}
\mu( \mu_\Lambda^{\lambda,\cdot}(f) ) = \mu(f)\qquad
\forall\mbox{ local } f,
\forall\Lambda\in\bb F.
\end{equation}
\end{definition}

The compactness of $\Omega$ readily implies that the set of
(infinite volume) Gibbs measure is not empty. The nonuniqueness of
solutions to the DLR equations (\ref{2dlr}) corresponds to phase
transitions.
As stated in the \hyperref[intro]{Introduction}, our analysis is restricted to the
high temperature regime. This is specified by a \textit{uniform
strong mixing condition} on the interaction $\Phi$. Referring to~\cite{CM}
for the precise formulation, this condition basically
requires that the pair correlation
$\mu_{\Lambda}^{\lambda,\sigma}(f;g)$ between two local functions
$f$ and $g$ decays exponentially fast in the distance between their
supports $\Delta_f$ and $\Delta_g$. This decay is required to be
uniform with respect the volume $\Lambda$, the boundary condition
$\sigma$ and the chemical potential $\lambda$. To be precise, one
also needs to allow chemical potentials which are not constant. As
it is easy to show, the uniform strong mixing condition implies that
for each $\lambda\in\bb R$ there exists a unique infinite volume
Gibbs measure $\mu^\lambda$. Moreover, $\mu^\lambda$ has exponential
decay of pair correlations. In the one-dimensional case $d=1$,
standard transfer matrix arguments show that the uniform strong
mixing condition is always satisfied (recall that the interaction
has finite range). For the standard Ising model in $d=2$, the
results in~\cite{BCO,SS} imply that the uniform strong mixing
condition is satisfied for any supercritical temperature. Finally,
the uniform strong mixing condition holds if the single site
Dobrushin criterion (\cite{Err}, Section 3.2) is satisfied uniformly in the
chemical potential $\lambda$. In particular, it holds if the
interaction $\Phi$ is small enough, that is, in the high temperature
regime. Throughout this paper we assume that the interaction
$\Phi$ satisfies the uniform strong mixing condition as stated in
\cite{CM}, \textit{Property USMT} there, without further mention.

Fix a configuration $\sigma\in\Omega$ and a sequence
$\{\Lambda_n\}$ of sets in $\bb F$ invading $\bb Z^d$ such that
$\lim_{n\to\infty} |\partial^+_{r_0} \Lambda_n| /|\Lambda_n| =0$,
where $r_0$ is the range of the interaction and $\partial^+_{r_0}
\Lambda:=\{x \in\Lambda^\complement\dvtx d(x,\Lambda)
\le r_0\}$. A classical result in statistical mechanics (see,
e.g.,~\cite{Err}, Section 2.3), states that the \textit{pressure},
$p\dvtx\bb R\to
\bb R$,
\[
p(\lambda):= \lim_{n} \frac1{|\Lambda_n|}
\log Z_{\Lambda_n}^{\lambda,\sigma}
= \lim_{N\to\infty} \frac{1}{N^d} \log Z_N^\lambda
\]
is well defined, that is, the limits exist (the first is also
independent of $\sigma$ and the sequence $\{\Lambda_n\}$), and
are convex. In view of the uniform strong mixing condition (see
\cite{SS} and reference therein), the pressure $p$ is uniformly
convex and real analytic.
The \textit{free energy} $f\dvtx[0,1]\to\bb R$ is defined as the Legendre
transform of $p$, namely,
%
%
\begin{equation}
\label{2freee}
f(\rho):= \sup\{ \lambda\rho- p(\lambda),
\lambda\in\bb R \},
\end{equation}
which is a continuous uniformly convex function in $[0,1]$ and
real analytic in $(0,1)$.
Moreover, as $\rho\uparrow1$ and $\rho\downarrow0$, we have
$f'(\rho)\uparrow+\infty$ and $f'(\rho)\downarrow-\infty$,
respectively.
Given $\rho\in[0,1]$, let $\mu_\rho:= \mu^{f'(\rho)}$ be the (unique)
infinite volume Gibbs measure with chemical potential $f'(\rho)$. We
understand that $\mu_0$ and $\mu_1$ are, respectively, the Dirac
measures in the configurations identically equal to zero and one.
From the definition of the free energy and the regularity of $p$, we
then have $\mu_\rho(\eta_x)=\rho$, so that $\rho$ is the density.
We also define the \textit{compressibility} $\chi\dvtx[0,1] \to\bb
[0,\infty)$ as
%
%
\begin{equation}
\label{2chi}
\chi(\rho):= \sum_{x\in\bb Z^d} \mu_\rho(\eta_0;\eta_x) =
\frac
1{f''(\rho)},
\end{equation}
understanding that $\chi(0)=\chi(1)=0$. By using the uniform strong
mixing condition, it is not difficult to show the compressibility
$\chi$ satisfies the following bound. There exists a real
$C\in(0,\infty)$ such that for any $\rho\in(0,1)$,
%
%
\begin{equation}
\label{2bchi}
\frac1C \le\frac{\chi(\rho)}{\rho(1-\rho)} \le C.
\end{equation}

The free energy $f$ gives the asymptotic probability of deviations of
the density in the following sense. Fix $\upbar{\rho}\in(0,1)$,
recall $\upbar{\eta}_V$ is the average number of particles in $V$ and
let $\{\Lambda_n\}$ be a sequence invading $\bb Z^d$ as before.
The sequence of probability measures on $[0,1]$ given by
$\{\mu_{\bar{\rho}} \circ(\bar{\eta}_{\L_n})^{-1}\}$
satisfies a large deviation principle (see,
e.g.,~\cite{Err}, Theorem 2.4.3.1) with speed $|\L_n|$ and convex rate
function $f_{\bar{\rho}}\dvtx[0,1]\to[0,+\infty)$ given by
%
%
\begin{equation}
\label{2febr}
f_{\bar\rho}(\rho):= f(\rho)-f(\bar{\rho}) - f'(\bar{\rho})
(\rho-\bar{\rho}).
\end{equation}
The same result holds if one replaces the infinite volume Gibbs
measure $\mu_{\bar{\rho}}$ with a finite volume Gibbs measure,
either with a fixed boundary condition $\sigma$ or with periodic
boundary, on $\L_n$ with chemical potential $f'(\bar\rho)$.

\subsection{Kawasaki dynamics}
\label{dinamocremisi}
Having introduced the formalism of the lattice gases at equilibrium,
here we define the dynamics we are interested in.

Given a bond $\{x,y\}\in\bb B$ and $\eta\in\Omega$, we let
$\eta^{x,y}$ be the configuration obtained from $\eta$ by exchanging
the occupation numbers in $x$ and $y$, that is,
\[
(\eta^{x,y})_z:=
\cases{
\eta_y, &\quad if $ z=x$,\cr
\eta_x, &\quad if $ z=y$,\cr
\eta_z, &\quad otherwise,}
\]
and let $\nabla_{x,y}$ be the operator defined by
$(\nabla_{x,y} f) (\eta):=f(\eta^{x,y}) -f(\eta)$,
where $f\dvtx\Omega\to\bb R$.
Recall that $\Omega_N:=\{0,1\}^{\bb T_N^d}$ is the configuration
space in the discrete $d$-dimensional torus $\bb T_N^d$ of side
length $N$. The \textit{symmetric Kawasaki dynamics} is then defined
by the Markov generator $L_{0,N}$ acting on functions $f\dvtx\Omega_N\to
\bb R$ as
%
%
\begin{equation}
\label{2sg}
L_{0,N} f (\eta):= N^2 \sum_{\{x,y\}\in\bb B_N}
c_{x,y}^0(\eta) \nabla_{x,y} f (\eta),
\end{equation}
where we recall that $\bb B_N$ is the collection of (unordered)
bonds in $\bb T_N^d$. Note that the generator has been speeded up by
the factor $N^2$ which corresponds to the diffusive scaling. We need
some conditions, that are detailed below, on the jump rates
$c^0_{x,y}$ (recall $r_0$ is the range of the interaction $\Phi$).\vadjust{\goodbreak}
%
%
\begin{definition}
\label{tjr0}
The \textit{symmetric jump rates}
$c^0_{x,y}\dvtx \Omega\to\bb R_+$, $\{x,y\}\in\bb B_N$ satisfy the
following conditions.
\begin{longlist}[(iii)]
\item[(i)] \textit{Detailed balance.} For any $\{x,y\}\in\bb B_N$
and $\eta\in\Omega_N$ we have
\[
c^0_{x,y}(\eta^{x,y}) = c^0_{x,y}(\eta)
\exp\{ \nabla_{x,y} H_N (\eta) \}.
\]
\item[(ii)] \textit{Finite range.} The support of $c^0_{x,y}$ is a
subset of $\{z\in\bb T_N^d\dvtx \break d(z,\{x,y\}) \le r_0\}$.
\item[(iii)] \textit{Translation invariance.} For each $\{x,y\} \in\bb
B_N$ and $z\in\bb T_N^d$ we have $\tau_z c^0_{x,y} =
c^0_{x+z,y+z}$.\vspace*{1pt}
\item[(iv)] \textit{Positivity and boundedness.} There exists $C\in
(0,\infty)$ such that for any $\{x,y\}\in\bb B_N$ we have
$C^{-1}\le c^0_{x,y} \le C$.
\end{longlist}
In some statements we also assume that the jump rates are isotropic,
namely:

\begin{longlist}[(v)]
\item[(v)] \textit{Rotation invariance.}
For each $\{x,y\} \in\bb B_N$ and each $R\in\mc R$
we have $R c^0_{x,y} = c^0_{Rx, Ry}$.
\end{longlist}
\end{definition}

Note that we consider the jump rates $c_{x,y}^0$ as functions on
$\Omega$ and not $\Omega_N$. In view of the finite range assumption,
the generator $L_{0,N}$ is well defined on $\Omega_N$ as soon as
$N>r_0$. The detailed balance condition implies that the generator
$L_{0,N}$ is self-adjoint in $L^2(\Omega_N,d\mu_N^\lambda)$ for any
$\lambda\in\bb R$. Since the Kawasaki dynamics conserves the total
number of particles, the ergodic measures for $L_{0,N}$ are the
canonical Gibbs measures $\nu_{N,K}$ on $\bb T_N^d$. In
\cite{CM,LY,Y} it is shown that if the interaction satisfies the
uniform strong mixing condition then the spectral gap of the
generator $L_{0,N}$ considered on $L^2(\Omega_N,\nu_{N,K})$ is of
order one uniformly in $N$ and $K$ (recall that $L_{0,N}$ has been
speeded up by $N^2$).

We next extend the previous symmetric dynamics by allowing the
presence of an external field $E$ of order $1/N$. Let $\bb T^d:= \bb
R^d/\bb Z^d$ be the $d$-dimen\-sional torus of side length one [the
coordinate on $\bb T^d$ is denoted by $r=(r_1,\ldots,r_d)]$. The
gradient and the divergence on $\bb T^d$ are, respectively, denoted by
$\nabla$ and $\nabla\cdot$. We denote by $\langle\cdot,\cdot
\rangle$
the inner product in $L^2(\bb T^d,dr)$. Let $\wwidetilde{\bb
B}_N$ be the collection of \textit{ordered} bonds in $\bb T_N^d$.
Given a\vspace*{1pt} $C^1$ vector field $E\dvtx \bb T^d \to\bb R^d$, we introduce a
\textit{discrete vector field} $E_N\dvtx \wwidetilde{\bb B}_N \to\bb R$ as
follows. Given $(x,y)\in\wwidetilde{\bb B}_N$, let $\gamma^N_{x,y}$ be
the oriented segment on $\bb T^d$ given by $\gamma^N_{x,y}(t):=
(x/N) (1-t) + (y/N) t$, $t\in[0,1]$. We then set
%
%
\begin{equation}
\label{2EN=}
E_N(x,y):= \int_0^1 dt \,E(\gamma^N_{x,y}(t)) \cdot
\frac{d}{dt} \gamma^N_{x,y}(t),
\end{equation}
where $\cdot$ is the inner product in $\bb R^d$.
Note that $ E_N(x,y)$ is the work done by the vector field $E$ along
the path $\gamma^N_{x,y}$. Moreover, $E_N(y,x)=-E_N(x,y)$ and, if $E$
is constant, we simply have $E_N(x,y) = (1/N) E\cdot(y-x)$.
The \textit{weakly asymmetric Kawasaki dynamics} is then defined by the
Markov generator $L_{E,N}$ acting on functions $f\dvtx\Omega_N\to\bb R$
as
%
%
\begin{equation}
\label{2LEN}
L_{E,N} f (\eta):= N^2 \sum_{\{x,y\}\in\bb B_N}
c_{x,y}^E(\eta) \nabla_{x,y} f (\eta),
\end{equation}
where the weakly asymmetric jump rates $c_{x,y}^E(\eta)$ satisfy the
so-called \textit{local detailed balance} condition (see, e.g.,
\cite{S}, Section II.1.4)
\[
c_{x,y}^E (\eta^{x,y} ) =
c_{x,y}^E (\eta) \exp
\{ \nabla_{x,y} H_N (\eta) + E_N (x,y) (\eta_y-\eta_x)
\}.
\]
Note indeed that $E_N(x,y) (\eta_y-\eta_x)$ does not depend on the
orientation of the bond $(x,y)\in\wwidetilde{\bb B}_N$.
In this paper, for simplicity, we shall consider the explicit choice
%
%
\begin{equation}
\label{tassi}
c_{x,y}^E(\eta):= c_{x,y}^0(\eta)
\exp\{ E_N(x,y) (\eta_x-\eta_y) / 2 \}
\end{equation}
in which $c^0_{x,y}$ are the jump rates of the symmetric Kawasaki
dynamics.

Given $T>0$, we denote by $D([0,T];\Omega_N)$ the Skorohod space
given by the set of c\`adl\`ag paths from $[0,T]$ to $\Omega_N$.
We consider $D([0,T];\Omega_N)$ endowed with the Skorohod topology
and the corresponding Borel $\sigma$-algebra.
Elements in $D([0,T];\Omega_N)$ are denoted by $\eta(t)$, $t\in
[0,T]$. The distribution of the Markov chain on $\Omega_N$ with
generator $L_{E,N}$ and initial distribution $\nu$ on $ \Omega_N$ is
a probability measure on $D([0,T];\Omega_N)$ which we denote by $\bb
P^{E,N}_\nu$. In particular, $\bb P^{0,N}_\nu$ is the distribution of
the symmetric Kawasaki dynamics defined by the generator $L_{0,N}$ in
(\ref{2sg}), with initial distribution $\nu$. If $\nu= \delta_\eta
$ with
$\eta\in\Omega_N$, we write simply $\bb P^{E,N}_\eta$.
The expectation with respect to $\bb P^{E,N}_\eta$ is denoted by
$\bb E^{E,N}_\eta$.

If the vector field $E$ is \textit{conservative}, that is, $E = -\nabla U$
for some $C^2$ function $U\dvtx\bb T^d\to\bb R$, then $E_N(x,y) =
U(x/N)-U(y/N)$ and the jump rates $c_{x,y}^E$ satisfy the detailed
balance condition with respect to the Hamiltonian
%
%
\begin{equation}
\label{UHamiltonian}
H_N^U(\eta):= H_N(\eta) + \sum_{x\in\bb T_N^d} U(x/N) \eta_x.
\end{equation}
In particular, if $E$ is conservative, the weakly asymmetric Kawasaki
dynamics is reversible with respect to the canonical or
grand-canonical Gibbs measures on $\bb T^d _N$ associated to the
Hamiltonian $H_N^U$.
On the other hand, when the vector field $E$ is not conservative, then
the weakly asymmetric Kawasaki dynamics is not reversible.
If the unperturbed jump rates $c^0_{x,y}$ satisfy the
gradient\vspace*{1pt}
condition (\ref{1grad}) and the vector field $E$ is constant then
(see~\cite{KLS} and~\cite{S}, Section II.1.4) the canonical Gibbs measures
$\nu_{N,K}$, which are the reversible measures for the symmetric
dynamics, are also the invariant measures of the weakly asymmetric
dynamics. This statement also holds if the field $E$ has vanishing
divergence (see~\cite{BDGJLn}, Section 2.5) for the precise formulation.
In the general case in which the gradient condition does not hold and
the vector field $E$ is not conservative, the invariant measures for
the asymmetric dynamics cannot be computed explicitly.

\section{Main results}
\label{smr}

\subsection{Hydrodynamic scaling limit}

The hydrodynamic scaling limit of the symmetric Kawasaki dynamics has
been proven in~\cite{VY}. As discussed here, the proof extends to the
weakly asymmetric case.

We set $M:= L^\infty(\bb T^d;[0,1])$ which we consider equipped
with the weak* topology, namely, a sequence $\{\gamma^n\} \subset M$
converges to $\gamma$ iff $\langle\gamma^n, \phi\rangle\to
\langle\gamma, \phi\rangle$ for any function $\phi\in L^1(\bb
T^d,dr)$, equivalently for any smooth function $\phi\in C^\infty(\bb
T^d)$ [recall $\langle\cdot,\cdot\rangle$ is the inner product in
$L^2(\bb T^d,dr)$]. The set $M$ is a compact Polish space that we
consider endowed with the corresponding Borel $\sigma$-algebra.
Given $N\ge1$ and $x\in\bb T_N^d$, let $Q_{1/N}(x/N)\subset\bb
T^d$ be the set $Q_{1/N}(x/N):= x/N+[0,1/N)^d$.
The \textit{empirical density} is the map $\pi^N\colon\Omega_N\to M$
defined by
%
%
\begin{equation}
\label{2ed} \pi^N (\eta) (r):= \sum_{x\in\bb T_N^d} \eta_x
\bb I_{Q_{1/N}(x/N)} (r),
\end{equation}
where $\bb I_A$ stands for the indicator function of the set $A$.

We say that a sequence $\{\eta^N\in\Omega_N\}$ is \textit{associated}
to the profile $\gamma\in M$ iff the sequence
$\{\pi^N(\eta^N)\}\subset M$ converges to $\gamma$. Given
$T_1<T_2$, we denote by $\mc M_{[T_1,T_2]}:=D([T_1,T_2];M)$ the
Skorohod space of paths from $[T_1,T_2]$ to $M$ equipped with its
Borel $\sigma$-algebra. Elements of $D([T_1,T_2];M)$ will be denoted by
$\pi\equiv\pi_t(r)$. Note that the evaluation map $\mc
M_{[T_1,T_2]} \ni\pi\mapsto\pi_t \in M$ is not continuous for
$t\in(T_1,T_2)$ but it is continuous for $t=T_1,T_2$. We also denote by
$\pi^N$ the map from $D([T_1,T_2];\Omega_N)$ to $\mc
M_{[T_1,T_2]}$ defined by $[\pi^N(\eta) ]_t:=
\pi^N(\eta(t))$.

Recall that $\mu_\rho$ is the unique infinite volume Gibbs measure
with density $\rho$ and the formal series defined in
(\ref{2gvec}). Given $\rho\in[0,1]$, the \textit{mobility}
$\sigma(\rho)$ is defined as the symmetric $d\times d$ matrix given
by the following variational formula~\cite{V,VY},
%
%
\begin{equation}
\label{2sr}
v \cdot\sigma(\rho) v:=
\inf_{f} \frac12 \mu_\rho\Biggl[
\sum_{i=1}^d c^0_{0,e_i} ( v_i [\eta_{e_i} - \eta_0]
+ \nabla_{0,e_i} \downbar{f} )^2 \Biggr],
\end{equation}
where $v\in\bb R^d$ and the infimum is carried out over all local
functions $f\dvtx\Omega\to\bb R$. Since $f$ is local, $\nabla_{0,e_i}
\downbar{f}$ is well defined as only finitely many terms in the sum
do not vanish. As shown in~\cite{VY}, Lemma 8.3, if the interaction
and the symmetric jump rates are isotropic then the mobility is a
multiple of the identity. Namely, there exists a scalar function,
still denoted by $\sigma$, such that
$\sigma_{i,j}(\rho)=\sigma(\rho)\delta_{i,j}$, $i,j=1,\ldots,d$.

Let $\varkappa^{(i)}\dvtx[0,1]\to\bb R_+$, $i=1,\ldots,d$, be the
function $ \varkappa^{(i)} (\rho):= \mu_\rho
([\eta_0-\eta_{e_i}]^2)$. As it is simple to check, the
functions $\varkappa^{(i)}$ satisfy the following bound. There
exists $C\in(0,\infty)$ such that for any $i=1,\ldots,d$ and
$\rho\in(0,1)$,
%
%
\begin{equation}
\label{2bkr}
\frac1C \le\frac{\varkappa^{(i)}(\rho)}{\rho(1-\rho)} \le C.
\end{equation}
The mobility $\sigma$ satisfies the following bounds. There exists a
real $C>0$ such that for any $\rho\in[0,1]$ and any $v\in\bb R^d$,
%
%
\begin{equation}
\label{2bsr} C^{-1} \sum_{i=1}^d \varkappa^{(i)}(\rho) v_i^2
\le
v \cdot\sigma(\rho) v \le C \sum_{i=1}^d \varkappa^{(i)}(\rho)
v_i^2.
\end{equation}
Indeed, the upper bound follows directly from the variational
expression (\ref{2sr}) by taking $f=0$, while the lower bound is
proven in~\cite{SY}.

Given $\rho\in[0,1]$, the \textit{diffusion matrix} $D(\rho)$ is defined
as the symmetric $d\times d$ matrix given by
%
%
\begin{equation}
\label{2D}
D(\rho):= \sigma(\rho) \frac{1}{\chi(\rho)} = \sigma(\rho)
f''(\rho),
\end{equation}
where the free energy $f$ has been defined in (\ref{2freee}) and
the compressibility $\chi$ (which is a scalar) has been defined in
(\ref{2chi}). Note that, by (\ref{2bchi}), (\ref{2bkr}) and
(\ref{2bsr}), the diffusion matrix $D$ is bounded and strictly
positive uniformly for $\rho\in[0,1]$. As follows from~\cite{VY} and
the arguments in~\cite{KL}, Chapter 7, the maps $[0,1]\ni\rho\to
\sigma
(\rho)$ and $(0,1)\ni\rho\to D(\rho)$ are continuous. In our
analysis, however, we need the smoothness of these maps on the
interval $[0,1]$. In the case in which the Gibbs measure is product,
that is, the interaction vanishes, this result is proven in~\cite{Ced}.
The general case remains, however, a long standing open problem in
hydrodynamic limits.
%
%
\begin{assumption}
\label{tregmob}
The maps $[0,1]\ni\rho\mapsto\sigma(\rho)$ and $[0,1]\ni\rho
\mapsto D (\rho)$ are continuously differentiable.
\end{assumption}

The hydrodynamic scaling limit for the weakly asymmetric Kawasaki
dynamics is stated as follows. Given a sequence $\{\eta^N \in
\Omega_N\}$, we set $\mc P^{E,N}_{\eta^N}:= \bb P^{E,N}_{\eta^N}
\circ(\pi^N)^{-1}$, that is, $\mc P^{E,N}_{\eta^N}$ is the law of the
empirical density when $\eta(t)$, $t\in[0,T]$, is sampled according
to $\bb P^{E,N}_{\eta^N}$. Then $\mc P^{E,N}_{\eta^N}$ is a
probability measure on the path space $\mc M_{[0,T]}$.
%
%
\begin{theorem}
\label{thl}
Fix $T>0$, a vector field $E\in C^1(\bb T^d;\bb R^d)$, a profile
$\gamma\in M$ and a sequence $\{\eta^N \in\Omega_N\}$
associated to $\gamma$. The sequence of probability
measures $\{\mc P^{E,N}_{\eta^N}\}$ on $\mc M_{[0,T]}$ converges
weakly to
$\delta_{u}$ where $u\equiv u_t(r)$ is the unique element of $\mc
M_{[0,T]}$ satisfying the two following conditions.
\begin{longlist}[(ii)]
\item[(i)] \textup{Energy estimate.} The weak gradient of $u$
is in $L^2([0,T]\times\bb T^d, dt \,dr;\bb R^d)$,
%
%
\begin{equation}
\label{2ee}
\int_0^T dt \,\langle\nabla u_t, \nabla u_t\rangle<
+\infty.
\end{equation}
\item[(ii)] \textup{Hydrodynamic equation.} The function $u$ is a weak
solution to
%
%
\begin{equation}
\label{2he}
\cases{
\partial_t u + \nabla\cdot[ \sigma(u) E ] =
\nabla\cdot[ D(u) \nabla u],
&\quad $(t,r) \in(0,T)\times\bb T^d$,\cr
u_0(r)=\gamma(r), &\quad $r \in\bb T^d$.}
\end{equation}
\end{longlist}
\end{theorem}

Of course, a function $u$ in $\mc M_{[0,T]}$ satisfying the energy
estimate (\ref{2ee}) is said to be a weak solution to (\ref{2he})
iff the identity
%
%
\begin{equation}\label{stellina}\quad
\langle u_T,H_T\rangle- \langle\gamma,H_0\rangle
= \int_0^T dt\,
[\langle u_t,\partial_t H_t\rangle+
\langle\sigma(u_t) E - D(u_t)\nabla u_t, \nabla H_t\rangle]\hspace*{-18pt}
\end{equation}
holds for any $H\equiv H_t(r)\in C^1([0,T]\times\bb T^d)$. We
emphasize that the above condition is meaningful in view of the energy
estimate. Since we assumed $E$ to be a $C^1$ vector field, the
uniqueness of a function $u\in\mc M_{[0,T]}$ satisfying the two
conditions stated in the theorem can be proven by repeating the
argument in~\cite{VY}. We emphasize uniqueness holds either if
$\sigma$ is Lipschitz (recall Assumption~\ref{tregmob}) or if
$\sigma$ is a multiple of the identity and continuous.

\subsection{Dynamical large deviation principle}

In order to state the large deviation principle associated to the law
of large numbers in Theorem~\ref{thl}, we first introduce the rate
functional. Fix a function $\gamma\in M$ corresponding to the initial
density profile. Given $\pi\in\mc M_{[0,T]}$ satisfying the energy
estimate [i.e., such that (\ref{2ee}) holds with $u$ replaced by
$\pi$], let $\ell_{\gamma,\pi}$ be the linear functional on
$C^{1}([0,T]\times\bb T^d)$ defined by
%
%
\begin{eqnarray}
\label{2lpH}
\ell_{\gamma,\pi} (H)&:=&
\langle\pi_T, H_T \rangle- \langle\gamma, H_0 \rangle
\nonumber\\[-8pt]\\[-8pt]
&&{} - \int_0^{T} dt\, [ \langle\pi_t, \partial_t H_t
\rangle
+\langle\sigma(\pi_t) E - D(\pi_t)\nabla\pi_t, \nabla H_t\rangle
].\nonumber
\end{eqnarray}
Note that $\ell_{\gamma,\pi}$ vanishes iff $\pi$ is a weak solution
to the
hydrodynamic equation (\ref{2he}).
The rate functional $I^E_{[0,T]}(\cdot|\gamma)\dvtx \mc M_{[0,T]} \to
[0,+\infty]$
is then defined by
%
%
\begin{equation}
\label{2Ig}\quad
I^E_{[0,T]}(\pi|\gamma):=
\sup_{H\in C^1([0,T]\times\bb T^d)} \biggl\{ \ell_{\gamma,\pi} (H)
- \int_0^T dt\,
\langle\nabla H_t, \sigma(\pi_t) \nabla H_t\rangle
\biggr\},\hspace*{-18pt}
\end{equation}
if $\int_0^T dt\, \langle\nabla\pi_t, \nabla\pi_t\rangle<
+\infty$ and $I^E_{[0,T]}(\pi|\gamma):= +\infty$ otherwise.
It is not difficult to check, by choosing
suitable test functions $H$ above, that $I^E_{[0,T]} (\pi|
\gamma)<+\infty$ implies $\pi\in C([0,T];M)$ and $\pi_0=\gamma$.

An application of Riesz's representation lemma allows us to write the
rate function $I^E_{[0,T]}(\cdot| \gamma)$ in a more explicit form
(\cite{KL}, Lemma 10.5.3). For this purpose, we introduce some Hilbert
spaces. Given a path $\pi\in\mc M_{[0,T]}$, let $\mc H^{1}
(\sigma(\pi))$ be the Hilbert space obtained by quotienting and
completing $C^1([0,T]\times\bb T^d)$ with respect to the
pre-inner product defined by
\[
\llangle G,H \rrangle_{1,\sigma(\pi)}
:= \int_0^T dt\,
\langle\nabla G_t, \sigma(\pi_t) \nabla H_t \rangle.
\]
Denote the norm in $\mc H^{1} (\sigma(\pi))$ by $\Vert\cdot
\Vert_{1,\sigma(\pi)}$ and let $\mc H^{-1} (\sigma(\pi))$ be the dual
space. The latter is a Hilbert space equipped with the norm $\Vert
\cdot\Vert_{-1, \sigma(\pi)}$ defined by
\[
\|\wp\|^2_{-1,\sigma(\pi)}:=\mathop{\sup_{ H\in\mc H^{1}
(\sigma(\pi)):}}_{
\| H\|_{1, \s(\pi) } =1} \wp(H)^2
= \sup_{H\in\mc H^{1} (\s(\pi)) } \bigl\{ 2 \wp(H) - \Vert H
\Vert^2_{1, \sigma(\pi)} \bigr\}.
\]
By density, in the above formula one can restrict to $H\in
C^1([0,T]\times\bb T^d)$.

Fix a path $\pi\in\mc M_{[0,T]}$ such that $I^E_{[0,T]} (\pi|
\gamma) < +\infty$, in particular, $\pi$ satisfies the energy estimate.
Since the right-hand side of (\ref{2Ig}) is finite, the linear
functional~$\ell_{\gamma,\pi}$, as defined in (\ref{2lpH}), extends
univocally to a continuous linear functional on $\mc H^1(\sigma(u))$,
that we still denote by $\ell_{\gamma,\pi}$. From (\ref{2Ig}) we
deduce $\| \ell_{\gamma,\pi} \|^2_{-1,\sigma(\pi)}= 4 I^E_{[0,T]}
(\p|\g)$. Therefore, by Riesz's representation lemma, there exists a
unique $\Psi_{\gamma,\pi}\in\mc H^1(\sigma(\pi))$ such that
%
%
\begin{equation}\label{riesz}
\ell_{\gamma,\pi}(H) =
2 \llangle\Psi_{\gamma,\pi}, H \rrangle
_{1,\sigma(\pi)}\qquad
\mbox{for any } H\in\mc H^1(\sigma(u)),
\end{equation}
thus leading to the identity $\|
\ell_{\gamma,\pi}\|_{-1,\sigma(\pi)}=2\|
\Psi_{\gamma,\pi}\|_{1,\s(\pi) }$. In conclusion, it holds
%
%
\begin{equation}
\label{2rI} I^E_{[0,T]} (\pi| \gamma) = \|
\Psi_{\gamma,\pi}\|_{1,\sigma(\pi)}^2 =\tfrac14 \|
\ell_{\gamma,\pi}\|^2_{-1,\sigma(\pi)}.
\end{equation}
In view of (\ref{riesz}), $\pi$ is a weak solution to
%
%
\begin{eqnarray}
\label{2hepsi}
\partial_t \pi+ \nabla\cdot[ \sigma(\pi)
(E + 2 \nabla\Psi_{\gamma,\pi}) ] &=&
\nabla\cdot[ D(\pi) \nabla\pi],\nonumber\\
&&\eqntext{(t,r) \in(0,T)\times\bb T^d,}\\
[-18pt]\\
\pi_0(r)&=&\gamma(r),\qquad r \in\bb T^d,\nonumber
\end{eqnarray}
so that $2 \nabla\Psi_{\gamma,\pi}$ can be interpreted as the extra
driving field to produce the fluctuation $\pi$.
%
%
\begin{theorem}
\label{tdldp} Fix $T>0$, a vector field $E\in C^1(\bb T^d;\bb R^d)$, a
profile $\gamma\in M$ and a sequence $\{\eta^N \in\Omega_N\}$
associated to $\gamma$. The sequence of probability measures $\{\mc
P^{E,N}_{\eta^N}\}$ on $\mc M_{[0,T]}$ satisfies a large deviation
principle with speed $N^d$ and good rate function
$I^E_{[0,T]}(\cdot|\gamma)$. Namely, $I^E_{[0,T]}(\cdot|\gamma)\dvtx
\mc M_{[0,T]} \to[0,+\infty]$ has\vspace*{1pt} compact level sets and\vadjust{\goodbreak}
for each closed set $\mc C \subset\mc M_{[0,T]}$ and each open set $\mc
O \subset\mc M_{[0,T]}$
%
%
\begin{eqnarray}
\label{2ub}
\limsup_{N\to\infty} \frac1{N^d} \log\mc P^{E,N}_{\eta^N}
( \mc C )
&\leq&- \inf_{\pi\in\mc C} I^E_{[0,T]} (\pi| \gamma),
\\
\label{2lb}
\liminf_{N\to\infty} \frac1{N^d} \log\mc P^{E,N}_{\eta^N}
( \mc O ) &\geq&- \inf_{\pi\in\mc O} I^E_{[0,T]} (\pi|
\gamma).
\end{eqnarray}
\end{theorem}

\subsection{The quasi-potential}
\label{sqp}

From now on we assume that the driving field $E$ admits the following
orthogonal decomposition.
%
%
\begin{definition}
\label{tbello}
The vector field $E\in C^1(\bb T^d;\bb R^d)$ is \textit{orthogonally
decomposable} iff it admits the following decomposition. There
exists\vspace*{1pt} a function $U\in C^2(\bb T^d)$ and a vector
field $\wwidetilde E \in C^1(\bb T^d; \bb R^d)$ such that
%
%
\begin{equation}
\label{2E=}\qquad
E = -\nabla U + \wwidetilde E,\qquad
\nabla\cdot\wwidetilde E =0,\qquad
\nabla U (r) \cdot\wwidetilde E(r) =0 \qquad\forall r\in\bb T^d.
\end{equation}
\end{definition}

Given a $C^1$ vector field $E$, the first two requirements in the
above definition are met by letting $U$ be a solution to the Poisson
equation $-\Delta U =\nabla\cdot E$ and then setting $\wwidetilde E
= E +\nabla U$. Then\vspace*{1pt} (\ref{2E=}) requires that for each $r\in\bb
T^d$ we have $\nabla U (r)\cdot\wwidetilde E (r) =0$. Observe that a
conservative or divergenceless vector field is orthogonally
decomposable; indeed in first case (\ref{2E=}) holds with $\tilde
E=0$, while in the second case (\ref{2E=}) holds with a constant
$U$ and $\wwidetilde E=E$. In the one-dimensional case $d=1$, a
vector field is orthogonally decomposable either if it is constant
or if it is conservative. On the other hand, when $d\ge
2$ there exist orthogonally decomposable vector fields for which the
decomposition (\ref{2E=}) is not trivial.
Although $U$ is univocally determined by (\ref{2E=}) apart an
additive constant, all the $U$-dependent definitions given
below are not affected by the choice of the additive constant.
In the sequel, we shall restrict to either one of the three following
cases: (i) $E$ is a conservative vector field,
(ii) $E$ is a constant vector field,
(iii) the mobility $\sigma$ is a multiple of the identity
and $E$ is orthogonally decomposable.
As stated above, if the interaction $\Phi$ and the symmetric jump
rates $c^0$ are isotropic, then $\sigma$ is indeed a multiple of the
identity.

Recall the definition of the free energy $f$ given in (\ref{2freee}).
Given an orthogonally decomposable field $E$ and $\upbar{\rho}\in
(0,1)$, let $\gamma_{\upbar{\rho}}\dvtx\bb T^d\to(0,1)$ be the
function satisfying
%
%
\begin{equation}
\label{2gbar}
f' ( \gamma_{\bar{\rho}} (r) )
+ U(r) = \alpha(\bar\rho),
\end{equation}
where $\alpha(\upbar{\rho}) \in\bb R$ is chosen so that $\int dr\,
\gamma_{\upbar{\rho}} (r) =\upbar{\rho}$. Equivalently,
$\gamma_{\bar{\rho}}$ is defined as $\gamma_{\bar{\rho}} (r)=
(f')^{-1}( -U(r)+c)$, where the constant $c$ is chosen such that
\mbox{$\int dr\, \gamma_{\bar{\rho}} (r) =\bar{\rho}$}. By the properties
of the free energy mentioned just after (\ref{2freee}), the function
$\gamma_{\bar{\rho}}$ is well defined.
When $\upbar{\rho}$ equals $0$ or $1$ then we define
$\gamma_{\bar{\rho}}$ as the function, respectively, identically equal
to $0$ or $1$.
A simple computation shows that, under either condition (i), (ii) or
(iii) above, for each $\upbar{\rho}\in[0,1]$ the function
$\gamma_{\upbar{\rho}}$ is a stationary solution of the hydrodynamic
equation (\ref{2he}). Moreover, as we show in Section~\ref{s6},
under the flow defined by the hydrodynamic equation (\ref{2he}) any
point in the closed subset $M(\upbar{\rho})\subset M$ defined by
%
%
\begin{equation}
\label{7Mbr}
M(\upbar{\rho}):=
\biggl\{\rho\in M\dvtx \int dr\, \rho(r) =
\upbar{\rho}\biggr\}
\end{equation}
converges as $t\to+\infty$ to the stationary solution
$\gamma_{\upbar{\rho}}$. Furthermore, this convergence is uniform with
respect to the initial condition.

We next define the \textit{quasi-potential} as in the classical
Freidlin--Wentzell theory for finite-dimensional diffusion processes
\cite{FW}. We denote\vspace*{1pt} by $I^E_{[T_1,T_2]}(\cdot|\gamma)$ the functional
(\ref{2Ig}) when the time window is $[T_1,T_2]$. Given
$\bar{\rho}\in[0,1]$ we then let the quasi-potential
$V^E_{\upbar{\rho}}\dvtx M \to[0,+\infty]$ be the functional defined by
%
%
\begin{equation}
\label{2qp}
V^E_{\upbar{\rho}}(\rho):= \inf_{T>0} \inf\bigl\{
I^E_{[-T,0]}(\pi| \gamma_{\upbar{\rho}} ),
\pi\in\mc M_{[-T,0]}\dvtx \pi_0 =\rho\bigr\}.
\end{equation}
Since $I_{[-T,0]}(\pi|\gamma)<+\infty$ implies $\pi_{-T}=\gamma$, the
quasi-potential $V^E_{\bar{\rho}}(\rho)$ measures the minimal cost to
reach the profile $\rho\in M$ starting from the stationary
solution~$\gamma_{\upbar{\rho}}$.

We can also define the quasi-potential by considering directly paths
defined on a semi-infinite time interval.
To this end, let $I^E_{[T_1,T_2]}\dvtx \mc M_{[T_1,T_2]} \to[0,+\infty]$
be the functional defined by
\[
I^E_{[T_1,T_2]}(\pi):=I^E_{[T_1,T_2]}(\pi|\pi(T_1)).
\]
This functional can also be expressed by the variational formula
(\ref{2Ig}) in which the linear functional $\ell_{\gamma,\pi}$ is
replaced by
%
%
\begin{eqnarray}
\label{ldav}
\ell_{\pi}(H) &:=& \langle\pi_{T_2},H_{T_2}\rangle-\langle
\pi_{T_1},H_{T_1}\rangle
\nonumber\\[-8pt]\\[-8pt]
&&{} -
\int_{T_1}^{T_2} dt\, [
\langle\pi_t,\partial_tH_t\rangle+
\langle\sigma(\pi_t) E - D(\pi_t) \nabla\pi_t
, \nabla H_t\rangle].\nonumber
\end{eqnarray}
Given $\upbar{\rho}\in[0,1]$, we define $\mathcal M_{(-\infty,0]}
(\upbar{\rho}) \subset\mathcal M_{(-\infty,0]} $ by
%
%
\begin{equation}
\label{infdav}
\mathcal M_{(-\infty,0]} (\upbar{\rho})
:=\Bigl\{\pi\in\mc M_{(-\infty,0]}\dvtx
\lim_{t\to-\infty}\pi_t=\gamma_{\upbar{\rho}}\Bigr\}.
\end{equation}
We then let $I^E_{(-\infty,0]}\dvtx \mathcal M_{(-\infty,0]}(\upbar
{\rho})
\to[0,+\infty]$ be the lower semicontinuous functional given by
%
%
\begin{equation}
\label{funinfdav}
I^E_{(-\infty,0]}(\pi):= \lim_{T\to+\infty} I^E_{[-T,0]}(\pi)
\end{equation}
observing that the limit on the right-hand side (possibly taking the
value $+\infty$) exists by monotonicity.
We finally let $\widehat{V}{}^E_{\upbar{\rho}}\dvtx M \to[0,+\infty]$ be
the functional defined by
%
%
\begin{equation}
\label{genqpdav}
\widehat{V}{}^E_{\upbar{\rho}} (\rho)
:= \inf\bigl\{ I^E_{(-\infty,0]} (\pi),
\pi\in\mc M_{(-\infty,0]}(\upbar{\rho})\dvtx \pi_0=\rho\bigr\}.\vadjust{\goodbreak}
\end{equation}
In the context of diffusion processes in $\bb R^n$, in view of the
continuity of the quasi-potential, it is simple to check that the
functionals defined in (\ref{2qp}) and (\ref{genqpdav}) are
identical. We show this is also the case in the present setting in
which the quasi-potential is only lower semicontinuous.

The next result states that the quasi-potential has a simple
representation in terms of the function $\gamma_{\upbar{\rho}}$, which
does not depend on the divergenceless part $\wwidetilde E$ in the
decomposition (\ref{2E=}). Moreover, the variational problem on the
right-hand side of (\ref{genqpdav}) has a unique minimizer that can be
explicitly characterized. We first introduce such optimal path.
Recall (\ref{7Mbr}). Fix $\upbar{\rho}\in[0,1]$, $\rho\in M
(\bar\rho)$, and let $v\dvtx [0,+\infty)\times\bb T^d \to[0,1]$ be the
weak solution to
%
%
\begin{eqnarray}
\label{2he*}
\partial_t v
+ \nabla\cdot[ \sigma(v)
( - \nabla U - \wwidetilde E ) ] &=&
\nabla\cdot[ D(v) \nabla v],
\nonumber\\
&&\eqntext{(t,r) \in(0,+\infty) \times\bb T^d,}\\
[-18pt]\\
v_0(r)&=&\rho(r),\qquad r \in\bb T^d.\nonumber
\end{eqnarray}
Note the change of sign in the field $\wwidetilde E$ with respect to
(\ref{2he}). Then, as we show is Section~\ref{s6}, $v_t\to
\gamma_{\upbar{\rho}}$ as $t\to+\infty$. Therefore, denoting by
$\theta$ the time reversal, that is, $(\theta v)_t:= v_{-t}$, it holds
$\theta v \in\mc M_{(-\infty,0]}(\upbar{\rho})$ so that $\theta v$ is
a legal test path for the variational problem (\ref{genqpdav}).
%
%
\begin{theorem}
\label{tqp=f}
Assume either one of the three following conditions:
\begin{longlist}[(iii)]
\item[(i)] $E$ is a conservative vector field;
\item[(ii)] $E$ is a constant vector field;
\item[(iii)] the mobility $\s$ is a multiple of the identity and $E$
is orthogonally decomposable.
\end{longlist}
For each $\upbar{\rho}\in[0,1]$ we have $V^E_{\upbar{\rho}} =
{\widehat V}{}^E_{\upbar{\rho}} = \mc F^U_{\upbar{\rho}}$, where the
functional $\mc F^U_{\upbar{\rho}}\dvtx M \to[0,+\infty)$ is given
by
%
%
\begin{equation}
\label{2Fbr}
\mc F^U_{\bar{\rho}}(\rho) =
\cases{
\displaystyle \int dr\,
f_{\bar{\rho}}^U(r,\rho(r)) , &\quad
if $ \rho\in M(\bar\rho)$, \vspace*{2pt}\cr
+\infty&\quad otherwise,}
\end{equation}
in which, recalling (\ref{2febr}), $f_{\bar{\rho}}^U\dvtx \bb T^d
\times[0,1] \to\bb R_+$ is the function
%
%
\begin{equation}
\label{2frU}
f_{\bar{\rho}}^U(r,\rho):=
\int_{\gamma_{\bar{\rho}}(r)}^\rho du
\int_{\gamma_{\bar{\rho}}(r)}^u dv\, f''(v) =
f_{\gamma_{\bar{\rho}}(r)}(\rho).
\end{equation}
Moreover, the unique minimizer for the variational problem on the
right-hand side of (\ref{genqpdav}) is the path $\theta v$, where $v$
is the weak solution to (\ref{2he*}).
\end{theorem}

Note that $\mc F^U_{\bar{\rho}}$ is a lower semicontinuous strictly
convex functional which attains its minimum for $\rho=
\gamma_{\bar{\rho}}$. Moreover, if $E$ has vanishing divergence then
$U$ is constant and $ \gamma_{\bar{\rho}}(r) \equiv\bar\rho$; in
particular $f_{\bar{\rho}}^U(r;\rho) $ does not depend on $r$ and
coincides with $f_{\bar\rho}(\rho)$ [see (\ref{2febr})].\vadjust{\goodbreak} In this
case, we drop the dependence on $U$ from the notation. Note, however,
that the optimal path $\theta v$ depends also on the divergenceless
part $\wwidetilde E$ in the decomposition (\ref{2E=}). As stated
before, the previous result is an infinite-dimensional analogue of
\cite{FW}, Theorem 4.3.1. The condition that $\sigma(\rho)$ is a
multiple of the identity can be slightly relaxed.

%
\begin{remark}
Assume $\sigma(\rho)=\sigma_0(\rho) \Sigma$ for some \textit{scalar}
function $\sigma_0\dvtx[0,1]\to[0,+\infty)$ and some \textit{constant}
symmetric strictly positive $d\times d$ matrix $\Sigma$. Replace the
condition (\ref{2E=}) on the driving field $E$ with the following
assumption. There exists a $C^2$ function $U\dvtx\bb T^d\to\bb R$ and a
$C^1$ vector field $\wwidetilde E\dvtx\bb T^d\to\bb R^d$ such that
$E=-\nabla U +\wwidetilde E$ with $\nabla U(r) \cdot\Sigma\wwidetilde
E(r) =0$ for any $r\in\bb T^d$ and $\nabla\cdot(\Sigma\wwidetilde
E)=0$. Then Theorem~\ref{tqp=f} still holds.
\end{remark}

\subsection{Stationary large deviation principle}

As a corollary of the large deviations analysis of the weakly
asymmetric dynamics and the characterization of the quasi-potential
in Theorem~\ref{tqp=f}, we deduce the asymptotic behavior of the
corresponding invariant measures.

We first discuss the case of the symmetric dynamics. As stated before,
in this case the ergodic invariant measures are the canonical Gibbs
measures $\nu_{N,K}$. Fix a sequence $\{K_N \}\subset\bb N$ such that
$N^{-d} K_N \to\bar{\rho} \in[0,1]$ and set
$P_N^0:=\nu_{N,K_N}\circ(\pi^N)^{-1}$, that is, $P_N^0$ is the law of
the empirical density when the configuration $\eta$ is sampled
according to $\nu_{N,K_N}$. Then the sequence of probability measures
on $M$ given by $\{ P_N^0 \}$ satisfies a large deviations principle
with speed $N^d$ and convex rate function $\mc F_{\bar{\rho}}$ (recall
that $\mc F_{\bar{\rho}} = \mc F_{\bar{\rho}}^U$ when $U$ is
constant).\vspace*{1pt} This result can be derived from the large
deviation principle for the sequence of grand-canonical Gibbs measures
$\{\mu_N\}$. On the other hand, it is also a corollary of Theorem
\ref{tqp=f} and Theorem~\ref{tBG} below.

We now consider the weakly asymmetric dynamics with a smooth
orthogonally decomposable external field $E$. Since the total number
of particles is conserved, we have a well defined dynamics on the
hyperplanes $\Omega_{N,K}:= \{\eta\in\Omega_N\dvtx \sum
_{x\in\bb
T_N^d} \eta_x =K \}$, $K=0,\ldots,N^d$. It easy to check that
the generator $L_{E,N}$ is irreducible when restricted to
$\Omega_{N,K}$ so that there exists a unique invariant measure denoted
by $\nu^E_{N,K}$. Fix a sequence $\{K_N \}\subset\bb N$ such that
$N^{-d} K_N \to\bar{\rho} \in[0,1]$ and set
$P_N^E:=\nu_{N,K_N}^E\circ(\pi^N)^{-1}$. As discussed in Section
\ref{dinamocremisi}, if $E=-\nabla U$ the weakly asymmetric Kawasaki
dynamics is reversible with respect to the Gibbs measures on $\bb
T_N^d$ corresponding to the Hamiltonian $H_N^U$ defined in
(\ref{UHamiltonian}). Accordingly, the sequence of probability
measures $\{P_N^E\}$ on $M$ satisfies a large deviation principle with
convex rate function $\mc F^U_{\bar{\rho}}$ as defined in
(\ref{2Fbr}). Also this statement can be obtained as a corollary of
Theorem~\ref{tqp=f} and Theorem~\ref{tBG} below. It remains to
discuss the more interesting case in which either the vector $E$ is
constant or $\sigma$ is a multiple of the identity and $E$ is
orthogonally decomposable with some nontrivial $\wwidetilde E$. We
emphasize that in this case the invariant measures $\nu_{N,K}^E$
cannot be computed explicitly. The following result, which states
that the quasi-potential\vadjust{\goodbreak} $V^E_{\bar{\rho}}$ gives the rate function of
the empirical density when particles are distributed according to
$\nu^E_{N,K_N}$, is proven in~\cite{BG} for the one-dimensional
boundary driven symmetric simple exclusion process. See also~\cite{F}
(where this statement is proven in greater generality) for more
details. The basic argument is analogous to the one for diffusions on
$\bb R^n$ (see~\cite{FW}, Theorem~4.4.3). In view of the dynamical large
deviation principle stated in Theorem~\ref{tdldp} and the uniform
convergence of the hydrodynamic equation (\ref{2he}) proven in
Theorem~\ref{tcfhe*} below, the arguments presented in~\cite{BG,F}
extend to the current setting of nongradient weakly asymmetric
stochastic lattice gases with periodic boundary conditions. We
therefore only state precisely the result.
%
%
\begin{theorem}
\label{tBG}
Fix a vector field $E\in C^1(\bb T^d;\bb R^d)$ satisfying either one
of the conditions in Theorem~\ref{tqp=f} and a sequence
$\{K_N\}\subset\bb N$ such that $N^{-d} K_N \to\bar{\rho} \in
[0,1]$. Then, the sequence of probability measures $\{P_N^E\}$ on
the compact space $M$ satisfies a large deviation principle with
speed $N^d$ and rate function $V^E_{\bar{\rho}}\dvtx M\to[0,+\infty]$
as defined in (\ref{2qp}). Namely, for each closed set $C \subset
M$ and each open set $O\subset M$,
\begin{eqnarray*}
\limsup_{N\to\infty} \frac1{N^d} \log P^E_{N} ( C )
&\leq&- \inf_{\gamma\in C} V^E_{\bar{\rho}} (\gamma),
\\
\liminf_{N\to\infty} \frac1{N^d} \log P^E_{N} ( O )
&\geq&- \inf_{\gamma\in O} V^E_{\bar{\rho}} (\gamma).
\end{eqnarray*}
\end{theorem}

The above result, together with Theorem~\ref{tqp=f}, describes
explicitly the large deviations behavior of the sequence $\{P^E_{N}\}$
in the scaling limit $N\to\infty$. In particular, as discussed before,
it implies that, as far as stationary large deviations of the
empirical density are concerned, weakly asymmetric nongradient
stochastic lattice gases behave as gradient models.

\section{Nongradient tools}
\label{s3}

In this section we collect some technical results which will be used
in the proof both of the hydrodynamic limit and of the dynamical large
deviation principle. We heavily rely on the results in
Vardahan and Yau~\cite{VY}.

\subsection{Additional notation}
\label{pampers}

For the reader's convenience, we fix here some additional notation
needed in the sequel. We first define some (not scaled) generators.
Given a bond $b\in\bb B$ we set $L_{0,b}:= c_b^0 \nabla_b$, $L_{E,b}:=
c_b^E \nabla_b$. Moreover, given $\L\subset\bb Z^d$, we define
$L_{0,\Lambda}:= \sum_{b \in\bb B_\Lambda} L_{0,b}$ and
$L_{E,\Lambda}:= \sum_{b \in\bb B_\Lambda} L_{E,b}$. Recalling
(\ref{2sg}), for $f\dvtx\Omega_N\to\bb R$ we set
\[
L_0 f(\eta):= N^{-2} L_{0,N} f(\eta)
= \sum_{ \{x,y\}\in\bb B_N } c_{x,y}^0 (\eta) \nabla_{x,y}f(\eta).
\]
With some abuse, we also denote by $L_0$ the operator
\[
L_0 f(\eta):= \sum_{x\in\mathbb{Z}^d}\sum_{i=1}^d
c_{x,x+e_i}^0 (\eta) \nabla_{x,x+e} f (\eta)\vadjust{\goodbreak}
\]
acting on local functions $f\dvtx\Omega\rightarrow\bb R$. The meaning
of $L_0$ will be clear from the context. The same definitions hold
for $L_E$ by replacing $c^0_{x,y}$ with $c^E_{x,y}$.

As in~\cite{VY}, given an integer $\ell$ we set $\ell_1 = \ell
-\sqrt{\ell}$ and, given parameters $a_1, a_2, \ldots, a_n$, such that
$a_i \rightarrow\alpha_i$, $i=1,\ldots,n$,
$\limsup_{a_1\rightarrow\alpha_1, a_2 \rightarrow\alpha_2, \ldots , a_n
\rightarrow\alpha_n}$ is a\vspace*{1pt} shorthand for $\limsup
_{a_1\rightarrow\alpha_1} \limsup_{ a_2 \rightarrow \alpha_2}
\cdots\limsup_{ a_n \rightarrow\alpha_n}$. We recall\vspace*{1pt} that
we write ${\Av}_x$ and $\sum_x$ instead of $\Av_{x\in \bbT^d_N}$ and
$\sum_{x \in\bbT^d_N}$, respectively.

Given $\k\in(0,1)$, fix a $C^\infty$ function $\psi^{(\k)}\dvtx \bb R^d
\to[0, \infty)$ such that $ \psi^{(\k)} (r)=0$ if $|r|>1$,
$\psi^{(\k)} (r) = 2^{-d} $ is $|r|<1-\k$ and $\int dr\, \psi
^{(\k)}
(r) =1 $. We write $\psi^{(\k)}_\epsilon$ for the mollifier $\psi
^{(\k)}_\epsilon(r):= {\e}^{-d}\psi^{(\k)} (r/\e)$. Given $\pi
\in
M$, we then define the smooth mollified function $\tilde{\pi}^{\k,
\e}$ as the convolution
%
%
\begin{equation}\label{iacopo}
\tilde{\pi}^{\k, \e}(r):= \pi\ast\psi^{(\k)}_\e(r).
\end{equation}

Finally, we isolate some classes of special functions.
Recall the definition (\ref{misurecanoniche}) of the canonical Gibbs
measure.
As in~\cite{VY} we define the function space $\mathcal{G}$ by
%
%
\begin{eqnarray}
\label{grifondoro}
&&\mathcal{G}:=\bigl\{
f\dvtx\Omega\to\bb R\dvtx f \mbox{ is local
and } \nu^\sigma_{\Delta_f,K} (f)=0\nonumber\\[-8pt]\\[-8pt]
&&\hspace*{89.5pt}
\forall K\in\{0,\ldots,|\D_f|\}, \sigma\in\Omega\bigr\}.\nonumber
\end{eqnarray}
If $f \in\mc G$ then $\nu^\sigma_{\Lambda,K} (f)=0$ for any
$\Lambda\in\bb F$ such that $\Lambda\supset\Delta_f$. It is simple
to check that the current $j^0_{0,e} (\eta)= c^0_{0,e}(\eta)
(\eta_0-\eta_{e})$, where $e$ is an\vspace*{1pt} element of the canonical basis,
belongs to $\mc G$. Moreover, if $g$ is a local function on $\Omega$
then $L_0 g\in\mc G$.

The following class of functions will also play an important role in
the sequel.
%
%
\begin{definition}
\label{buono}
A function $g\equiv g_\rho(\eta)\equiv g(\eta, \rho)
\dvtx \Omega\times[0,1] \to\bb R$ is called \textit{good} iff:
\begin{longlist}[(ii)]
\item[(i)] $g$ is Lipschitz in $\rho$ uniformly with respect to
$\eta$, that is, there exists $C>0$ such that for any
$\rho,\rho'\in[0,1]$ and $\eta\in\Omega$
\[
| g_\rho(\eta)-g_{\rho'}(\eta)| \le C |\rho-\rho'|;
\]
\item[(ii)] $g$ is local in $\eta$ uniformly with respect to $\rho$,
that is, there exists a set $\Delta_0\in\bb F$ such that for any
$\rho\in[0,1]$ we have $\Delta_{g_\rho}\subset\Delta_0$.
\end{longlist}
\end{definition}

Note that good functions are bounded. Working with good functions it
is convenient to introduce the following convention. Given a good
function $g \equiv g(\eta, \rho)$ we will add the superscript 1 both
to generators and to gradients applied to expressions as $g(\tau_y
\eta, \upbar{\eta}_{x, \ell} )$ when these operators act only on the
first entry. For example,
%
%
\begin{equation}
\label{albachiara}
\nabla_{z, z+e} ^1 g ( \tau_y \eta, \upbar{\eta}_{x, \ell}
)
= g ( \tau_y (\eta^{z,z+e}), \upbar{\eta}_{x, \ell} )
- g ( \tau_y \eta, \upbar{\eta}_{x, \ell} ).
\end{equation}
Given a good function $g$ and a function $m\dvtx\Omega\to[0,1]$ we set
%
%
\begin{equation}
\label{pesca}
\underbar{g} (\eta, m(\eta)):=
\sum_{x\in\bb Z^d} g(\tau_x \eta, m(\eta)).
\end{equation}
In words, $\underbar{g} (\eta, m(\eta))$ is obtained by first considering
the formal series $\underbar{g}_\rho$ as defined in
(\ref{2gvec}) and then setting $\rho= m(\eta)$.

\subsection{Spectral estimates}
\label{ocacaterina}

Recall that $\mu_N$ denotes the grand-canonical Gibbs measure on
$\Omega_N$ with zero chemical potential and that $\bb
P^{0,N}_{\mu_N}$ denotes the law of the reversible symmetric
Kawasaki dynamics with initial distribution $\mu_N$. We discuss a
standard method to get super-exponential estimates of the type
%
%
\begin{equation}
\label{porchetta}
\limsup_{k\uparrow\infty, N\uparrow\infty} \frac1{N^d}
\log\bb P^{0,N}_{\mu_N} (B^N_k) = -\infty
\end{equation}
for events of the form $B^N_k=\{ | \int_0^T ds\,
h^N_k(s,\eta(s))|>\delta\}$ for some function $h^N_k$ on
$[0,T]\times\O_N$. Since $e^{|x|}\leq e^x+e^{-x}$ and $\log(a+b)
\leq\log[2(a\lor b)]$, by the exponential Chebyshev inequality and
the Feynman--Kac formula (see, e.g.,~\cite{KL}, Appendix 1, Lemma 7.2)
for each $\gamma>0$ we have
\begin{eqnarray*}
&&
\frac{1}{N^d} \log\bbP^{0,N}_{\mu_N} (B^N_k )\\
&&\qquad\leq-\g\delta
+ \frac{1}{N^d} \log\mathbb{E}^{0,N}_{\mu_N}
\biggl(\exp\biggl\{\biggl|
\int_0^T ds\, \g N^d h^N_k( s, \eta(s) ) \biggr| \biggr\}
\biggr)\\
&&\qquad\leq
-\g\delta+\frac{\log2}{ N^d} + \frac{1}{N^d}
\sup_{\s=\pm1} \log\mathbb{E} ^{0,N}_{\mu_N} \biggl(\exp\biggl\{
\int_0^T ds\, \s\g N^d h^N_k( s, \eta(s) ) \biggr\} \biggr)
\\
&&\qquad\leq
-\g\delta+\frac{\log2 }{N^d}
+\g\sup_{\s=\pm1}
\int_0^T ds \mathop{\sup
\spec}_{L^2(\mu_N) } \{ \s h^N_k( s, \cdot)+ \g^{-1}
N^{2-d}L_0\},
\end{eqnarray*}
where\vspace*{1pt} $\spec_{L^2(\mu_N)}$ denotes the spectrum in $L^2(\mu_N)$.
Hence, in order to get (\ref{porchetta}) it is enough to show that for
each $\g>0$
%
%
\begin{equation}
\label{cocacola}
\limsup_{k\uparrow\infty, N\uparrow\infty}
\int_0^T ds \mathop{\sup
\spec}_{L^2(\mu_N)}
\{ \pm h^N_k( s, \cdot)+ \g^{-1} N^{2-d}L_0\}
\leq0.
\end{equation}

A useful tool to derive the estimate (\ref{cocacola}) is given by the
following perturbative result concerning $\sup\spec_{L^2(\nu)} \{\a V+
\mathfrak{L}\}$, where $\mathfrak{L}$ is an ergodic reversible Markov
generator on a countable set $E$ with invariant measure $\nu$,
$\a\in\bb R$, and $V$ is a function defined on $E$. We refer to
\cite{KL}, Appendix 3, Theorem 1.1, for the proof.
%
%
\begin{lemma}
\label{morbillo}
Let $\operatorname{gap}(\mathfrak{L},\nu)$ be the spectral gap of
$\mathfrak{L}$ in $L^2(\nu)$ and $(\cdot,\cdot)_\nu$ be the
inner product in $L^2(\nu)$. If $\nu(V)=0$ and
$2\a\operatorname{gap} (\mathfrak{L}, \nu)^{-1} \|V\|_\infty
<1$, then
\begin{eqnarray*}
0 &\leq&
\mathop{\sup\spec}_{L^2(\nu)} \{\a V +\mathfrak{L}\}\\
&\leq&\frac{\a^2}{1-2\a\operatorname{gap} (\mathfrak{L},
\nu)^{-1}\|V\|_\infty} ( V, -\mathfrak{L}^{-1} V)_\nu.
\end{eqnarray*}
\end{lemma}

Since the operator $\mathfrak{L}$ is not injective, we need to specify
the meaning of $ ( V, -\mathfrak{L}^{-1} V)_\nu$. By ergodicity, the
kernel of $\mathfrak{L}$ is given by constant functions. In
particular, $f-g$ is a constant function for all $f,g \in
\mathfrak{L}^{-1} (V)$. Since $\nu(V)=0$, we conclude that $(V,f)_\nu$
does not depend on the special function $f \in\mathfrak{L}^{-1}(V)$
and this constant value is the precise meaning of $(V,
-\mathfrak{L}^{-1} V)_\nu$.

\subsection{Central limit theorem variance}
\label{giulioconiglio}

Given a function $f \in\mathcal{G}$, an integer $\ell$ so large that
$\D_f \subset\Lambda_{\ell_1}$ (recall $\ell_1= \ell-\sqrt{\ell}$)
and a canonical measure $\nu$ on $\L_\ell$, we define $V_\ell
(f;\nu)$ as
%
%
\begin{equation}
\label{danza}
V_\ell(f;\nu):=
(2\ell_1+1)^d \Bigl( \Av_{y \in\L_{\ell_1 }} \t_y f,
- L_{0, \L_{\ell} }^{-1} \Av_{y \in\L_{\ell_1}}\tau_y f
\Bigr)_\nu.
\end{equation}
The above $H_{-1}$-seminorm appears from the application of Lemma
\ref{morbillo} to get super-exponential estimates of the form
(\ref{porchetta}) for $h^N_k = \Av_{x} \tau_x f$ (there is no
dependence on $k$).

Given $\Lambda\in\bb F$ let $\mathcal{K}_\L$ be the
$\s$-algebra generated by the random variables $N_\L$ and
$\eta_x$, $x \in\bb Z^d \setminus\L$.
In~\cite{VY}, Section 8, it is proven that for any $\rho\in[0,1]$, the
limit
%
%
\begin{equation}
\label{simpson}
V_\rho(f):=
\mathop{\lim_{\ell\rightarrow\infty}}_{\rho'\rightarrow\rho}
\mu_{\rho'} [
V_\ell( f; \mu_{\rho'} ( \cdot| \mathcal{K}_{\L
_\ell} )
) ]
\end{equation}
exists and is finite. The above limit is called \textit{central limit
theorem variance} and in what follows will be briefly denoted as CLTV.
We recall below some results of~\cite{VY} concerning the CLTV.

On the space $\mathcal{G}$ the functional $V_\rho(\cdot)^{1/2}$ defines
a semi-norm and, by polarization, a pre-inner product $\langle \cdot,
\cdot\rangle _\rho$, that is, $V_\rho(f)= \langle f,f\rangle _\rho$.
The corresponding completion $\mathcal{H}_\rho$ of $\mathcal{G} /
\mathcal{N}_\rho$, where $\mathcal{N}_\rho:=\{ f \in\mathcal{G}\dvtx
V_\rho(f)=0\}$, is therefore an Hilbert space. In what follows, given a
local function $f\in\mathcal{G}$, we will denote again by $f$ the image
of $f$ under the projection plus the inclusion map $ \mathcal{G}
\rightarrow
\mathcal{G}/\mathcal{N}_\rho\hookrightarrow\mathcal{H}_\rho$. In
general, given an element $e$ of the canonical basis, $\nabla_e \eta=
\eta_e - \eta_0$ does not belong to $\mathcal{G}$,
but it is possible to show (\cite{VY}, page~656) that
%
%
\begin{equation}
\label{gelato}
h_{e,s}= \nabla_e \eta- \mu_\rho( \nabla_e \eta|
\mathcal{K}_{\L_s} )
\end{equation}
is a Cauchy sequence in $\mathcal{H}_\rho$ as $s \uparrow\infty$. As
in~\cite{VY}, with some abuse of notation we denote by $\nabla_e \eta$
the limiting point of $h_{e,s}$ in $\mathcal{H}_\rho$.\vadjust{\goodbreak}

We recall a table of computations in the Hilbert space $\mc H_\rho$.
Below $e,e'$ belong to the canonical basis, $j^0_{0,e} (\eta)=
c^0_{0,e}(\eta) (\eta_0-\eta_{e})$ is the current in the direction
$e$ and $g,h$ are generic local functions. Recall the notation
introduced in (\ref{2gvec}) and (\ref{2chi}).
%
%
\begin{eqnarray}
\label{pp1}
\langle j^0_{0,e}, j^0 _{0, e'}\rangle _\rho &=& \tfrac12
\delta_{e,e'}
\mu_\rho[ c_{0,e}^0(\eta) (\eta_e-\eta_0)^2 ],
\\
\label{pp2}
\langle  j^0_{0,e}, L_0 g \rangle _\rho&=& \tfrac{1}{2} \mu_\rho[
c_{0,e}^0(\eta) (\eta_0-\eta_e) \nabla_{0,e}\underbar{g}
],
\\
\label{pp3}
\langle j^0_{0,e}, \nabla_{e'} \eta\rangle _\rho &=& - \delta_{e,e'}
\chi(\rho),
\\
\label{pp4}
\langle \nabla_e \eta, L_0 g \rangle _\rho &=& 0,
\\
\label{pp5}
\langle  L_0 g, L_0 h \rangle _\rho &=& \frac{1}{2}
\sum_{i=1}^d \mu_\rho[ c_{0,e_i} ^0 (\eta)
\nabla_{0,e_i} \underbar g \nabla_{0,e_i}\underbar h ].
\end{eqnarray}
See, respectively, equations (8.7), (8.8), (8.13), (8.14) and the computations
after (8.6) in~\cite{VY}.
We stress that the signs in (\ref{pp2}) and (\ref{pp3}) differ from
the ones in~\cite{VY}. A simple check of the correctness of the above
statement, in the case (\ref{pp3}), is the following. When the
Hamiltonian is zero, the jump rates are constant and $j^0_{0,e}=
c(\eta_0-\eta_e) $, $c>0$. In particular, $\nabla_e\eta$ coincides
in $\mathcal{H}_\rho$ with the standard gradient and it holds
$\langle j^0_{0,e}, \nabla_{e} \eta\rangle _\rho= -c\langle \eta_0-\eta_e, \eta_0-
\eta_e\rangle _\rho$, which must be negative as in (\ref{pp3}).

Define the following linear subspaces of $\mc H_\rho$
\[
\mathcal{G}^{(0)} = \Biggl\{ \sum_{i=1}^d a_i \nabla_{e_i} \eta
, a \in\mathbb{R}^d\Biggr\},\qquad
L_0 \mathcal{G}= \{L_0 g, g \in\mathcal{G}\}.
\]
As follows from~\cite{VY} and the arguments in~\cite{KL}, Chapter 7, the
closure of $\{ L_0 g, g \mbox{ local function}\}$ in
$\mathcal{H}_\rho$ coincides with the closure of $L_0\mathcal{G}$.
Moreover, $\mathcal{H}_\rho$ admits the orthogonal decomposition
%
%
\begin{equation}
\label{orto}
\mathcal{H}_\rho= \mathcal{G}^{(0)} \oplus\overline{ L_0
\mathcal{G}}.
\end{equation}
Observe that orthogonality follows easily from (\ref{pp4}).

Recall the definitions (\ref{2sr}) and (\ref{2D}) of the mobility
$\sigma(\rho)$ and the diffusion coefficient $D(\rho)$. We can give a
simple geometric interpretation of $\s(\rho)$ and $D(\rho)$ referred to
the Hilbert space $\mathcal{H}_\rho$. Indeed, due
to the table of computations (\ref{pp1})--(\ref{pp5}), for each $v \in
\bb R^d$,
%
%
\begin{eqnarray}
\label{sirenetta}
&& \frac{1}{2} \mu_\rho\Biggl[
\sum_{i=1}^d c^0_{0,e_i} ( v_i [\eta_{e_i} - \eta_0]
+ \nabla_{0,e_i} \underbar f )^2 \Biggr]
\nonumber\\
&&\qquad
=V_\rho\Biggl(\sum_{i=1}^d v_i j^0_{0,e_i} \Biggr)
+ 2 \Biggl\langle  \sum_{i=1}^d v_i j^0_{0,e_i},L_0 f\Biggr\rangle _\rho
+ V_\rho( L_0 f)
\\
&&\qquad
= V_\rho\Biggl( \sum_{i=1}^d v_i j^0_{0,e_i} + L_0 f\Biggr).\nonumber
\end{eqnarray}
Let $ P\dvtx \mathcal{H}_\rho\rightarrow\mathcal{G}^{(0)}$ be the
orthogonal projection of $\mathcal{H}_\rho$ onto $\mathcal{G}^{(0)}$.
Then, in view of (\ref{sirenetta}), the variational formula (\ref{2sr})
simply reads
%
%
\begin{equation}
\label{vagabondo}
v\cdot\sigma(\rho) v =
V_\rho\Biggl( P \sum_{i=1}^d v_i j^0_{0,e_i} \Biggr).
\end{equation}
Equivalently,
%
%
\begin{equation}
\label{vagabondo1}\quad
\sigma_{i,k} (\rho)=
\langle  P j^0_{0,e_i}, P j^0_{0,e_k} \rangle _\rho
= \langle  P j^0_{0,e_i}, j^0_{0,e_k} \rangle _\rho,\qquad
i,k=1,\ldots,d.
\end{equation}
By writing $P j^0_{0,e_i}= - \sum_{k=1}^d a_{i,k}(\rho)
\nabla_{e_k}\eta$, from (\ref{pp3}) and (\ref{vagabondo1}) we deduce
$a_{i,k}(\rho) \chi(\rho) = \s_{i,k }(\rho)$.
This implies the key identity
%
%
\begin{equation}
\label{santostefano}
j^0_{0,e_i} = -
\sum_{k=1}^d D_{i,k}(\rho) \nabla_{e_k} \eta
+ (\bb I -P) j^0_{0,e_i}\qquad\mbox{in } \mathcal{H}_\rho.
\end{equation}
In the next lemma we give some additional characterization of the
entries of $\s(\rho)$, which will be used below. We omit the proof,
which easily follows from (\ref{pp1}) and (\ref{vagabondo1}).\vspace*{-2pt}
%
%
\begin{lemma}
\label{tporchetta}
For each $\rho\in[0,1]$ and $i,k=1,\ldots,d$ it holds
%
%
\begin{eqnarray}\qquad
\label{serio}
\sigma_{i,i}(\rho) &=&
\langle j^0_{0,e_i}, j^0_{0,e_i} \rangle _\rho-
\langle j^0_{0,e_i}, (\bb{I}-P) j^0_{0,e_i}\rangle _\rho,
\\[-2pt]
\label{vagabondo2}
\sigma_{i,k}(\rho)
&=& -\langle  j^0_{0,e_i}, (\bb I -P) j^0_{0,e_k} \rangle
_\rho\nonumber\\[-9pt]\\[-9pt]
&=& -\langle  (\bb{I}-P) j^0_{0,e_i}, j^0_{0, e_k}\rangle _\rho,\qquad
i \not= k.\nonumber\vspace*{-2pt}
\end{eqnarray}
\end{lemma}

By definition of $P$, for each $\rho\in[0,1]$ and $i=1,\ldots,d$
there exist local functions $g_\rho^{(i)}$ such that $-L_0
g^{(i)}_\rho$ approximates $(\bbI- P)j^0_{0,e_i}$ in
$\mathcal{H}_\rho$. Moreover, it is possible to choose the family of
approximating functions in such a way that some regularity is
achieved. More precisely, recalling Definition~\ref{buono},
(\ref{santostefano}) and~\cite{VY}, Corrolary 3.5, imply the following
statement.\vspace*{-2pt}
%
%
\begin{lemma}
\label{salsiccia}
For each $i=1,\ldots,d$ and $\delta>0$ there exists a good
function $g^{(i)}_\rho(\eta)\dvtx[0,1]\times\Omega\to\bb R$
such that, setting
\[
\phi^{(i)}_\rho:= j^0_{0,e_i}+ \sum_{k=1}^d D_{i,k}(\rho)
\nabla_{e_k}
\eta+L_0 g^{(i)}_\rho= (\bbI- P) j^0_{0,e_i} +L_0 g^{(i)}_\rho,
\]
we have
%
%
\begin{equation}
\label{salsicciab}
\sup_{\rho\in[0,1]}
V_\rho\bigl( \phi^{(i)}_\rho\bigr) \leq\delta.\vspace*{-2pt}
\end{equation}
\end{lemma}

\subsection{Super-exponential estimates}
\label{s31}

We introduce some perturbations of the weakly asymmetric
dynamics.\vadjust{\goodbreak}
Given $\ell\ge1$, $H\in C^{1,2}([0,T]\times\bb T ^d)$
and a family of good functions $\mathbf{g}= \{ g^{(i)} (\eta,
\rho), i=1,\ldots,d\}$ we define the functions
$F\equiv F^N_{H,\ell,\mathbf{g}}$ and $\upbar{F} \equiv{\upbar{F}}^N_{H,
\ell,
\mathbf{g}}$ on $[0,T]\times\Omega_N$ by
%
%
\begin{eqnarray}
\label{alfa}
F(t,\eta) &:=& \frac{1}{2}\sum_{x} H_t\biggl( \frac xN\biggr) \eta_x
+ \upbar{F}(t, \eta),
\\
\label{beta}
\upbar F(t, \eta) &:=&
\frac{1}{2N} \sum_{x} \sum_{i=1}^d
\nabla^N_i H_t\biggl(\frac xN\biggr) g^{(i)}(\tau_x \eta, \upbar{\eta
}_{x,\ell}),
\end{eqnarray}
where the discrete gradient $\nabla^N_i$ is defined by $\nabla^N_i
f(r):= N[ f(r+ e_i/N)-f(r) ]$, $r\in\bb T^d$. We then
consider the time-inhomogeneous Markov chain on $\Omega_N$ with jump
rates $N^2 c_{x,y}^{E,H,\mathbf{g}}$, where $c_{x,y}^{E,H,\mathbf{g}}$
is defined at time $t$ by
%
%
\begin{eqnarray}
\label{delfinodino}
c_{x,y}^{E,H, \mathbf{g} }(\eta) :\!&=&
c_{x,y} ^E(\eta) \exp\{F(t, \eta^{x,y} )- F(t,\eta)\}
\nonumber\\[-8pt]\\[-8pt]
&=& c_{x,y}^{E+\nabla H_t} (\eta)
\exp\{ \upbar F(t, \eta^{x,y} )- \upbar F(t,\eta)\}\nonumber
\end{eqnarray}
in which the rate $c_{x,y}^{E+\nabla H_t}$ is defined as in
(\ref{tassi}) with the field $E$ replaced by $E+ \nabla H_t$.
We let $L^{E,H,\mathbf{g}}_{t,N}$ be the corresponding
time-inhomogeneous generator and denote by
$\bb P_{\eta^N}^{E, H, \mathbf{g},N}$ the law of the perturbed
chain with initial condition $\eta^N$.
We convey to write simply $ \bb P_{\eta^N}^{E,H,N}$ and
$L^{E,H}_{t,N}$ if $\mathbf{g}=0$.
Note that\vspace*{1pt} in this case, in view of the last identity in
(\ref{delfinodino}), the above dynamics coincides with the weakly
asymmetric Kawasaki dynamics with time-inhomogeneous external field
$E+\nabla H_t$.

We observe that there exists a constant $C>0$ depending
only on $H$ and the functions $g^{(i)}$ such that for any
$\{x,y\}\in\bb B_N$ it holds
%
%
\begin{eqnarray}
\label{enrico}
\sup_{0\leq t \leq T}\sup_{\eta\in\Omega_N}|\nabla_{x,y}
F(t,\eta)| &\leq&\frac{C }{N},\nonumber\\[-8pt]\\[-8pt] \sup_{0\leq t \leq
T}\sup_{\eta\in\Omega_N}|\nabla_{x,y} \bar F(t,\eta)| &\leq&
\frac{C }{N }.\nonumber
\end{eqnarray}

%
\begin{lemma}
\label{trn}
Fix $E\in C^1 (\bb T^d;\bb R^d)$, $H\in C^{1,2}([0,T]\times\bbT^d)$,
$\ell\ge1$, a family of good functions $\mathbf{g}$ and let
$\bb P^{E,H,\mathbf{g},N}_{\eta^N}$ as defined above.
For each $p\in[1,\infty)$ there exists a constant $C_0$
such that for any $N\ge1$, $T>0$, and any sequence $\{\eta^N\in
\Omega
_N\}$
\[
\limsup_{N\to\infty} \frac1{N^d}
\log\bb E^{0,N}_{\mu_N} \biggl( \biggl[
\frac{d\bb P^{E,H, \mathbf{g}, N}_{\eta^N}}{d\bb P^{0,N}_{\mu_N}}
\biggr]^p \biggr) \le C_0 (T +1).
\]
\end{lemma}
\begin{pf}
By the assumptions on the interaction (see Definition~\ref{tint}),
there exists a constant $C$ depending only on $\Phi$ such that for
any $\eta^N\in\Omega_N$ we have $\log\mu_N(\eta^N) \ge-C N^d$. It
is therefore enough to prove the lemma with $\bb P^{0,N}_{\mu_N}$
replaced by $\bb P^{0,N}_{\eta^N}$.

Given an ordered bond $(x,y) \in\wwidetilde{\bb B}_N$, $t\in
[0,T]$ and $\eta\in D([0,T];\Omega_N)$, denote by
$\mc N_{x,y}^\eta(t)$ the total number of particles that in the time
interval $[0,t]$ jumped from $x$ to $y$. Set also $J_{x,y}^\eta(t)
:= \mc N_{x,y}^\eta(t) - \mc N_{y,x}^\eta(t)$. By standard tools in
the theory of jump Markov processes (see, e.g.,~\cite{Br}, Section VI.2)
we can compute the Radon--Nikodym derivative as
\begin{eqnarray*}
&&\frac{d\bb P^{E,H,\mathbf{g},N}_{\eta^N}}{ d\bb
P^{0,N}_{\eta^N}} (\eta) \\
&&\qquad= \exp\biggl\{
\sum_{\{x,y\}\in{\bb B}_N}
\biggl[ E_N(x,y)J_{x,y}^\eta(T) +F(T,\eta_T)-F(0,\eta_0)
\\
&&\hspace*{61pt}\qquad\quad{} - N^2 \int_0^T dt\, c^0_{x,y} (\eta(t))\\
&&\qquad\hspace*{115.6pt}{}\times\bigl( e^{E_N(x,y) [\eta_x(t)- \eta_y(t)]
+\nabla_{x,y} F(t,\eta(t))} -1 \bigr)\biggr] \biggr\}.
\end{eqnarray*}
Note indeed that $ E_N(x,y)J_{x,y}^\eta(T)$ and $E_N(x,y)
[\eta_x(t)- \eta_y(t)]$ do not depend on the orientation of the bond
$(x,y)$; therefore they can be thought of, as in the above expression,
as functions of the unoriented bond $\{x,y\}$. The previous
expression yields
\begin{eqnarray*}
\hspace*{-4pt}&&\biggl[\frac{d \bb P^{E,H,\mathbf{g},N}_{\eta^N}}{d\bb P^{0,N}_{\eta
^N}}(\eta)
\biggr]^p\\
\hspace*{-4pt}&&\qquad= \frac{d\bb P^{p E, p H, \mathbf{g}, N}_{\eta^N}}
{d\bb P^{0,N}_{\eta^N}}(\eta)
\\
\hspace*{-4pt}&&\qquad\quad{} \times
\exp\biggl\{ N^2 \int_0^T dt\, \sum_{\{x,y\}\in\bb B_N}
c^0_{x,y} (\eta(t))\\
\hspace*{-4pt}&&\hspace*{112pt}\qquad\quad{}\times\bigl[ e^{p E_N(x,y) [\eta_x(t)- \eta_y(t)]+p\nabla_{x,y}F(t,\eta(t))}
- 1\\
\hspace*{-4pt}&&\hspace*{127pt}\qquad\quad{}
- p \bigl( e^{ E_N(x,y) [\eta_x(t) -
\eta_y(t)]+\nabla_{x,y}F(t,\eta(t))} - 1 \bigr) \bigr]
\biggr\}.
\end{eqnarray*}
By using (\ref{enrico}) and the bound $| E_N(x,y)| \le C
N^{-1}$ for some $C>0$ [see (\ref{2EN=})] we get that there exists a
constant $C'=C'(E,H,\mathbf{g},p)>0$ such that
\[
\bigl[ e^{p E_N(x,y) [\eta_x- \eta_y]+p\nabla_{x,y}F(t,\eta)} - 1
- p \bigl( e^{ E_N(x,y) [\eta_x - \eta_y]+\nabla_{x,y}F(t,\eta)} - 1
\bigr)\bigr]
\le\frac{C'} {N^2}\cdot
\]
The lemma follows readily.\vadjust{\goodbreak}
\end{pf}

The following simple consequence of the previous lemma will be
repeatedly used to deduce super-exponential estimates from those
obtained in~\cite{VY}.

%
\begin{remark}
\label{turn}
Consider a sequence of events $\{ B^N_k\}$ in
$D([0,T];\Omega_N)$ which have super-exponentially small
probability with respect to the stationary process $\bb P^{0,N}_{\mu_N}$,
that is, such that
\[
\limsup_{k\uparrow\infty, N\uparrow\infty}
\frac1{N^d} \log\bb P^{0,N}_{\mu_N} (B^N_k) = -\infty.
\]
In view of Lemma~\ref{trn}, an application of the H\"older
inequality shows that the previous estimate holds also for the
probability $\bb P^{E,H, \mathbf{g},N}_{\eta^N}$.
\end{remark}

As is well known, key points in the proof of the hydrodynamic limit are
the so-called one and two block estimates. By standard methods (see,
e.g.,~\cite{KL}, Chapter~10) one can prove the one block estimate at a
super-exponential level. The basic statement is given in the following
lemma; in the sequel we also use, without further mention, slight
variations of this result.
%
%
\begin{lemma}[(One block estimate)]
\label{t1b}
For each $\varphi\in C([0,T]\times\bb T^d)$, each local function
$h$ on $\Omega$ and each $\zeta>0$ it holds
\[
\limsup_{\ell\uparrow\infty, N\uparrow\infty}
\frac{1}{N^d} \log
\bb P^{0,N}_{\mu_N} \biggl( \biggl|
\int_0^T dt
\Av_x \varphi_t\biggl(\frac xN\biggr)
[ h (\tau_x \eta) - \mu_{\upbar{\eta}_{x, \ell}} (h) ]
\biggr|> \zeta\biggr) = -\infty.
\]
\end{lemma}

As explained in~\cite{VY}, as a byproduct of the spectral estimates
in Section~\ref{ocacaterina} and~\cite{VY}, Theorem 6.2, the two
blocks estimate holds in super-exponential sense with respect to $\bb
P_{\mu_N}^{0,N}$.
%
%
\begin{lemma}[(Two blocks estimate)]
\label{duedue}
For each local function $h$ on $\Omega$ and each $\zeta>0$, it holds
\begin{eqnarray*}
&&\limsup_{\ell\uparrow\infty, a\downarrow0, N\uparrow\infty}
\frac{1}{N^d} \log
\bb P^{0,N}_{\mu_N} \biggl(
{\int_0^T dt \Av_x
\Av_{y\dvtx|y-x|\leq aN}}
| h(\upbar{\eta}_{x,\ell}(t))
- h(\upbar{\eta}_{y,\ell})(t) |> \zeta\biggr)
\\
&&\qquad
= -\infty,\\
&&\limsup_{\ell\uparrow\infty, a\downarrow0, N\uparrow\infty}
\frac{1}{N^d} \log
\bb P^{0,N}_{\mu_N} \biggl(
{\int_0^T dt \Av_x} | h(\upbar{\eta} _{x,aN }(t))
- h(\upbar{\eta}_{x, \ell})(t) |>\zeta\biggr) = -\infty.
\end{eqnarray*}
\end{lemma}

As in~\cite{VY}, Theorem 3.9, given $c>0$, $i=1,\ldots,d$ and a site
$x$, we define the density gradient in the direction $e_i$ as
%
%
\begin{equation}
\label{marina}
\Psi^{(i)} _{x,N,c}(\eta):=
\frac{ \eta_{x+cNe_i} -\eta_{x -c N e_i}}{2 c N}.
\end{equation}

In Proposition~\ref{cile} below we collect super-exponential bounds
for suitable events. Such events appear naturally\vadjust{\goodbreak} in the proof of
the hydrodynamic limit and the dynamical large deviation principle.
To introduce these events, we first fix some notation: in the
following definitions $\varphi\equiv\varphi_t (r)$ and $H\equiv
H_t(r)$ are functions in $C^{1,2}( [0,T]\times\mathbb{T}^d)$, while
$\mathbf{g}$ and $\hat{\mathbf{g}}$ are families of good functions
$\mathbf{g}= \{ g^{(i)} (\eta, \rho), i=1,\ldots,d\}$,
$\hat{\mathbf{g}}= \{ \hat{g}^{(i)} (\eta, \rho),
i=1,\ldots,d\}$ [the function $\upbar{F}^N_{H,\ell,\mathbf{g}}$ has
been defined in (\ref{beta})] and recalling the notation
(\ref{iacopo}) for the smooth convolution we shorthand
${\widetilde{\pi^N(\eta)}}{}^{\k,\epsilon}$ with
$\tilde{\pi}^N(\eta)^{\k,\epsilon}$. In addition, we set
\begin{eqnarray*}
T_1 (t, \eta)
:\!&=& N \Av_x \sum_{i=1}^d \nabla^N_i
\varphi_t\biggl(\frac xN\biggr) j^0_{x,x+e_i},
\\
T_2 (t, \eta) :\!&=& N \Av_x \sum_{i=1}^d \nabla^N_i
\varphi_t\biggl(\frac xN\biggr) \Av_{y\dvtx |y-x|\leq\ell_1} j^0_{y,y+e_i},
\\
T_3(t, \eta) :\!&=& \frac{1}{2} \Av_x \sum_{i=1}^d \nabla^N_i
\varphi_t \biggl(\frac{x}{N}\biggr) c^0_{x,x+e_i}(\eta) (\eta_x-
\eta_{x+e_i})^2 E_i({x}/{N}),
\\
T_{4 } (t, \eta) :\!&=& \frac{1}{2} \Av_x \sum_{i=1}^d \nabla^N_i
\varphi_t\biggl(\frac{x}{N}\biggr) c^0_{x,x+e_i}(\eta) (\eta_x-
\eta_{x+e_i})^2 \partial_i H_t({x}/{N}),
\\[-0.3pt]
T_{5, \mathbf{g} } (t, \eta)
:\!&=& N \Av_x \sum_{i=1}^d
\nabla^N_i \varphi_t\biggl(\frac{x}{N}\biggr) c^0_{x,x+e_i}(\eta) (\eta_x-
\eta_{x+e_i})
\nabla_{x,x+e_i} \upbar{F}^N_{H,\ell,\mathbf{g}}(t,\eta)
\\[-0.3pt]
&=& \frac{1}{2} \Av_x\sum_z \sum_{i=1}^d \sum_{j=1}^d
\nabla^N_i \varphi_t \biggl(\frac xN\biggr) \nabla^N_j H_t \biggl(\frac{z}{N}\biggr)
c_{x,x+e_i}^0(\eta) (\eta_x- \eta_{x+e_i})
\\[-0.3pt]
&&\hspace*{69pt}{} \times
\nabla_{x,x+e_i} g^{(j)} (\tau_z \eta, \upbar{\eta}_{z, \ell}),
\\[-0.3pt]
T_{6, \mathbf{g}} (t, \eta)
:\!&=& \frac{1}{2} \Av_x \sum_z
\sum_{i=1}^d \sum_{j=1}^d \nabla^N_i \varphi_t\biggl(\frac{x}{N}\biggr)
\nabla^N_j H_t \biggl(\frac{z}{N}\biggr)
\\[-0.3pt]
&&\hspace*{69pt}{}
\times c_{x,x+e_i}^0(\eta) (\eta_x- \eta_{x+e_i})
\nabla^1_{x,x+e_i} g^{(j)} ( \t_z\eta, \upbar{\eta}_{z, \ell} ),
\\[-0.3pt]
T_{7, \mathbf{ g}, \hat{\mathbf{g}} } (t, \eta)
:\!&=& N \Av_x
\sum_{i=1}^d \nabla^N_i \varphi_t \biggl(\frac{x}{N}\biggr) \Av
_{y\dvtx|y-x|\leq
\ell_1 } L^{E, H, \mathbf{g}, 1}_{t,N} \hat{g}^{(i)} ( \tau_y
\eta, {\upbar{\eta}}_{x,\ell} ),
\\[-0.3pt]
T_{8, \mathbf{g},\hat{\mathbf{g}}} (t,\eta)
:\!&=&
N \Av_x \sum_{i=1}^d
\nabla^N_i \varphi_t \biggl(\frac{x}{N}\biggr)
\Av_{y\dvtx|y-x|\leq\ell_1} L^{E, H, \mathbf{g}}_{t,N}
\hat{g}^{(i)} ( \tau_y \eta,{\upbar{\eta}}_{x,\ell} ),
\\[-0.3pt]
T_{9, \mathbf{g}} (t,\eta) :\!&=& N \Av_x \sum_{i=1}^d
\nabla^N_i \varphi_t\biggl(\frac{x}{N}\biggr)
\Av_{y\dvtx|y-x|\leq\ell_1} L_0
g^{(i)} ( \tau_y \eta, {\upbar{\eta}}_{y,\ell}),
\\[-0.3pt]
T_{10, \mathbf{g}} (t, \eta) :\!&=& N \Av_x \sum_{i=1}^d
\nabla^N_i \varphi_t\biggl(\frac{x}{N}\biggr)
\Av_{y\dvtx|y-x|\leq\ell_1}
L_0 g^{(i)} ( \tau_y\eta, {\upbar{\eta}}_{x,\ell})
, \\[-0.3pt]
T_{11, \mathbf{g}} (t,\eta) :\!&=& N \Av_x \sum_{i=1}^d
\nabla^N_i \varphi_t\biggl(\frac{x}{N}\biggr)
\Av_{y\dvtx|y-x|\leq\ell_1 }
L_0^1 g^{(i)} ( \tau_y \eta, {\upbar{\eta}}_{x,\ell} )
,\\[-0.3pt]
T_{12}(t,\eta) :\!&=& N \Av_x \sum_{i=1}^d \sum_{j=1}^d
\nabla^N_i \varphi_t\biggl(\frac{x}{N}\biggr)
D_{i,j}( \upbar{\eta}_{x, a N})
\Av_{y\dvtx|y-x|\leq\ell_1} \Psi^{(j)}_{y,N,c} (\eta)
\\[-0.3pt]
&=& \Av_x \sum_{i=1}^d \sum_{j=1}^d
\nabla^N_i \varphi_t \biggl(\frac{x}{N}\biggr)
D_{i,j}(\upbar{\eta}_{x, a N})
\frac{\upbar{\eta}_{x+cN e_j,\ell_1}-\upbar{\eta}_{x-c N
e_j,\ell_1}}
{2c},
\\[-0.3pt]
T_{13, \mathbf{g}} (t,\eta) :\!&=& \frac12 \Av_x\sum_z \sum_{i=1}^d
\sum_{j=1}^d
\nabla_i^N\varphi_t\biggl(\frac{x}{N}\biggr) c^0_{z,z+e_j}(\eta)
(\eta_z - \eta_{z+e_j})[E_j+ \partial_j H_t]
\\[-0.3pt]
&&\hspace*{69pt}{}
\times\biggl(\frac{z}{N}\biggr)
\Av_{y\dvtx|y-x|\leq\ell_1} \nabla_{z,z+e_j}^1
g^{(i)}(\tau_y \eta, \upbar{\eta}_{x,\ell} )
,\\
T_{14,\mathbf{g},\hat{\mathbf{g}}}(t,\eta) :\!&=& \frac{1}{2}
\Av_x\sum_z \sum_v \sum_{i=1}^d \sum_{j=1}^d \sum_{k=1}^d
\nabla_i^N \varphi_t\biggl(\frac{x}{N}\biggr)
\nabla^N_k H_t\biggl(\frac{v}{N}\biggr)
c^0 _{z,z+e_j}(\eta)
\\
&&\hspace*{101pt}{}
\times
\nabla_{z, z+e_j} g^{(k)}(\tau_v\eta,\upbar{\eta}_{v,\ell
})\\
&&\hspace*{101pt}{}
\times
\Av_{y\dvtx|y-x|\leq\ell_1 }
\nabla^1_{z,z+e_j} \hat{g}^{(i)}
( \tau_y \eta,\upbar{\eta}_{x,\ell})
,\\
T_{15,\mathbf{g},\hat{\mathbf{g}}}(t, \eta) :\!&=& \frac{1}{2}
\Av_x \sum_z \sum_v\sum_{i=1}^d \sum_{j=1}^d \sum_{k=1}^d
\nabla_i^N\varphi_t\biggl(\frac{x}{N}\biggr) \nabla^N_k H_t\biggl(\frac{v}{N}\biggr)
c^0 _{z,z+e_j}(\eta)
\\
&&\hspace*{101pt}{}
\times\nabla^1_{z,z+e_j} g^{(k)}( \t_v \eta,\upbar{\eta
}_{v,\ell
})\\
&&\hspace*{101pt}{}
\times
\Av_{y\dvtx|y-x|\leq\ell_1 } \nabla^1 _{z,z+e_j}
\hat{g}^{(i)} ( \tau_y \eta,\upbar{\eta}_{x,\ell})
,\\
T_{16} (t, \eta) :\!&=& \sum_{i=1}^d \sum_{j=1}^d
\int_{\bb T^d} dr\,
\partial_i \varphi_t (r)
D_{i,j}( \tilde{\pi}^N(\eta)^{\k, a} (r) )\,
\partial_j \tilde{\pi}^N(\eta)^{\k', \epsilon} (r).
\end{eqnarray*}
Moreover, recalling (\ref{pesca}) and introducing $\z$ as variable
of integration on $\O$, we also define
\begin{eqnarray*}
K_1 (t, \eta) &:=& \frac{1}{2} \Av_x \sum_{i=1}^d
\partial_i\varphi_t \biggl(\frac{x}{N}\biggr)
[E_i+\partial_i H_t]\biggl(\frac{x}{N}\biggr)
\mu_{\upbar{\eta}_{ x,\ell} }[
c^0_{0,e_i}(\z) (\z_0 -\z_{e_i})^2 ],
\\
K_{2, \mathbf{g}} (t, \eta)&:=&\frac{1}{2} \Av_x\sum_{i=1}^d \sum_{j=1}^d
\partial_i \varphi_t\biggl(\frac{x}{N}\biggr)
[E_j+ \partial_j H_t] \biggl(\frac{x}{N}\biggr)
\\
&&\hspace*{53pt}{}
\times
\mu_{\upbar{\eta}_{x, \ell} }
\bigl[c^0_{0,e_j}(\z) (\z_0 - \z_{e_j})
\nabla^1_{0,e_j} \underbar{g}^{(i)}( \z, \upbar{\eta}_{x, \ell} )
\bigr],
\\
K_{3, \mathbf{g},\hat{\mathbf{g}} } (t, \eta) &:=& \frac{1}{2}
\Av_x \sum_{i=1}^d\sum_{j=1}^d \sum_{k=1}^d
\partial_i \varphi_t \biggl(\frac{x}{N}\biggr)\,
\partial_k H_t\biggl(\frac{x}{N}\biggr)
\\
&&\hspace*{69pt}{}
\times\mu_{\upbar{\eta}_{x,\ell}} \bigl[ c^0_{0,e_j}(\z)
\nabla^1_{0,e_j} \underbar{g}^{(k)}(\z,\upbar{\eta}_{x,\ell})\\
&&\hspace*{131pt}{}\times
\nabla^1_{0,e_j}\underbar{\hat{g}}^{(i)}(\z,
\upbar{\eta}_{x,\ell})\bigr]
, \\
K_{4,\mathbf{g}}(t, \eta) &:=& \frac{1}{2}\Av_x \sum_{i=1}^d \sum_{j=1}^d
\partial_i \varphi_t \biggl( \frac{x}{N}\biggr)\,
\partial_j H_t \biggl(\frac{x}{N}\biggr)
\\
&&\hspace*{53pt}{}
\times
\mu_{\upbar{\eta}_{x, \ell}} \bigl[ c_{0,e_i}^0(\z)
(\z_0- \z_{e_i}) \nabla^1 _{0,e_i}
\underbar g^{(j)} (\z, \upbar{\eta}_{x, \ell})\bigr],
\\
K_5(t, \eta)&:=& \Av_x \sum_{i=1}^d \sum_{k=1}^d
\partial_i \varphi_t \biggl(\frac{x}{N}\biggr)
\sigma_{i,k} (\upbar{\eta}_{x, \ell})
[E_k+\partial_k H_t] \biggl(\frac{x}{N}\biggr).
\end{eqnarray*}

In the above definitions, instead of a generic family of good
functions, we will sometimes take the family of good functions
provided by Lemma~\ref{salsiccia}, which we denote by $\mathbf{g}
[\delta]$.
In this case, we will add the dependence on $\delta$ in the notation.
For instance, $T_{5, \mathbf{g}[\delta]}(t, \eta)$ denotes the
function $T_{5,\mathbf{g}} (t,\eta)$ when the family $\mathbf{g}$
is chosen so that the bound (\ref{salsicciab}) holds.
%
%
\begin{proposition}
\label{cile}
Let $\varphi,H\in C^{1,2}( [0,T]\times\mathbb{T}^d)$, and let
$\mathbf{g}$, $\hat{\mathbf{g}}$ be families of good functions.
Then for each $\zeta>0$ the expressions $T_1,\ldots, T_{16}$,
$K_1,\ldots, K_5$ defined above satisfy the following
super-exponential estimates:
%
%
\begin{eqnarray}
\label{grazia1}
\limsup_{N \uparrow\infty}
\frac{1}{N^d} \log\bb P^{0,N}_{\mu_N}\biggl(
\biggl| \int_0^T dt\,
[T_1- T_2] (t,\eta(t)) \biggr|>\zeta\biggr)
&=& -\infty,\\
\label{grazia3}
\limsup_{\delta\downarrow0, \ell\uparrow\infty, a\downarrow0,
c\downarrow0, N \uparrow\infty}
\frac{1}{N^d}\hspace*{131.6pt}\qquad&& \nonumber\\[-8pt]\\[-8pt]
{}\times\log
\bb P^{0,N}_{\mu_N}\biggl( \biggl| \int_0^T dt
\,\bigl[ T_2 + T_{11,\mathbf{g}[\delta]} + T_{12} \bigr]
(t, \eta(t)) \biggr|>\zeta\biggr)
&=& -\infty,\nonumber\\
\label{grazia4}
\limsup_{\ell\uparrow\infty, N\uparrow\infty}
\frac{1}{N^d} \log
\bb P^{0,N}_{\mu_N} \biggl( \biggl|
\int_0^T dt\,
[T_3 +T_4- K_1 ] (t, \eta(t)) \biggr|> \z\biggr)&=&-\infty,
\\
\label{grazia5}
\limsup_{\ell\uparrow\infty, N\uparrow\infty}
\frac{1}{N^d} \log\bb P^{0,N}_{\mu_N} \biggl( \biggl|
\int_0^T dt\,
[T_{6,\mathbf{g}} - K_{4,\mathbf{g}} ]
(t,\eta(t))\biggr|> \zeta\biggr)&=& -\infty,
\\
\label{bella1}
\limsup_{\ell\uparrow\infty, N\uparrow\infty}
\frac{1}{N^d} \log\bb P^{0,N}_{\mu_N} \biggl( \biggl|
\int_0^T dt\, [ T_{13, \mathbf{g}}- K_{2, \mathbf{g}} ]
(t,\eta(t))\biggr|>\zeta\biggr)&=& -\infty,
\\[-3pt]
\label{bella2}
\limsup_{\ell\uparrow\infty, N\uparrow\infty}
\frac{1}{N^d} \log\bb P^{0,N}_{\mu_N} \biggl( \biggl|
\int_0^T dt\,
[T_{15,\mathbf{g},\hat{\mathbf g}}- K_{3, \mathbf{g},\hat
{\mathbf
g}}]
(t,\eta(t))\biggr|>\zeta\biggr)&=& -\infty,
\\[-3pt]
\label{annarella1}
\limsup_{\ell\uparrow\infty, N \uparrow\infty}
\frac{1}{N^d} \log\bb P^{0,N}_{\mu_N}\biggl( \biggl|
\int_0^T dt\, [T_{5, \mathbf{g} }- T_{6, \mathbf{g}}]
(t, \eta(t)) \biggr|> \zeta\biggr)&=& -\infty,
\\[-3pt]
\label{annarella2}
\limsup_{\ell\uparrow\infty, N \uparrow\infty}
\frac{1}{N^d} \log\bb P^{0,N}_{\mu_N} \biggl( \biggl|
\int_0^T dt\, [
T_{7, \mathbf{ g}, \hat{\mathbf{g}}}- T_{8, \mathbf{ g},\hat
{\mathbf{g}}}
] (t, \eta(t)) \biggr|>\zeta\biggr) &=& -\infty,
\\[-3pt]
\label{annarella3}
\limsup_{\ell\uparrow\infty, N \uparrow\infty}
\frac{1}{N^d} \log\bb P^{0,N}_{\mu_N} \biggl( \biggl|
\int_0^T dt\, [T_{9, \mathbf{g} }- T_{10, \mathbf{g}} ]
(t, \eta(t)) \biggr| > \zeta\biggr)&=& -\infty,
\\[-3pt]
\label{annarella4}
\limsup_{\ell\uparrow\infty, N \uparrow\infty}
\frac{1}{N^d} \log\bb P^{0,N}_{\mu_N}\biggl( \biggl|
\int_0^T dt\, [T_{10, \mathbf{g}}- T_{11, \mathbf{g}} ]
(t, \eta(t)) \biggr| > \zeta\biggr)&=& -\infty,
\\[-3pt]
\label{gaetano}\hspace*{35pt}
\limsup_{\ell\uparrow\infty, N \uparrow\infty}
\frac{1}{N^d} \log\bb P^{0,N}_{\mu_N}\biggl(\biggl|
\int_0^T dt\, [T_{14, \mathbf{g}, \hat{\mathbf{g}}}- T_{15,
\mathbf{g},
\hat{\mathbf{g}}} ] (t, \eta(t))\biggr|> \zeta\biggr)
&=&-\infty,\vspace*{-16pt}
\end{eqnarray}
\begin{eqnarray}
\label{barbalalla1}
&&\limsup_{\delta\downarrow0, \ell\uparrow\infty, N\uparrow
\infty}
\frac{1}{N^d} \log\bb P^{0,N}_{\mu_N}\biggl(
\biggl| \int_0^T dt\, \bigl[K_{3, \mathbf{g}, \mathbf{g}[\delta
]} +K_{4,
\mathbf{g}} \bigr] (t, \eta(t))\biggr| >\z\biggr)\nonumber\\[-9pt]\\[-9pt]
&&\qquad=-\infty,\nonumber
\\[-3pt]
\label{barbalalla2}\qquad
&&\limsup_{\delta\downarrow0, \ell\uparrow\infty, N\uparrow
\infty}
\frac{1}{N^d} \log\bb P^{0,N}_{\mu_N} \biggl( \biggl|
\int_0^T dt\, \bigl[K_1+K_{2, \mathbf{g}[\delta]} -K_5 \bigr]
(t, \eta(t)) \biggr| >\z\biggr)\nonumber\\[-9pt]\\[-9pt]
&&\qquad= -\infty,\nonumber
\\[-3pt]
\label{miracolo}
&&
\limsup_{\k\downarrow0, \ell\uparrow\infty, a
\downarrow0, \k' \downarrow0, \e\downarrow0, c
\downarrow0, N \uparrow\infty}
\frac{1}{N^d} \nonumber\\[-3pt]
&&\quad{}\times\log\bb P^{0,N}_{\mu_N} \biggl( \biggl|
\int_0^T dt\, [T_{12}- T_{16}]
(t,\eta(t)) \biggr| > \zeta\biggr)\\[-3pt]
&&\qquad= -\infty.\nonumber
\end{eqnarray}
\end{proposition}
\begin{pf} We prove the stated super-exponential bounds one after
the other. We denote by $C$ a generic constant, independent of the
parameters we are taking the limit, whose numerical value
can change from line to line.\vspace*{8pt}

\textit{The estimate} (\ref{grazia1}).
Summing by parts we get
\[
T_1(t, \eta)- T_2 (t, \eta)= N \Av_x \sum_{i=1}^d
j^0_{x,x+e_i}
\Av_{y\dvtx |y-x|\leq\ell_1} \biggl[
\nabla^N_i\varphi_t\biggl(\frac{x}{N}\biggr)
-\nabla^N_i\varphi_t\biggl(\frac{y}{N}\biggr)\biggr].
\]
The term inside the square brackets, after taking average, gives a
contribution of the order $\ell^2/N^2$. Hence, $T_1-T_2$ is of the
order $\ell^2/N$.\vadjust{\goodbreak}

\textit{The estimate} (\ref{grazia3}).
This is the core of~\cite{VY} and follows from
\cite{VY}, Theorem~3.9, the arguments presented in Section
\ref{ocacaterina} and the definition of $\mathbf{g}[\delta]$ (look
also at~\cite{VY}, Step 3, page 637).\vspace*{8pt}

\textit{The estimate} (\ref{grazia4}).
It is an immediate consequence of the one block estimate.\vspace*{8pt}

\textit{The estimate} (\ref{grazia5}).
Let us define $T_{6,\mathbf{g} }^{(1)} (t, \eta)$ as the expression
obtained from $T_{6,\mathbf{g}}$ by replacing the term $g^{(j)}(\t_z
\eta,\upbar{\eta}_{z,\ell})$ with $g^{(j)}(\t_z
\eta,\upbar{\eta}_{x, \ell} )$. We observe that, due to the
definitions of good functions and of the gradient $\nabla^1$, both
in $T_{6,\mathbf{g}}$ and $T_{6,\mathbf{g}}^{(1)}$ we can\vspace*{1pt} restrict
the sum over $z$ to the sites $z$ such that $|x-z|\leq C$. In view
of the Lipschitz property of good functions, we thus have
\[
\bigl| T_{6,\mathbf{g} } (t, \eta)- T_{6,\mathbf{g}}^{(1)} (t, \eta)
\bigr|
\leq\frac{C}{N^d} \sum_x \sum_{z\dvtx |z-x|\leq C} | \bar{\eta}_{x,
\ell} - \bar{\eta}_{z, \ell} | \leq\frac{C}{ \ell} \cdot
\]
Using again the above sum restriction and due to the smoothness of
$H$, in $T_{6,\mathbf{g}}^{(1)}$ we can afterward replace $\nabla^N
_j H_t
(z/N)$ with $\nabla^N _j H_t (x/N)$ with an error $O(1/N)$.
Finally, we can remove the sum restriction over $z$. At the end we
get
\begin{eqnarray*}
T_{6,\mathbf{g}} (t, \eta) & = & \frac{1}{2} \Av_x\sum_{i=1}^d \sum_{j=1}^d
\nabla^N_i \varphi_t \biggl( \frac{x}{N} \biggr)
\nabla^N _j H_t \biggl(\frac{x}{N}\biggr) c_{x,x+e_i}^0(\eta)
(\eta_x- \eta_{x+e_i})
\\
&&\hspace*{53pt}{}
\times\nabla^1_{x,x+e_i} \sum_z g^{(j)}(\t_z\eta,\upbar{\eta
}_{x,\ell})
+ O\biggl(\frac1N\biggr)+O\biggl( \frac1\ell\biggr)
\end{eqnarray*}
and (\ref{grazia5}) follows from the one block estimate.\vspace*{8pt}

\textit{The estimate} (\ref{bella1}).
Recalling the definition of $\nabla^1$, we observe again that we can
restrict the sum over $z$ to the sum over $z\dvtx|z-y|\leq C$. As a
consequence, $|z-x|\leq C+\ell_1$. Hence, by an error of order
$O(\ell/N)$, we can replace $\nabla^N_i \varphi_t (x/N)$ with
$\nabla^N_i \varphi_t (z/N)$. We call $ T_{13,\mathbf{g}}^{(1)}$
the resulting expression.
Let us now define $T_{13, \mathbf{g}}^{(2)}$ as
$T_{13,\mathbf{g}}^{(1)}$ with $g^{(i)} (\tau_y \eta, \bar
\eta_{x,\ell} )$ replaced by $g^{(i)} (\tau_y \eta,
\bar
\eta_{x,a N} )$. By the Lipschitz property
of good functions, we can estimate
\[
\bigl| T_{13,\mathbf{g} }^{(1)}-T_{13, \mathbf{g} }^{(2)}\bigr|
(t, \eta) \leq {C \Av_x}
| \upbar{\eta}_{x, \ell}- \upbar{\eta}_{x, a N } |.
\]
By the two blocks estimate (see Lemma~\ref{duedue}), we conclude that
%
%
\begin{eqnarray}
\label{investigatore}
&&\limsup_{\ell\uparrow\infty, a \downarrow0, N \uparrow
\infty}
\frac{1}{N^d} \ln\bbP_{\mu_N} ^{0,N} \biggl( \biggl|\int_0^T dt\,
\bigl[T^{(1)}_{13,\mathbf{g} }-T^{(2)}_{13, \mathbf{g} }\bigr](t,
\eta(t))\,dt \biggr|> \z\biggr)
\nonumber\\[-8pt]\\[-8pt]
&&\qquad
=-\infty.\nonumber
\end{eqnarray}
We next\vspace*{2pt} define $T^{(3)}_{13, \mathbf{g}}$ as
$T^{(2)}_{13,\mathbf{g}}$ with $g^{(i)} (\tau_y \eta, \bar
\eta_{x,a N} )$ replaced by $g^{(i)} (\tau_y \eta,
\bar\eta_{z, a N } )$. Since $|x-z| \leq C + \ell_1$, by the
Lipschitz property of good functions we get
\[
\bigl|T^{(2)}_{13,\mathbf{g} }-T^{(3)}_{13, \mathbf{g} } \bigr|(t, \eta)
\leq {C \Av_x }| \upbar{\eta}_{x, a N }- \upbar{\eta}_{z, a N }
| \leq C \frac{C+\ell}{a N}\cdot
\]
At this point we define $T^{(4)}_{13,\mathbf{g}}$ as
$T^{(3)}_{13,\mathbf{g} }$ with the term $g^{(i)} (\tau_y \eta
, \bar\eta_{z,a N } )$ replaced by $g^{(i)} (\tau_y
\eta, \bar\eta_{z, \ell} )$.
As in\vspace*{2pt} (\ref{investigatore}), we obtain that the event $\{
|\int_0^T [T^{(3)}_{13, \mathbf{g}} -T^{(4)}_{13, \mathbf{g}}]
(t$, $\eta_t)\,dt |> \z\}$ has\vspace*{1pt} super-exponentially small
probability.
In order to prove (\ref{bella1}) we can therefore replace
$T_{13, \mathbf{g} }$ with $T^{(4)}_{13, \mathbf{g}} $,
\begin{eqnarray*}
T^{(4)}_{13, \mathbf{g}} (t,\eta)
&:=& \frac12 \Av_z \sum_{i=1}^d \sum_{j=1}^d
\nabla_i^N \varphi_t \biggl(\frac{z}{N}\biggr)
[E_j+ \partial_jH_t]\biggl(\frac{z}{N}\biggr)
c^0_{z,z+e_j}(\eta) (\eta_z - \eta_{z+e_j})
\\
&&\hspace*{53pt}{}
\times\nabla_{z,z+e_j}^1
\sum_y g^{(i)}(\tau_y \eta, \bar\eta_{z,\ell} ).
\end{eqnarray*}
The thesis now follows from the one block estimate.\vspace*{8pt}

\textit{The estimate} (\ref{bella2}).
The proof of this bound follows by the same ideas used in the
proof of (\ref{bella1}), apart the fact that now there are more
indexes. Anyway, in $T_{15, \mathbf{g}, \hat{\mathbf g}} $ one can
sum over $z \in\bbT^d_N$, $y\dvtx |y-z| \leq C$, $x\dvtx |x-y| \leq\ell_1$
and $v\dvtx |v-z| \leq C+\ell$. Then one has to use the two blocks estimate
and, at the end, the one block estimate.\vspace*{8pt}

\textit{The estimate} (\ref{annarella1}).
If
%
%
\begin{equation}
\label{elezionilazio}
\nabla_{x,x+e_i} g^{(j)} ( \t_z \eta, \bar{\eta} _{z, \ell})
\not=
\nabla^1 _{x,x+e_i} g^{(j)} (\t_z\eta,\bar{\eta}_{z,\ell} ),
\end{equation}
then the bond $\{x, x+e_i\}$ must intersect both $\L_{z, \ell}$ and
its complement. In particular, given $z$ the number of sites $x$
leading to the inequality (\ref{elezionilazio}) are of order $O(
\ell^{d-1})$. In addition, since $g^{(j)} (\eta, \rho)$ is Lipschitz
in $\rho$ uniformly in $\eta$, setting $\omega= \eta^{x,x+e_i}$ with
$\{x,x+e_i\}$ intersecting both $\L_{z,\ell}$ and its complement,
we get
\begin{eqnarray*}
&& \bigl| \nabla_{x,x+e_i} g^{(j)} ( \t_z \eta, \bar{\eta} _{z,
\ell} ) - \nabla^1 _{x,x+e_i} g^{(j)} ( \t_z \eta, \bar{\eta}
_{z, \ell} )\bigr|
\\
&&\qquad
= \bigl| g^{(j)} ( \t_z \omega, \bar{\omega}_{z, \ell}
)- g^{(j)} ( \t_z \omega, \bar{\eta}_{z, \ell} ) \bigr|
\leq C |\bar{\omega}_{z, \ell}-\bar{\eta}_{z,\ell}|
\leq C \frac1{\ell^d}.
\end{eqnarray*}
The above observations imply that
$|T_{5,\mathbf{g}}-T_{6,\mathbf{g}}| \leq C / \ell$,
which trivially implies (\ref{annarella1}).\vspace*{8pt}

\textit{The estimate} (\ref{annarella2}).
We define
\begin{eqnarray*}
T_{7, \hat{\mathbf{g} } } ^{(1)} (t, \eta)&:=& N
\Av_x\sum_{i=1}^d
\nabla^N_i \varphi_t \biggl(\frac{x}{N}\biggr)
\Av_{y\dvtx|y-x|\leq\ell_1 } L^1 _0
\hat{g}^{(i)} ( \tau_y \eta, {\bar\eta}_{x,\ell}),
\\
T^{(1)}_{8, \hat{\mathbf{g} } } (t, \eta)&:=& N
\Av_x \sum_{i=1}^d
\nabla^N_i \varphi_t \biggl(\frac{x}{N}\biggr)
\Av_{y\dvtx|y-x|\leq\ell_1 } L_0 \hat{g}^{(i)} ( \tau_y \eta,
{\bar\eta}_{x,\ell} ).
\end{eqnarray*}
By Taylor expansion of the perturbed jump rates [see
(\ref{taylor-rate}) below together with (\ref{enrico})], we can write
$T_{7,\mathbf{g},\hat{\mathbf{g}}} = T^{(1)}_{7,\hat{\mathbf{g} } }
+ V$ and $T_{8,\mathbf{g}, \hat{\mathbf{g} } } =
T^{(1)}_{8,\hat{\mathbf{g} } } + W$, where $V $ and $W $ are
uniformly bounded functions of $t, \eta$. One can then prove that
$\| V - W\|_\infty\leq C/ \ell$ by the same arguments used in the
proof of
(\ref{annarella1}). Finally,\vspace*{-1pt} the event $\{ | T^{(1)}_{7,\hat{\mathbf
{g} } }-
T^{(1)}_{8,\hat{\mathbf{g} } } | >\z\}$ has super-exponentially
small probability as proved in~\cite{VY}, between Lem\-ma~3.8 and
Theorem 3.9 there.\vspace*{8pt}

\textit{The estimate} (\ref{annarella3}).
In view of (\ref{cocacola}), we only need to prove that for each $\g>0$
%
%
\begin{eqnarray}\label{cantano}\quad
&&
\limsup_{\ell\uparrow\infty, N \uparrow\infty}
\sup_{t\in[0,T]} \mathop{\sup\spec}_{L^2 (\mu_N)} \{ \pm
(T_{9,\mathbf{g}
}-T_{10,\mathbf{g} } )(t, \eta)+ \g^{-1} N^{2-d} L_0 \}
\nonumber\\[-8pt]\\[-8pt]
&&\qquad
\leq0.\nonumber
\end{eqnarray}
We point out three facts. (i) It holds
$N^{2-d} L_0 \leq c(d) \Av_x N^2 \ell^{-d} L_{0, \L_{x, 10 \ell} }$
in the operator sense.
(ii) Since for self-adjoint operators $W$ the quantity
$\sup\spec_{L^2(\mu_N)} \{W\}$ equals the supremum of
$(f, Wf)_{\mu_N}$ among the functions $f \in L^2 (\mu_N)$ satisfying
$(f,f)_{\mu_N}=1$, the map $W\to\sup\spec_{L^2(\mu_N)} \{W\}$ is
subadditive. (iii) Both in $T_{9,\mathbf{g} }$ and
$T_{10,\mathbf{g}}$ we can replace $L_0$ with
$L_{0, \L_{x,10 \ell}}$ if $\ell$ large.
Combining (i), (ii) and (iii) we deduce
%
%
\begin{eqnarray}
\label{ballano}
&&
\mathop{\sup\spec}_{L^2 (\mu_N)}
\{ \pm(T_{9,\mathbf{g} }-T_{10,\mathbf{g} } )
+ \g^{-1} N^{2-d} L_0 \}
\nonumber\\[-2pt]
&&\qquad
\leq C \Av_x \sup_{\nu} \mathop{\sup\spec}_{L^2 (\nu)}
\biggl\{ \pm\nabla^N_i \varphi_t\biggl(\frac xN\biggr) N L_{0, \L_{x, 10
\ell} }
R( \t_x\eta)
\\[-2pt]
&&\hspace*{102pt}\qquad\quad{}
+ c(d)\g^{-1} N^2 \ell^{-d} L_{0, \L_{x, 10 \ell}}
\biggr\},\nonumber
\end{eqnarray}
where $\nu$ varies among all canonical Gibbs measures on $\L
_{x,10\ell}$
and
\[
R(\eta):= \Av_{y\dvtx|y|\leq
\ell_1} \bigl[ g^{(i)} (\t_y \eta, \bar{\eta}_{y, \ell}
)-g^{(i) }
(\t_y \eta, \bar{\eta}_{ \ell} ) \bigr].
\]

By the uniform strong mixing assumption on interaction, there exists
a constant $C>0$ such that $\operatorname{gap}( L_{0,\L_{x,10\ell} }
) \geq
\ell^{-2}/C$ (see~\cite{CM,LY,Y}). Applying Lemma~\ref{morbillo} with
$\mathfrak{L}= c(d) \g^{-1} N^2 \ell^{-d} L_{0,\L_{x,10\ell}} $,
using translation invariance and the expression of the Dirichlet
form for reversible processes, we can then bound the right-hand side
of (\ref{ballano}) by
%
%
\begin{eqnarray}
\label{urlano}
&&C \ell^d \sup_{\nu}
( R, -L_{0, \L_{10\ell} } R )_\nu
\nonumber\\[-9pt]\\[-9pt]
&&\qquad
\leq
C \sup_\nu\Biggl\{ \Av_{x \in\L_{10\ell} } \sum_{j=1}^d
\bb I_{\{ \{x,x+e_j\} \subset\L_{10 \ell} \}}
\nu[ (\nabla_{x,x+e_j} \ell^d R
)^2]\Biggr\},\nonumber
\end{eqnarray}
where $\bb I$ denotes the indicator function. By the same
arguments used in the proof of (\ref{annarella1}), in (\ref{urlano})
we can replace $\nabla_{x,x+e_j}$ by $\nabla_{x,x+e_j}^1$ with an
error of order $O( \ell^{-1})$. On the other hand, due to the
definition of good function, there exists a constant $K >0$ such
that $g^{(i)}(\cdot, \rho)$ has support in $\L_K$ for all $ \rho\in
[0,1]$. We take $\ell\gg K$. Then, using the Lipschitz property of
good functions, we can bound the right-hand side of (\ref{urlano})
by
\begin{eqnarray*}
&& C \Av_{x \in\L_{10\ell} } \sum_{j=1}^d
\bb I_{\{\{x,x+e_j\} \subset\L_{10 \ell}\}}
\\[-2pt]
&&\hspace*{40pt}\quad{} \times
\nu\biggl[\biggl(\nabla^1_{x,x+e_j}
\sum_{y\in\L_{\ell_1}\dvtx |y-x|\leq2K }
\bigl[ g^{(i)}( \t_y \eta, \bar{\eta}_{y,\ell})-
g^{(i)} (\t_y \eta, \bar{\eta}_\ell) \bigr] \biggr)^2\biggr]\\[-2pt]
&&\quad{}+ O\biggl(\frac1\ell\biggr)
\\[-2pt]
&&\qquad \leq
C \Av_{x \in\L_{10\ell} }\nu\biggl[\biggl(\sum_{y\in
\L_{\ell_1}\dvtx |y-x|\leq2K } | \bar{\eta}_{y, \ell}-
\bar{\eta}_\ell| \biggr)^2\biggr]
+ O\biggl(\frac1\ell\biggr).
\end{eqnarray*}
The proof is now concluded observing that the last bound above
vanishes uniformly in $\nu$ as $\ell\to\infty$ by the equivalence
of ensembles.\vspace*{8pt}

\textit{The estimates} (\ref{annarella4}) \textit{and} (\ref{gaetano}).
The proof is similar to the proof of
(\ref{annarella1}).\vspace*{8pt}

\textit{The estimate} (\ref{barbalalla1}).
Due to (\ref{pp5}) and (\ref{pp2}) we can write
\begin{eqnarray*}
K_{3,\mathbf{g}, \mathbf{g}[\delta]} (t, \eta)
&=& \Av_x \sum_{i=1}^d \sum_{k=1}^d
\partial_i \varphi_t\biggl(\frac xN\biggr)
\,\partial_k H_t \biggl(\frac xN\biggr)
\bigl\langle  L_0 {g}^{(k)} _\rho, L_0 g^{(i)}_\rho[\delta] \bigr\rangle _{\rho=\bar
{\eta
}_{x,\ell}},
\\[-2pt]
K_{4, \mathbf{g} } (t, \eta) &=& \Av_x \sum_{i=1}^d \sum_{j=1}^d
\partial_i \varphi_t \biggl( \frac{x}{N}\biggr)
\,\partial_j H_t \biggl(\frac{x}{N}\biggr)
\bigl\langle j_{0,e_i} ^0, L_0 g^{(j)} _\rho\bigr\rangle _{\rho= \bar{\eta}_{x,\ell}
}\cdot
\end{eqnarray*}
Hence,
\begin{eqnarray*}
&&\bigl[ K_{3,\mathbf{g}, \mathbf{g}[\delta]} + K_{4, \mathbf
{g}}\bigr](t,
\eta)
\\
&&\qquad =
\Av_x \sum_{i=1}^d \sum_{k=1}^d
\partial_i \varphi_t\biggl(\frac xN\biggr) \,\partial_k H_t\biggl(\frac xN\biggr)
\bigl\langle L_0 {g}^{(k)} _\rho,j_{0,e_i} ^0
+ L_0 g^{(i)}_\rho[\delta]\bigr\rangle _{\rho= \bar{\eta}_{x,\ell} }.
\end{eqnarray*}
Due to Lemma~\ref{salsiccia}, the orthogonal decomposition
(\ref{orto}) and the definition of the orthogonal projection $P$ we
can write for all $\rho\in[0,1]$
\[
\bigl\langle L_0 {g}^{(k)} _\rho, j_{0,e_i} ^0 + L_0 g^{(i)}[ \delta]\bigr\rangle  _\rho= \bigl\langle L_0
{g}^{(k)} _\rho, P j_{0,e_i} ^0 \bigr\rangle +o(1)=o(1),
\]
where the error term $o(1)$ goes to zero uniformly in $\rho\in
[0,1]$ as $\delta$ goes to zero. The thesis follows.

\textit{The estimate} (\ref{barbalalla2}).
Using (\ref{pp1}), (\ref{pp2}) and Lemma~\ref{salsiccia} we can
write
\begin{eqnarray*}
K_1 (t, \eta) &=& \Av_x \sum_{i=1}^d
\partial_i \varphi_t \biggl(\frac{x}{N}\biggr)
[E_i+\partial_i H_t] \biggl(\frac{x}{N}\biggr)
\langle  j^0 _{0,e_i},j^0_{0, e_i} \rangle  _{\rho=\bar{\eta}_{x, \ell} },\\
K_{2, \mathbf{g}[\delta]} (t,\eta) &=& \Av_x \sum_{i=1}^d\sum_{j=1}^d
\partial_i \varphi_t \biggl(\frac{x}{N}\biggr)
[E_j+ \partial_j H_t]\biggl(\frac{x}{N}\biggr)
\bigl\langle  j^0 _{0,e_j }, L_0 g^{(i) }_\rho[\delta] \bigr\rangle  _{\rho=\bar{\eta}_{x,
\ell}}
\\
&=& -\Av_x \sum_{i=1}^d \sum_{j=1}^d
\partial_i \varphi_t \biggl(\frac{x}{N}\biggr)
[E_j+ \partial_j H_t]\biggl(\frac{x}{N}\biggr)
\langle  j^0 _{0,e_j }, (\bbI-P) j^0_{0,e_i}\rangle _{\rho=\bar{\eta}_{x,
\ell}}\\
&&{} + o(1).
\end{eqnarray*}
We apply Lemma~\ref{tporchetta} in order to rewrite the
above terms $K_1,K_{2, \mathbf{g}[\delta]}$ in terms of the matrix
$\s$.
By (\ref{serio}) and (\ref{vagabondo2}), respectively, we can write
\begin{eqnarray*}
K_1 (t, \eta) &=& \Av_x \sum_{i=1}^d
\partial_i \varphi_t \biggl(\frac{x}{N}\biggr)
[E_i+\partial_i H_t] \biggl(\frac{x}{N}\biggr)
\s_{i,i}(\bar{\eta}_{x, \ell} )+ E(t, \eta)
,\\
K_{2, \mathbf{g}[\delta]} (t, \eta) &=& \Av_x \sum_{i=1}^d
\mathop{\sum_{ j\dvtx1\leq j \leq d}}_{ j \not=i }
\partial_i \varphi_t \biggl(\frac{x}{N}\biggr)
[E_j+ \partial_j H_t] \biggl(\frac{x}{N}\biggr)
\s_{i,j} (\bar{\eta}_{x, \ell} )
\\
&&{}
-E(t, \eta) + o(1),
\end{eqnarray*}
where $E(t, \eta):=\Av_x \sum_{i=1}^d
\partial_i \varphi_t (\frac{x}{N})
[E_i+\partial_i H_t](\frac{x}{N} )
\langle  j^0 _{0,e_i}, [\bbI-P]j^0_{0, e_i} \rangle _{\rho= \bar{\eta}_{x,
\ell
} }$.
Comparing with $K_5(t,\eta)$, the above identities trivially imply
the thesis.\vspace*{8pt}

\textit{The estimate} (\ref{miracolo}).
Given $x \in\bb T^d_N$ and $s>0$, denote by $\mc K_{x,s}$ the
$\s$-algebra generated by the observables $\eta_y$, $y \in\bb T^d_N
\setminus\L_{x,s}$, and by $\bar{\eta}_{x,s}$. In the proof of
Theorem 3.9 in~\cite{VY}, page 649, it is shown that in $T_{12}$ one
can replace $D_{i,j}( \upbar{\eta}_{x, a N})
\Av_{y\dvtx|y-x|\leq\ell_1} \Psi^{(j)}_{y,N,c} (\eta)$ with $
D_{i,j}( \upbar{\eta}_{x, \ell}) \Av_{y\dvtx|y-x|\leq\ell_1}
(\eta_{y+e_j}-\eta_y)$. We call $T_{12}'$ the resulting expression.
As the proof is based on the two blocks estimate and
\cite{VY}, Theorem 5.3, it needs the property that the function
$D_{i,j}( \bar{\eta}_{x, a N })$ is $\mc K_{x, A N}$ for some $A$
(take $A= a$). The same property holds indeed also for
$D_{i,j}(\tilde{\pi}^N(\eta)^{\k,a} (x/N))$ with $A= (1-\k)a $. In
view of Assumption~\ref{tregmob}, it holds
\[
\bigl| D_{i,j }\bigl(\tilde{\pi}^N(\eta)^{\k,a} (x/N)\bigr)
- D_{i,j}(\bar{\eta}_{x, a N })\bigr| \leq C\k,
\]
which allows us to apply the two blocks estimate as in
\cite{VY}, page 650. As a consequence, the expression
$T^{(1)}_{12}$, obtained from
$T_{12}$ by replacing
$D_{i,j}( \bar{\eta}_{x, a N })$ with
$D_{i,j}(\tilde{\pi}^N(\eta)^{\k, a} (x/N))$,
is equivalent to
$T_{12}'$ and therefore to $T_{12}$,
\[
\limsup_{\k\downarrow0, \ell\uparrow\infty, a
\downarrow0, c \downarrow0, N \uparrow\infty}
\frac{1}{N^d}\log\bb P^{0,N}_{\mu_N}\biggl( \biggl|
\int_0^T dt\, \bigl[T_{12} - T_{12}^{(1)}\bigr]
(t, \eta(t)) \biggr| >\z\biggr)= -\infty.
\]

By replacing $\nabla^N_i\varphi_t$ with $\partial_i\varphi_t$ and
summing by parts, we can write
\begin{eqnarray*}
&&
T^{(1)} _{12} (t,\eta)\\
&&\qquad=
- \Av_x \sum_{i=1}^d \sum_{j=1}^d \upbar{\eta}_{x, \ell_1}
\frac1{2c} \biggl[
\partial_i\varphi_t \biggl(\frac{x}{N}+c e_j\biggr)
D_{i,j}\biggl( \tilde{\pi}^N(\eta)^{\k, a} \biggl(\frac xN +c e_j\biggr)
\biggr)
\\
&&\qquad\quad\hspace*{95pt}{}
-\partial_i\varphi_t \biggl(\frac{x}{N}- c e_j\biggr)
D_{i,j}\biggl(\tilde{\pi}^N(\eta)^{\k, a} \biggl(\frac xN -c e_j\biggr) \biggr)
\biggr]\\
&&\qquad\quad{}+ o(1).
\end{eqnarray*}
Observe that $\tilde{\pi}^N(\eta)^{\k, a}$ belongs to
$C^\infty(\bb T^d)$. Moreover, fixed $a,\k$, we can bound its derivatives
by a constant depending only on $a,\k$. Hence, by Taylor expansion,
\begin{eqnarray*}
&& \biggl| \frac1{2c} \biggl[
\partial_i\varphi_t \biggl(\frac{x}{N}+c e_j\biggr)
D_{i,j}\biggl( \tilde{\pi}^N(\eta)^{\k, a} \biggl(\frac xN +c e_j\biggr)\biggr)
\\
&&\qquad{}
-\partial_i\varphi_t \biggl(\frac{x}{N}- c e_j\biggr)
D_{i,j}\biggl(\tilde{\pi}^N(\eta)^{\k, a} \biggl(\frac xN -c e_j\biggr) \biggr)
\biggr]
\\
&&\hspace*{73.1pt}\qquad{}
- \partial_j [
\partial_i\varphi_t D_{i,j}( \tilde{\pi}^N(\eta)^{\k, a}
)
] \biggl(\frac xN\biggr) \biggr|
\le C c,
\end{eqnarray*}
where $C=C(\k,a)$ and $c$ is the scale parameter. Up to now we have
proved that $T_{12}$ is equivalent, in the super-exponential sense
stated in (\ref{miracolo}), to $T_{12}^{(1)}$, which, by the above
observations, is equivalent to
\[
T^{(2)} _{12}(t, \eta) = - \Av_x \sum_{i=1}^d
\sum_{j=1}^d \upbar{\eta}_{x,\ell_1}\,
\partial_j [
\partial_i\varphi_t D_{i,j}( \tilde{\pi}^N(\eta)^{\k, a}
)] \biggl(\frac xN\biggr).
\]
Note that the scale parameter $c$ does not appear anymore. By the
same argument used in the proof of equation (\ref{grazia1}), in
$T^{(2)}_{12}$ we can replace the local density
$\upbar{\eta}_{x,\ell_1}$ with $\eta_x$ paying an error bounded by
$C(a,\k) (\ell/N)^2$, and therefore negligible. We call the new
expression $T^{(3)}_{12}$.
By the same argument used to derive (\ref{grazia1}) we can replace
$\eta_x$ by $\upbar{\eta}_{x, \e N}$ with an error bounded by
$C(a,\k) \e^2$, therefore negligible. By this replacement we get
$T^{(4)}_{12}$. Since
$|\bar\eta_{x,\e N}-\tilde{\pi}^N(\eta)^{\k',\epsilon
}(x/N)|
\leq C( \k'+1/N\e) $
and the limits $N\uparrow\infty,\k'\downarrow0$ and $\e\downarrow
0$ are taken
before the limit $a\downarrow0$, by a uniform estimate we can
replace $\upbar{\eta}_{x,\e N}$ with $\tilde{\pi}^N(\eta)^{\k
',\epsilon}(x/N)$
getting
\[
T^{(5)}_{12} (t, \eta):= - \Av_x \sum_{i=1}^d
\sum_{j=1}^d
\tilde{\pi}^N(\eta)^{\k',\epsilon}\biggl(\frac xN\biggr)
\,\partial_j[ \partial_i\varphi_t
D_{i,j}(\tilde{\pi}^N(\eta)^{\k,a}) ]
\biggl(\frac xN\biggr).
\]
With an error negligible as $N\uparrow\infty$, in $T^{(5)}_{12}$ we
can replace the average $\Av_x$ with the integral over $\bb T^d$. By
an integration by parts, the resulting expression is indeed~$T_{16}$.
\end{pf}

\section{Hydrodynamic limit}
\label{shl}

In this section we prove the hydrodynamic scaling limit for the weakly
asymmetric Kawasaki dynamics. In order to prove the dynamical large
deviation principle, we need a more general version of
Theorem~\ref{thl} that is stated below. Recall that $\bb
P_{\eta^N}^{E, H, \mathbf{g},N}$ is the law of the process with the
perturbed rates defined in (\ref{delfinodino}) and observe that by
setting $H=0$ and $\mathbf{g}=0$ we recover the law $\bb
P_{\eta^N}^{E, N}$ of the original weakly asymmetric Kawasaki dynamics
as defined in (\ref{2LEN}).
%
%
\begin{theorem}
\label{dentini}
Fix $T>0$, functions $E\in C^1(\bb T^d; \bb R^d)$, $H\in
C^{1,2}([0,T]\times\bb T ^d)$, a profile $\gamma\in M$, a
sequence $\{\eta^N \in\Omega_N\}$ associated to $\gamma$ and a
family $\mathbf{g}=\{ g^{(i)}\dvtx 1\leq i \leq d\}$ of good
functions. The sequence of probability measures
$\{\bb P^{E,H,\mathbf{g},N}_{\eta^N} \circ(\pi^N)^{-1} \}
_{N\ge1}$
on $\mc M_{[0,T]}$ converges weakly to $\delta_{u}$,
where $u$ is the unique element of $\mc M_{[0,T]}$ satisfying the two
following conditions.
\begin{longlist}[(ii)]
\item[(i)] \textup{Energy estimate.} The weak gradient of $u$
is in $L^2([0,T]\times\bb T^d, dt\,dr;\bb R^d)$,
%
%
\begin{equation}
\label{4ee}
\int_0^T dt\, \langle\nabla u_t, \nabla u_t\rangle<
+ \infty.
\end{equation}
\item[(ii)] \textup{Hydrodynamic equation.} The function $u$ is a weak
solution to
%
%
\begin{eqnarray}\label{4he}\quad
\partial_t u + \nabla\cdot[ \sigma(u) ( E+\nabla H_t)
] &=&
\nabla\cdot[ D(u) \nabla u],\qquad
(t,r) \in(0,T)\times\bb T^d,
\nonumber\\[-8pt]\\[-8pt]
u_0(r)&=&\gamma(r),\qquad r \in\bb T^d.\nonumber
\end{eqnarray}
\end{longlist}
\end{theorem}

To prove this result, we shall first discuss the tightness of
the sequence
$\{\bb P^{E,H,\mathbf{g},N}_{\eta^N} \circ(\pi^N)^{-1}
\}_{N\ge1}$ and prove the energy estimate. Since these results
are also relevant for the large deviation principle, they will be
proven at the super-exponential level. We then discuss a
microscopic characterization of the hydrodynamic equation and
conclude the proof of the hydrodynamic limit.

\subsection*{Exponential tightness}
Recall that a sequence of probability measures $\{P_n\}$ on a Polish
space $\mc X$ is \textit{exponentially tight} iff there exists a
sequence $\{\mc K_\ell\}$ of compact subsets of $\mc X$ such that
%
%
\begin{equation}
\label{5exptigh}
\limsup_{\ell\uparrow\infty, n\uparrow\infty}
\frac1n \log P_n(\mc K_\ell^\complement) = -\infty.
\end{equation}
%
%
\begin{lemma}
\label{tmodcont}
Under the same hypotheses of Theorem~\ref{dentini}, for each
$\varphi\in C^{2}(\bb T^d)$ and each $\zeta>0$, it
holds
%
%
\begin{eqnarray}
\label{campanile1}
&&\limsup_{\tau\downarrow0, N\uparrow\infty}
\frac1{N^d} \log\bb P^{E,H, \mathbf{g}, N}_{\eta^N}
\Bigl(
\sup_{{s,t\in[0,T]\dvtx |s-t|\le\t}}
| \langle\pi^N_t,\varphi\rangle-\langle\pi^N_s,\varphi
\rangle
|
>\zeta\Bigr)
\nonumber\\[-8pt]\\[-8pt]
&&\qquad
=-\infty.\nonumber
\end{eqnarray}
\end{lemma}
\begin{pf}
The bound (\ref{campanile1}) is proven in~\cite{VY}, Section 4, for the
reversible process $\bb P^{0,N}_{\mu^N}$. Therefore, by
Remark~\ref{turn}, it holds also for $\bb P^{E,H,\mathbf{g},N}_{\eta^N}$.
\end{pf}

Since $M$ is compact, by definition of the weak* topology on $M$ and
standard characterizations of compacts in the Skorohod space,
the above lemma implies that the sequence $\{\bb
P^{E,H,\mathbf{g},N}_{\eta^N} \circ(\pi^N)^{-1}\}_{N\ge1}$ is
exponentially tight. We also observe that, since in
(\ref{campanile1}) we used the modulus of continuity on the set of
continuous path and not the one in the Skorohod space,
Lemma~\ref{tmodcont} also implies that any limit point of the
sequence $\{\bb P^{E,H, \mathbf{g},N}_{\eta^N}\}_{N\ge1}$ is
supported by $C([0,T];M)$.

\subsection*{Energy estimate}
Let $\mc Q\dvtx \mc M_{[0,T]} \to[0,+\infty]$ be the functional defined
by
%
%
\begin{equation}
\label{2QT}
\mc Q (\pi):= \sup\{ \mc Q_{F} (\pi),
F\in C^{1} ([0,T]\times\bb T^d; \bb R^d ) \},
\end{equation}
where, given $F\in C^{1}([0,T]\times\bb T^d;\bb R^d)$,
%
%
\begin{equation}
\label{2QTF}
\mc Q_{F} (\pi):= - 2 \int_0^T dt\, \langle\pi_t,
\nabla\cdot F_t\rangle- \int_0^T dt\, \langle F_t, F_t
\rangle.
\end{equation}
Observe that $\mc Q$ is convex and lower semicontinuous.
Moreover, by a standard argument,
$\mc Q(\p)=\sup_F ( \int_0^T dt\, \langle\p_t,F_t
\rangle)^2/\int_0^T dt\, \langle F_t, F_t\rangle$. Hence,
Riesz representation theorem implies that $\mc Q(\pi) < +\infty$ iff
the weak gradient of $\pi$ belongs to $L_2([0,T]\times\bb T^d,
dt \,dr;\bb R^d)$. If this is the case, we also have $\mc Q (\pi) =
\int_0^T dt\, \langle\nabla\pi_t,\nabla\pi_t\rangle$. In view of
Remark~\ref{turn}, the energy estimate proven in~\cite{VY}, Section 5,
implies the following bound.
%
%
\begin{lemma}\label{tee}
Under the same hypotheses of Theorem~\ref{dentini}, it holds
\[
\lim_{\alpha\uparrow\infty}
\sup_{F\in C^{\infty}([0,T]\times\bb T^d; \bb R^d)}
\limsup_{N\uparrow\infty}
\frac{1}{N^d} \log\bb P^{E,H, \mathbf{g},N}_{\eta^N}
\bigl(\mc Q_{F}(\pi^{N}) > \alpha\bigr) =-\infty.
\]
\end{lemma}

Fix a countable family $\{F_k\}\subset C^\infty([0,T]\times\bb
T^d; \bb R^d)$ of smooth vector fields dense in
$C^1([0,T]\times\bb T^d;\bb R^d)$.
Given $n\in\bb N$ and $\alpha\in\bb R_+$, set
%
%
\begin{equation}
\label{4man}
\mc M^{\alpha,n}:= \Bigl\{ \pi\in\mc M_{[0,T]}\dvtx
\max_{k\in\{1,\ldots,n\}} \mc Q_{F_k} (\pi) \le\alpha
\Bigr\},
\end{equation}
so that $\mc M^{\alpha}:= \{\pi\in\mc M_{[0,T]}\dvtx \mc Q (\pi
)\le
\alpha\} = \bigcap_n \mc M^{\alpha,n}$. The following statement is
then an immediate corollary of Lemma~\ref{tee}.
%
%
\begin{corollary}
\label{tcee}
Under the same hypotheses of Theorem~\ref{dentini}, it holds
\[
\limsup_{\alpha\uparrow\infty, n\uparrow\infty, N\uparrow
\infty}
\frac{1}{N^d} \log\bb P^{E,H, \mathbf{g}, N}_{\eta^N}
( \pi^{N} \notin\mc M^{\alpha,n} ) = -\infty.
\]
\end{corollary}

\subsection*{Identification of the hydrodynamic equation}

The following result will allow us to characterize the limit points of
$\{\bb P^{E,H,\mathbf{g},N}_{\eta^N} \circ(\pi^N)^{-1}
\}_{N\ge1}$. Recall the notation for the smooth convolution
introduced in (\ref{iacopo}).
%
%
\begin{proposition}
\label{asinelli}
Given $\varphi\in C^1([ 0,T]\times\bbT^d )$ and a path $\pi\in\mc
M_{[0,T]}$, set
\begin{eqnarray*}
W_T(\pi) & := & \langle\pi_T,\varphi_T\rangle
- \langle\pi_0, \varphi_0\rangle
-\int_0^T dt\, \langle\pi_t, \partial_t \varphi_t \rangle
\\
&&{} + \int_0^T dt\,
\langle\nabla\varphi_t,
\sigma(\tilde{\pi}^{\k, a}_t) [E +\nabla H_t]
- D(\tilde{\pi}^{\k, a}_t)
\nabla\tilde{\pi}^{\k', \epsilon}_t\rangle.
\end{eqnarray*}
Then, under the same hypotheses of Theorem~\ref{dentini}, for each
$\z>0$ it holds
\[
\limsup_{\k\downarrow0, a \downarrow0, \k'\downarrow0,
\e\downarrow0, N \uparrow\infty}
\mathbb{P} ^{E, H, \mathbf{g},N}_{\eta^N}\bigl(
|W_T(\p^N)|>\z\bigr) =0.
\]
\end{proposition}

The proof of the above result will be based on standard
martingale estimates, the super-exponential bounds in
Proposition~\ref{cile}, the Taylor expansion of the
rates
%
%
\begin{eqnarray}
\label{taylor-rate}\quad
c_{x,x+e_i}^{E, H, \mathbf{g}}(\eta)
& = & c_{x,x+e_i}^0(\eta)
+ \frac1{2N} c_{x,x+e_i}^0(\eta) ( \eta_x- \eta_{x+e_i} )
[E_i + \partial_i H_t] \biggl( \frac{x}{N}\biggr)
\nonumber\\[-8pt]\\[-8pt]
&&{} + c_{x,x+e_i}^0(\eta) \nabla_{x,x+e_i} \upbar{F} (t,
\eta)
+ O\biggl(\frac{1}{N^2}\biggr)\nonumber
\end{eqnarray}
and of the currents
%
%
\begin{eqnarray}
\label{taylor-current}\qquad
j_{x,x+e_i} ^{E,H, \mathbf{g}}(\eta)
& = & j_{x,x+e_i}^0 (\eta)
+\frac1{2N} c^0_{x,x+e_i} (\eta) (\eta_x- \eta_{x+e_i})^2
[E_i + \partial_i H_t] \biggl(\frac{x}{N}\biggr)
\nonumber\\[-8pt]\\[-8pt]
&&{} + c^0_{x,x+e_i}(\eta)
( \eta_x- \eta_{x+e_i} ) \nabla_{x,x+e_i} \upbar{F}(t,\eta)
+ O\biggl( \frac{1}{N^2}\biggr),\nonumber
\end{eqnarray}
where the function $\upbar F\equiv\upbar F ^N_{H,\ell, \mathbf{g} }
$ is
the one defined in (\ref{beta}).
\begin{pf}
Given $\psi_1,\psi_2\dvtx\bb T^d\to\bb R$, set
$\langle\psi_1,\psi_2\rangle_*:= \Av_x \psi_1(x/N) \psi_2(x/N)$
and observe that for any $\varphi\in C([ 0,T]\times\bbT^d )$ it
holds
\[
\lim_{N\to\infty} \langle\pi^N(\eta), \varphi_t\rangle_* =
\langle\p^N(\eta), \varphi_t \rangle
\]
uniformly for $t \in[0,T]$ and $\eta\in\Omega_N$. Hence, it is
enough to prove the statement with $\langle\cdot,\cdot\rangle$
replaced by $\langle\cdot,\cdot\rangle_*$.

By standard\vspace*{1pt} martingale estimates (see~\cite{KL}) and recalling the definition
of $L^{E,H, \mathbf{g}} _{t,N}$ given after (\ref{delfinodino}),
we get
%
%
\begin{eqnarray}\label{passo0}
&&
\lim_{N\uparrow\infty} \mathbb{P}^{E,H, \mathbf{g},N}_{\eta^N}
\biggl( \biggl|
\langle\pi^N(\eta(T)), \varphi_T\rangle_* -
\langle\pi^N(\eta(0)), \varphi_0\rangle_*\nonumber\\[-3pt]
&&\qquad\hspace*{49.5pt}{}-
\int_0^T dt\, \langle\pi^N(\eta(t)),
\partial_t \varphi_t\rangle_*
\\[-3pt]
&&\hspace*{17.6pt}\hspace*{49.5pt}\qquad{}
- \int_0^T dt\,
L^{E,H, \mathbf{g}} _{t,N} \langle\pi^N(\eta(t)), \varphi_t\rangle_*
\biggr| > \z\biggr)=0.\nonumber
\end{eqnarray}

We next introduce the microscopic scale parameters $\ell,c$ and the
family of good functions provided by Lemma~\ref{salsiccia} which, as in
the previous section, is denoted by~$\mathbf{g}[\delta]$.
All approximations below have to be understood with respect to the
limits $N \uparrow\infty$, $c \downarrow0$, $\e\downarrow0$,
$\k'\downarrow0$, $a \downarrow0$, $\ell\uparrow\infty$, $\k
\downarrow0$ and finally $\delta\downarrow0$. We use the functions
$T_1,\ldots, T_{16}$ and $K_1, \ldots, K_5$ introduced in Section~\ref{s31}.
Below we frequently use Remark~\ref{turn} without
explicit mention.

Since\vspace*{-2pt}
\[
L_{t,N}^{E,H, \mathbf{g}} \eta_x
= N^2 \sum_{i=1}^d j^{E,H,\mathbf{g}}_{x-e_i,x} (\eta)
- N^2 \sum_{i=1}^d j^{E,H, \mathbf{g}}_{x,x+e_i}(\eta),\vspace*{-2pt}
\]
summing by parts and using the Taylor expansion
(\ref{taylor-current}) we deduce
\begin{eqnarray*}
L^{E,H, \mathbf{g}}_{t,N} \langle\pi^N (\eta), \varphi_t\rangle_*
&=& N \sum_{i=1}^d \Av_x
\nabla^N_i \varphi_t \biggl(\frac{x}{N}\biggr)
j^{E,H, \mathbf{g}}_{x,x+e_i} (\eta)
\\[-2pt]
&=& [ T_1 + T_3 + T_4 + T_{5, \mathbf{g}} ] (t,\eta)
+ o(1).
\end{eqnarray*}
In particular, inside (\ref{passo0}) we can replace the last
integrand by
$[T_1 + T_3+ T_4 + T_{5,\mathbf{g}}](t,\eta(t))$.
By (\ref{grazia1}) and (\ref{grazia3}), we can replace $T_1$ by
$T_2$ and then $T_2$ by $- T_{11, \mathbf{g}[\delta]}-T_{12}$.
By (\ref{grazia4}), we can replace $T_3+T_4$ by $K_1$. By
(\ref{annarella1}) and (\ref{grazia5}), we
can replace $T_{5, \mathbf{g}} $ by $T_{6, \mathbf{g}}$ and then
$T_{6, \mathbf{g}}$ by $K_{4,\mathbf{g}} $.
In conclusion, inside (\ref{passo0}) we can replace the last
integrand by
%
%
\begin{equation}
\label{rotto1}
\bigl[ K_1+K_{4, \mathbf{g}} - T_{11, \mathbf{g}[\delta]}-T_{12}\bigr]
(t, \eta(t)).
\end{equation}
By a standard martingale estimate (see the paragraph before Lemma
3.6 in~\cite{VY}), it holds
\[
\limsup_{\ell\uparrow\infty, N\uparrow\infty}
\mathbb{P}^{E,H, \mathbf{g},N}_{\eta^N} \biggl( \biggl|
\int_0^T dt\,
T_{8, \mathbf{g}, \mathbf{g}[\delta]} (t, \eta(t))
\biggr| > \z\biggr)=0.
\]
In particular, in (\ref{rotto1}) we can add
$T_{8, \mathbf{g}, \mathbf{g}[\delta]}(t, \eta(t) )$.
By (\ref{annarella2}), this last expression is equivalent to $T_{7,
\mathbf{g}, \mathbf{g}[\delta]}(t, \eta(t))$. On the other hand, by
the Taylor expansion (\ref{taylor-rate}) we can write
\[
T_{7, \mathbf{g}, \mathbf{g}[\delta]} (t, \eta)= \bigl[T_{11, \mathbf
{g}[\delta]}+T_{13,
\mathbf{g}[\delta]}+T_{14,\mathbf{g},
\mathbf{g}[\delta]}\bigr](t, \eta)+o(1).
\]
By (\ref{bella1}), we can replace
$T_{13,\mathbf{g}[\delta]}$ by $K_{2, \mathbf{g}[\delta]}$, while by
(\ref{gaetano}) and (\ref{bella2}) we can replace $T_{14,
\mathbf{g}, \mathbf{g}[\delta]}$ by $T_{15,\mathbf{g},
\mathbf{g}[\delta]}$ and this by $K_{3, \mathbf{g}, \mathbf
{g}[\delta]}$.\vadjust{\goodbreak}

Let us stop and see where we are: up to now we have showed that
inside (\ref{passo0}) we can replace the last integrand by
\[
\bigl[ K_1+K_{2, \mathbf{g}[\delta] }+ K_{3,
\mathbf{g},\mathbf{g}[\delta] }+K_{4, \mathbf{g} }
-T_{12}\bigr](t, \eta(t)).
\]
In view of the estimates (\ref{barbalalla1}) and (\ref{barbalalla2}),
the above expression can be replaced by $[K_5-T_{12}] (t, \eta(t) )$.
Finally, using the two blocks estimate in Lem\-ma~\ref{duedue}, we can
replace in $K_5$ the microscopic scale $\ell$ with the mesoscopic one
$a N$ getting a new expression $[K_5'-T_{12}] (t, \eta(t) )$. Given
$\p\in\mc M$, we define $\p^a (r):= \p\ast\psi(r) $ where
$\psi(r):=(2a)^{-d} \bb I ( |r| \leq a)$. Due to (\ref{miracolo}) we
can replace $T_{12}$ with $T_{16}$ and, using the regularity of $\s$,
we can replace $\int_0^T dt\, K_5'(t, \eta(t) ) $ by $\int_0^T dt\,
\langle\nabla\varphi_t, \sigma( \pi^N(\eta(t))^a ) [E+\nabla
H_t]\rangle$. In addition, since $|\pi^N(\eta)^{a}
-\tilde{\pi}^N(\eta)^{\k',a}|_\infty \leq C\k'$, we can replace
$\pi^N(\eta)^{a}$ with $\tilde{\pi}^N(\eta)^{\k',a}$. Comparing with
the definition of $W_T$, the proof is complete.
\end{pf}

We can now conclude the proof of the hydrodynamic limit.
\begin{pf*}{Proof of Theorem~\ref{dentini}}
Set $\mc P_N^{E,H}:=\bb P^{E,H,\mathbf{g},N}_{\eta^N}
\circ(\pi^N)^{-1}$. As proven before, the sequence
$\{\mc{P}^{E,H}_N\}$ is relatively compact. We therefore only need
to show that any limit point $\mc P$ equals $\delta_u$. By taking a
subsequence, we can assume that $\mc P_N^{E,H}$ converges weakly to
$\mc P$. By the continuity of $\mc Q_{F}$ and Portmanteau
theorem $\mc P(\mc{M}^{\alpha,n}) \ge\limsup_N \mc
P^{E,H}_N(\mc{M}^{\alpha,n})$. Corollary~\ref{tcee}
then yields $\lim_{\alpha\to\infty} \mc P( \mc{M}^\alpha)
=1$. Hence, $\mc P$ almost surely, the weak gradient $\nabla\pi$
belongs to $L^2 ([0,T]\times\bb T^d, dt \,dr;\bb R^d)$.

We write the function $W_T$ defined in Proposition~\ref{asinelli} as
$ W_T (\tilde{\p}^{\k,a}, \tilde{\p}^{\k',\e} )$. Moreover, given
$\pi\in\mc M_T$ satisfying the energy estimate, we let
$W_T(\pi,\pi)$ be the same expression with $\tilde{\p}^{\k,a}$ and
$\tilde{\p} ^{\k',\e}$ both replaced by $\pi$. By Schwarz
inequality and the regularity of $D$ and $\s$, there exists a
constant $C$ not depending on the scale parameters such that
\[
| W_T (\tilde{\p}^{\k,a}, \tilde{\p} ^{\k',\e} )
-W_T (\p, \p) |
\leq C ( \| \tilde\p^{\k,a} -\p\|_2
+ \| \nabla\tilde{\p}^{\k', \e}- \nabla\p\|_2 ),
\]
where $\|\cdot\|_2$ is the norm in $L^2 ([0,T]\times\bb T^d,
dt \,dr)$. Since $\|\pi_t \|_\infty\le1 $ and $\mc P$ almost surely
$\nabla\pi$ belongs to $L^2( [0,T]\times\bb T^d, dt \,dr;\bb R^d)$
by standard properties of convolutions we deduce that for each
$\zeta>0$
\[
\limsup_{\k\downarrow0, \epsilon\downarrow0,\k'\downarrow0,
a\downarrow0}
\mc P \bigl( | W_T (\tilde{\p}^{\k,a}, \tilde{\p} ^{\k',\e} )
-W_T (\p, \p) | > \zeta\bigr)=0.
\]
On the other hand, Proposition~\ref{asinelli} and Portmanteau
theorem imply that for each $\zeta>0$
\begin{eqnarray*}
&&
\limsup_{\k\downarrow0, \epsilon\downarrow0,\k'\downarrow0,
a\downarrow0}
\mc P\bigl(
|W_T(\tilde{\p}^{\k,a},\tilde{\p}^{\k',\e})|>\zeta
\bigr)
\\
&&\qquad \le
\limsup_{\k\downarrow0, \epsilon\downarrow0,\k'\downarrow0,
a\downarrow0, N\uparrow\infty}
\mc P_N^{E,H} \bigl(
|W_T(\tilde{\p}^{\k,a},\tilde{\p}^{\k',\e})|>\zeta
\bigr) =0.
\end{eqnarray*}

The above results readily imply that the identity $W_T(\pi,\pi)=0$
holds $\mc P$ almost surely. Since by hypothesis\vadjust{\goodbreak} the sequence
$\{\eta^N\}$ is associated to the profile $\g$, this amounts to say
that $\pi$ is $\mc P$ almost surely a weak solution to (\ref{4he}).
By the uniqueness of such solution we conclude $\mc
P=\delta_u$.
\end{pf*}

\section{Dynamical large deviation principle}

In this section we prove Theorem~\ref{tdldp}. Since the driving
field $E$ and the time $T$ are here kept fixed, we drop them from
most of the notation. In particular, the space $\mc M_{[0,T]}$ is
denoted by $\mc M$ and the rate function defined in (\ref{2Ig}) by
$I(\cdot|\gamma)$. Recall that $ \mc P^{E,N}_{\eta^N}:=\bb
P^{E,N}_{\eta^N}
\circ(\pi^N)^{-1}$.

\subsection{Upper bound}
\label{s4}

We first outline the basic strategy, which is the classical Varadhan's
one~\cite{Vld} for Markov processes applied to the context of
interacting particle systems in the diffusive scaling limit
\cite{BLM,KL,KOV,Q,QRV}.
In view of the exponential tightness already proven, it is enough to
show the upper bound (\ref{2ub}) for compact sets. Moreover,
Corollary~\ref{tcee} implies that the probability of paths $\pi$ not
satisfying the energy estimate is super-exponential small as $N$
diverges; more precisely, that the large deviations rate function is
infinite if the weak gradient of $\pi$ does not belong to
$L^2([0,T]\times\bb T^d,dt \,dr;\bb R^d)$, that is, the second line in
(\ref{2Ig}).
By constructing a suitable family of exponential martingales for the
probability measures $\bb P^{E,N}_{\eta^N}$, we then essentially prove
that for any measurable set $\mc B$ in $\mc M$ and any function
$H\in C^{1,2}([0,T]\times\bb T^d)$
%
%
\begin{equation}
\label{5ubb}
\limsup_{N\to\infty} \frac1{N^d} \log
\mc P^{E,N}_{\eta^N}(\mc B) \le- \inf_{\pi\in\mc B}
J_{H,\gamma}(\pi),
\end{equation}
where, recalling (\ref{2lpH}), if $\pi\in\mc M$ satisfies the
energy estimate $J_{H,\gamma}(\pi)$ is given by
%
%
\begin{eqnarray}
\label{5JH=}\quad
J_{H,\gamma}(\pi) &=& \ell_{\gamma,\pi}(H)
- \int_0^T dt\, \langle\nabla H_t,\sigma(\pi_t)\nabla H_t\rangle
\nonumber\\
&=& \langle\pi_T, H_T \rangle- \langle\gamma, H_0
\rangle
\\
&&{} -
\int_0^{T} dt\, \bigl[ \langle\pi_t,\partial_t
H_t \rangle+\langle\sigma(\pi_t) [E+\nabla H_t] -
D(\pi_t)\nabla\pi_t,\nabla H_t\rangle\bigr].\nonumber
\end{eqnarray}
This is clearly the main step of the proof; the exponential
martingales are constructed from the microscopic dynamics and are
not a function of the empirical density. However, the
super-exponential bounds proven in Proposition~\ref{cile} imply that
such exponential martingales can be approximated by functions of the
empirical density with probability super-exponentially close to one
as $N$ diverges. In view of the variational definition (\ref{2Ig})
of the rate function $I(\cdot|\gamma)$, the upper bound (\ref{2ub})
for compact sets then follows from (\ref{5ubb}) and
(\ref{5JH=})
by an application of a min--max lemma. As stated before, while for
gradient models the exponential martingales are constructed simply
by changing the driving field, for nongradient models, the
correction provided by Lemma~\ref{salsiccia} is needed.\vadjust{\goodbreak}

\subsubsection*{Exponential martingales}

Fix $E \in C^1( \bbT^d; \bbR^d)$,
$H\in C^{1,2} ([0,T]\times\bb T^d)$ and a family
of good functions $\mathbf{g}=\{g^{(i)}\dvtx 1\leq i \leq d \}$.
Given $\ell\ge1$, recall the definition of the function
$F\equiv F^N_{H,\ell, \mathbf{g}}$ given in (\ref{alfa}) and
consider the
exponential martingale $\mc E \equiv\mc E ^N_{H, \ell, \mathbf{g} }$
associated to the function $2F$, that is,
%
%
\begin{eqnarray}
\label{5expM}
\mc E (t) &:=& \exp\biggl\{ 2 F (t, \eta(t)) - 2 F(0,\eta(0))
\nonumber\\[-8pt]\\[-8pt]
&&\hspace*{20pt}{}
-\int_0^t ds\, \bigl[
e^{-2F (s,\eta(s))} (\partial_s+L_{E,N}) e^{2F(s,\eta(s))}
\bigr]\biggr\}.\nonumber
\end{eqnarray}
By, for example,~\cite{KL}, Appendix 1.7, $\mc E(t)$ is indeed a mean one positive
martingale with respect to the measure $\bb P^{E,N}_{\eta^N}$. We
next\vspace*{1pt}
show that, as $N$ diverges, $ \mc E$ is super-exponentially close to a
function of the empirical density. The first step, stated below, comes
directly from a Taylor expansion of the exponential and
(\ref{taylor-rate}); we therefore omit the proof.
%
%
\begin{lemma}
\label{texpE}
Set $\mc{J}^N_{H,\ell, \mathbf{g} }(\eta):= N^{-d} \log\mc
E^N_{H,\ell,\mathbf{g}}(T)$, $\eta\in D([0,T];\Omega_N)$. Then
\begin{eqnarray*}
\mc{J}^N_{H, \ell, \mathbf{g} }(\eta) &=& \langle\pi^N(\eta
(T)),H_T\rangle
-\langle\pi^N(\eta(0)),H_0\rangle
\\
&&{}
- \int_0^T dt\, [ \langle\pi^N(\eta(t)),\partial_t H_t\rangle
+ J_1 (t, \eta(t)) + J_2 (t, \eta(t))\\
&&\hspace*{125.1pt}{} + J_3 (t, \eta(t))
+ R (t,\eta(t))],
\end{eqnarray*}
where, for $\eta\in\Omega_N$,
\begin{eqnarray*}
J_1(t,\eta) &=& J^{N}_{1, H, \ell, \mathbf{g} } (t, \eta)\\
:\!&=&N \Av_x\sum_{i=1}^d c^0_{x,x+e_i}(\eta)
\Biggl[\nabla^N_i H_t\biggl(\frac xN\biggr) (\eta_x-\eta_{x+e_i})\\
&&\hspace*{90.5pt}{} +
\nabla_{x,x+e_i} \sum_{z }\sum_{j=1}^d
\nabla^N_j H_t\biggl(\frac zN\biggr)
g^{(j)} (\tau_z\eta,\bar{\eta}_{z,\ell}) \Biggr]
,\\
J_2(t,\eta) &=& J^{N}_{2,H, \ell, \mathbf{g}} (t, \eta)\\
:\!&=& \frac12 \Av_x \sum_{i=1}^d
c^0_{x,x+e_i}(\eta) E_i\biggl( \frac xN\biggr) (\eta_x- \eta_{x+e_i})\\
&&\hspace*{36.5pt}{} \times
\Biggl[ \nabla^N_i H_t\biggl(\frac xN\biggr) (\eta_{x}-\eta_{x+e_i})\\
&&\hspace*{54.1pt}{} + \nabla_{x,x+e_i} \sum_{z }
\sum_{j=1}^d \nabla^N_j H_t\biggl(\frac zN\biggr)
g^{(j)}(\tau_z\eta,\bar{\eta}_{z,\ell})\Biggr],
\\
J_3(t,\eta) &=& J^{N}_{3,H, \ell, \mathbf{g} }(t, \eta)\\[-2pt]
:\!&=& \frac12
\Av_x\sum_{i=1}^d c^0_{x,x+e_i}(\eta)
\Biggl[ \nabla^N_i H_t\biggl(\frac xN\biggr) (\eta_{x}-\eta_{x+e_i})\\[-2pt]
&&\hspace*{89pt}{} + \nabla_{x,x+e_i} \sum_{z}\sum_{j=1}^d
\nabla^N_j H_t\biggl(\frac zN\biggr)
g^{(j)}(\tau_z\eta,\bar{\eta}_{z,\ell}) \Biggr]^2,
\end{eqnarray*}
while the error term $R= R^N_{H, \ell, \mathbf{g} }$ satisfies
\[
\sup_{t\in[0,T]} \sup_{\eta\in\Omega_N}
| R(t, \eta) | \le\frac C{N}
\]
for some constant $C>0$ depending on $T,H,\ell, \mathbf{g}$.
\end{lemma}

We next choose the family $\mathbf{g}$ as the one provided by
Lemma~\ref{salsiccia}; as usual, we denote it by $ \mathbf{g}[\delta]$.
Then the super-exponential estimates in Proposition~\ref{cile}
together with Remark~\ref{turn} imply the following key result.
%
%
\begin{proposition}
\label{tJNJ}
Fix $T>0$, $E \in C^1(\bb T^d;\bbR^d)$, $H\in C^{1,2}
([0,T]\times\bb T^d)$, a~profile $\gamma\in M$, a sequence
$\{\eta^N\in\Omega_N\}$ associated to $\gamma$, and let
$\mc{J}^N_{H,\ell, \mathbf{g} }$ be defined as in
Lemma~\ref{texpE}.
Then, for each $\z>0$ it holds
\begin{eqnarray*}
&&\limsup_{\delta\downarrow0, \k\downarrow0, \ell\uparrow
\infty,
a\downarrow0, \k'\downarrow0, \epsilon\downarrow0, N
\uparrow\infty}
\frac1{N^d}
\log\bb P^{E,N}_{\eta^N} \bigl(
\bigl| \mc{J}^N_{H,\ell, \mathbf{g}[\delta]}(\eta)
- \widehat{J}_{H,\gamma}( \p^N(\eta) ) \bigr|>\zeta\bigr)
\\[-2pt]
&&\qquad
= -\infty,
\end{eqnarray*}
where, for $\pi\in\mc M$,
%
%
\begin{eqnarray}
\label{hatJ}
\widehat{J}_{H,\gamma}(\pi)
&=& \langle\pi_T, H_T \rangle- \langle\gamma, H_0 \rangle
\nonumber\\[-2pt]
&&{} -
\int_0^{T} dt\, \bigl[ \langle\pi_t, \partial_t H_t\rangle\\[-2pt]
&&\hspace*{45.2pt}{}
+\langle\sigma(\tilde{\pi}^{\k,a}_t) [E+\nabla H_t] - D
(\tilde{\pi}^{\k,a}_t) \nabla\tilde{\pi}^{\k',\e} _t,
\nabla H_t\rangle\bigr].\nonumber
\end{eqnarray}
\end{proposition}
\begin{pf}
In what follows we write $\mathbf{g}$ instead $\mathbf{g}[\delta]$,
understanding the dependence on $\delta$. In order to have compact
formulae below, it is also convenient to introduce the following
notation.
Given functions $F_1, F_2$ on $[0,T]\times\Omega_N$ depending also on the
parameters $\delta, \k, \ell, a, \k',\e,c,N$, we write $F_1\sim
F_2$ if
for any $\z>0$ it holds
\begin{eqnarray*}
&&\limsup_{\delta\downarrow0, \k\downarrow0, \ell\uparrow
\infty,
a\downarrow0, \k'\downarrow0, \e\downarrow0,
c\downarrow0, N \uparrow\infty}
\frac1{N^d} \\[-2pt]
&&\quad{}\times\log\bb P^{E,N}_{\eta^N} \biggl(
\biggl| \int_0^T dt\,
[F_1(t, \eta(t) )
- F_2(t, \eta(t) )] \biggr|>\zeta\biggr)\\[-2pt]
&&\qquad= -\infty.
\end{eqnarray*}

We use Lemma~\ref{texpE} and analyze separately the terms $J_1,
J_2, J_3$. We start by $J_1$, which can be rewritten as
\[
J_1(t, \eta)= N \Av_x \sum_{i=1}^d
\nabla^N_i H_t \biggl(\frac xN\biggr)
\bigl[ j^0 _{x,x+e_i}(\eta) +
L_0 g^{(i) } (\tau_x \eta,\upbar{\eta}_{x,\ell}) \bigr].
\]
Consider the expressions $T_1,\ldots, T_{16}, K_1,\ldots, K_5$ defined
in Section~\ref{s31}, where now the function $\varphi$ entering
in their definition has to be replaced by $H$. By the same
arguments used to derive (\ref{grazia1}), it holds $J_1 \sim T_2 +
T_{9, \mathbf{g} }$. Due to (\ref{annarella3}) and
(\ref{annarella4}) we then get
$T_{9, \mathbf{g}} \sim T_{10,\mathbf{g} } \sim T_{11, \mathbf{g}}$.
Hence, we get that $J_1 \sim T_2+ T_{11, \mathbf{g} }$. Finally,
by (\ref{grazia3}) and (\ref{miracolo}), we get
%
%
\begin{equation}
\label{carbonara1}
J_1 (t, \eta) \sim- T_{12 }(t, \eta)
\sim- T_{16}(t, \eta).
\end{equation}

We now analyze the term $J_2$. Due to Definition~\ref{buono} of good
function, in the expression of $J_2$ given in Lemma~\ref{texpE} we
can restrict the sum over $z$ to the set $\{z\dvtx |z-x|\leq
(C+\ell)\}$, where the constant $C>0$ is such that the functions
$g^{(i)} (\cdot, \rho) $ have support inside $\L_C$ for all $i=1,
\ldots, d$ and $\rho\in[0,1]$. As a consequence, in $J_2$ we can
first replace discrete gradients by partial derivatives; afterward
we can replace $\partial_j H_t(z/N)$ by $\partial_j H_t(x/N)$ with
an error $O(\ell/N)$. Moreover, similar to (\ref{annarella1}), we can
replace $\nabla_{x,x+e_i}$ with $\nabla^1_{x,x+e_i}$. At this point,
by the one block estimate and (\ref{bella1}), we get
\begin{eqnarray*}
J_2 (t, \eta)
&\sim&\frac{1}{2} \Av_x \sum_{i=1}^d
E_i\biggl(\frac xN\biggr) \,\partial_i H_t \biggl(\frac xN\biggr)
\mu_{\upbar{\eta}_{x, \ell} }
[ c_{0,e_i }^0(\z) ( \z_0 - \z_{e_i})^2]
\\
&&{} +\frac{1}{2} \Av_x \sum_{i=1}^d \sum_{j=1}^d
E_i\biggl( \frac xN\biggr) \,\partial_j H_t \biggl(\frac xN\biggr)\\
&&\hspace*{66pt}{}\times
\mu_{\upbar{\eta}_{x, \ell} }
\bigl[ c_{0,e_i}^0(\z) (\z_0-\z_i)
\nabla_{0, e_i}^1 \underbar{g}^{(j)}(\z,\bar{\eta}_{x,\ell})\bigr].
\end{eqnarray*}
Recall the discussion of the CLTV in Section~\ref{giulioconiglio},
in particular the definitions of the inner product
$\langle \cdot,\cdot\rangle _\rho$ and of the orthogonal projector $P$.
By (\ref{pp1}), (\ref{pp2}) and Lemma~\ref{salsiccia} we then get
\begin{eqnarray*}
J_2 (t, \eta) & \sim &
\Av_x \Biggl\langle  \sum_{i=1}^d E_i \biggl(\frac xN\biggr) j_{0,e_i}
, \sum_{i=1}^d \partial_i H_t\biggl(\frac xN\biggr)
\bigl[ j_{0,e_i} + L^1_0 \underbar{g}^{(i)}(\cdot,
\rho) \bigr]
\Biggr\rangle _{\rho=\bar{\eta}_{x, \ell} }
\\
&\sim&\Av_x \Biggl\langle
\sum_{i=1}^d E_i \biggl(\frac xN\biggr) j_{0,e_i},
\sum_{i=1}^d \partial_i H_t\biggl(\frac xN\biggr)
P j_{0,e_i} \Biggr\rangle _{\rho= \bar{\eta}_{x, \ell}
}.
\end{eqnarray*}
In view of (\ref{vagabondo1}), we deduce that
$J_2 (t, \eta) \sim{\Av} _x E (x/N) \cdot\s( \bar{\eta
}_{x,\ell}
) \nabla H_t (x/N)$. Applying the two blocks estimate and
afterward making a uniform estimate, we conclude that
%
%
\begin{eqnarray}
\label{carbonara2}
J_2 (t, \eta) &\sim&\Av_x
E \biggl(\frac xN\biggr) \cdot
\sigma\biggl( \tilde{\p}^{N}(\eta)^{\k,a} \biggl(\frac xN\biggr)
\biggr) \nabla H_t \biggl(\frac xN\biggr)\nonumber\\[-8pt]\\[-8pt]
&\sim&\langle E,
\sigma( \tilde{\p}^{N}(\eta)^{\k,a} ) \nabla
H_t\rangle.\nonumber
\end{eqnarray}

We finally consider $J_3$. As done for $J_2$, we can replace
discrete gradients by partial derivatives; afterward we can replace
$\partial_j H_t(z/N)$ by $\partial_j H_t(x/N)$ and finally $\nabla
_{x, x+e_i}$ by $\nabla^1 _{x,x+e_i}$. Then, by the one block
estimate together with (\ref{bella1}) and (\ref{bella2}), we can
write
\begin{eqnarray*}
J_3(t,\eta) &\sim& \frac{1}{2} \Av_x \sum_{i=1}^d
\partial_i H_t\biggl(\frac xN\biggr)^2
\mu_{\bar{\eta}_{x,\ell}}
[ c_{0, e_i}^0(\z) (\z_0-\z_{e_i})^2]
\\
&&{} + \Av_x \sum_{i=1}^d \sum_{j=1}^d
\partial_i H_t\biggl(\frac xN\biggr) \,\partial_j H_t\biggl(\frac xN\biggr)
\mu_{\bar\eta_{x, \ell} }\\
&&\qquad\quad\hspace*{24pt}{}\times\bigl[ c_{0,e_i}^0(\z) ( \z_0-\z_{e_i})
\nabla_{0, e_i}^{1}
\underbar{g}^{(j)}(\z,\bar{\eta}_{x,\ell})\bigr]
\\
&&{} +\frac{1}{2} \Av_x \sum_{i=1}^d \sum_{j=1}^d \sum_{k=1}^d
\partial_j H_t\biggl(\frac xN\biggr) \,\partial_k H_t \biggl(\frac xN\biggr)
\\
&&\hspace*{83pt}{} \times
\mu_{ \bar{\eta}_{x, \ell}}
\bigl[ c_{0,e_i}^0(\z)
\nabla_{0, e_i}^{1 } \underbar{g}^{(j)} (\z,\bar{\eta}_{x,\ell})\\
&&\qquad\quad\hspace*{111pt}{}\times\nabla_{0,e_i}^{1}\underbar{g}^{(k)}(\z,\bar{\eta}_{x,\ell})
\bigr].
\end{eqnarray*}
Recalling that $\langle f,f\rangle _\rho= V_\rho(f)$,
from the identities (\ref{pp1}), (\ref{pp2}), (\ref{pp5}) and
Lemma~\ref{salsiccia} we deduce
%
%
\begin{eqnarray}
\label{carbonara25}
J_3 (t,\eta) & \sim & \Av_x
V_{\rho=\bar{\eta}_{x, \ell}}
\Biggl( \sum_{i=1}^d \partial_i H_t \biggl(\frac xN\biggr)
\bigl[ j_{0,e_i}^0 +
L^1_0 g^{(i)}(\cdot,\rho)\bigr] \Biggr)
\nonumber\\[-8pt]\\[-8pt]
& \sim & \Av_x
V_{\rho=\bar{\eta}_{x, \ell}}
\Biggl( \sum_{i=1}^d
\partial_i H_t \biggl(\frac xN\biggr) P j_{0,e_i}^0 \Biggr).\nonumber
\end{eqnarray}
Then, by (\ref{vagabondo1}), we get $J_3 (t, \eta) \sim
\Av_x \nabla H_t (x/N) \cdot
\s( \bar{\eta}_{x, \ell}) \nabla H_t(x/N)$.
As in the derivation (\ref{carbonara2}) we then conclude
%
%
\begin{eqnarray}
\label{carbonara3}
J_3 (t, \eta) &\sim& \Av_x
\nabla H_t \biggl(\frac xN\biggr) \cdot
\s\bigl(\tilde{\p}^{N}(\eta)^{\k, a} (x/N) \bigr)
\nabla H_t \biggl(\frac xN\biggr)
\nonumber\\[-8pt]\\[-8pt]
&\sim& \langle\nabla H_t,
\s( \tilde{\p}^{N}(\eta)^{\k,a} ) \nabla
H_t\rangle.\nonumber
\end{eqnarray}
The thesis now follows combining Lemma~\ref{texpE},
(\ref{carbonara1}), (\ref{carbonara2}) and (\ref{carbonara3}).
\end{pf}

\subsubsection*{Conclusion}

Recall the definitions of the set $\mc{M}^{\alpha,n}$ in (\ref{4man})
and of the functional $\widehat{J}_{H,\gamma}$ in (\ref{hatJ}).
Let $J_{H,\gamma}^{\alpha,n}\dvtx \mc M \to[0,+\infty]$ be
the functional defined by
%
%
\begin{equation}
\label{5JHan}
J_{H,\gamma}^{\alpha,n}(\pi):=
\cases{
\widehat{J}_{H,\gamma}(\pi), &\quad if $\pi\in\mc M^{\alpha,n}$,\cr
+\infty, &\quad otherwise.}
\end{equation}
Note that, even if not explicitly indicated in the notation, the
functional $ J_{H,\gamma}^{\alpha,n}$ depends also on the parameters
$\k,a,\k',\epsilon$.
%
%
\begin{lemma}
\label{tubB}
Fix $T>0$, a vector field $E\in C^1(\bb T^d;\bb R^d)$, a profile
$\gamma\in M$ and a sequence $\{\eta^N \in\Omega_N\}$
associated to $\gamma$. For each $H\in C^{1,2}([0,T]\times\bb
T^d)$ and each measurable set $\mc B\subset\mc M$, it holds
\[
\limsup_{N\to\infty}
\frac1{N^d} \log\mc P^{E,N}_{\eta^N} (\mc B )
\le\Bigl[ -\inf_{\pi\in\mc B} J_{H,\gamma}^{\alpha,n}(\pi)\Bigr]
\vee R^{\alpha,n}_{\k, \ell, a,\k',\epsilon},
\]
where
\[
\limsup_{\alpha\uparrow\infty, n\uparrow\infty,
\k\downarrow0, \ell\uparrow\infty, a\downarrow0,
\k'\downarrow0, \epsilon\downarrow0}
R^{\alpha,n}_{\k,\ell, a,\k',\epsilon} = -\infty.
\]
\end{lemma}
\begin{pf}
Recall Proposition~\ref{tJNJ} and, given $\zeta>0$, let
$\mc G^N_{H} (\zeta)$ be the subset of $D([0,T];\Omega_N)$ defined
by
\[
\mc G^N_{H} (\zeta):=
\bigl\{\eta\in D([0,T];\Omega_N)\dvtx
\bigl| \mc{J}^N_{H,\ell, \mathbf{g}[\delta]} (\eta)
- \widehat{J}_{H,\gamma}(\pi^{N}(\eta)) \bigr|
\le\zeta\bigr\}.
\]
Given the measurable set $\mc B \subset\mc M$, set also
\[
B^{N}_{H} (\zeta):=
\{\eta\in D([0,T];\Omega_N)\dvtx
\pi^N(\eta) \in\mc B \cap\mc M^{\alpha,n}\}
\cap\mc G^N_{H} (\zeta).
\]
Then, by Proposition~\ref{tJNJ} and Corollary~\ref{tcee}, for each
$\zeta>0$
\[
\limsup_{\alpha\uparrow\infty, n\uparrow\infty,
\delta\downarrow0,
\k\downarrow0, \ell\uparrow\infty, a\downarrow0,
\k'\downarrow0, \epsilon\downarrow0, N\uparrow
\infty}
\frac1{N^d} \log
\bb P^{E,N}_{\eta^N} ( B^{N}_{H} (\zeta)^\complement)
=-\infty.
\]
On the other hand, recalling $\mc E(t)$ in (\ref{5expM}) is a
positive mean one martingale with respect to the probability $\bb
P^{E,N}_{\eta^N}$ and $\mc E(T)=\exp\{N^d \mc J^N_{H,\ell,
\mathbf{g}}\}$,
\begin{eqnarray*}
\bb P^{E,N}_{\eta^N}( B^{N}_{H} (\zeta))
&=& \bb E^{E,N}_{\eta^N}\bigl( \mc E(T)
\exp\bigl\{-N^d \mc J^N_{H,\ell, \mathbf{g}[\delta]}\bigr\}
\bb I_{B^{N}_{H} (\zeta)} \bigr)
\\
&\le& \sup_{\pi\in\mc B}
\exp\{ - N^d [ J^{\alpha,n}_{H,\gamma}(\pi) - \zeta
]\}.
\end{eqnarray*}
The statement is a straightforward consequence of the above bounds.
\end{pf}
\begin{pf*}{Proof of Theorem~\ref{tdldp} the upper bound}
In view of the exponential tightness of the sequence $\{\mc
P^{E,N}_{\eta^N}\}$, it is enough to prove the bound
(\ref{2ub}) for compact sets. Observe\vspace*{1pt} that, for each $H\in
C^{1,2}([0,T]\times\bb T^d)$, the functional
$J_{H,\gamma}^{\alpha,n}$ is lower semicontinuous\vadjust{\goodbreak} on $\mc M$. From
Lemma~\ref{tubB} and the min--max lemma in~\cite{KL}, Appendix 2,
Lemma 3.3, we deduce that for each compact $\mc K\subset
\mc M$
\[
\limsup_{N\to\infty}
\frac1{N^d} \log\mc{P}^{E,N}_{\eta^N} (\mc K )
\le- \inf_{\pi\in\mc K}
\sup_{H,\alpha, n,\k, \ell, a, \k', \epsilon}
\{ J_{H,\gamma}^{\alpha,n}(\pi)
\wedge(- R^{\alpha,n}_{\k,\ell, a,\k',\epsilon} )
\}.
\]
In view of Lemma~\ref{tubB} and the variational definition
(\ref{2Ig}) of the rate function, the proof of (\ref{2ub}) is now
completed by taking the limits $\epsilon\downarrow0,
\k'\downarrow0, a\downarrow0, \ell\uparrow\infty, \k\downarrow0,
n\uparrow\infty, \alpha\uparrow\infty$, and finally optimizing
over $H$ (see~\cite{BLM}, Section~3.3, for more details).
\end{pf*}

\subsection{Lower bound}
\label{s5}

The following is a general result concerning the large deviation
lower bound. Its proof is elementary (see~\cite{Je}, Proposition 4.1).
Given two probability measures $P$ and $Q$ we denote by
$\operatorname{Ent}(Q|P)=\int dQ \log\frac{dQ}{dP}$ the \textit{relative
entropy} of $Q$ with respect to $P$.
%
%
\begin{lemma}
\label{tglb}
Let $\{P_n\}$ be a sequence of probability measures on a Polish
space $\mc X$ and $\mc X^\circ\subset\mc X$. Assume that for each
$x\in\mc X^\circ$ there exists a sequence of probability measures
$\{Q^x_n\}$ which converges weakly to $\delta_x$ and such that
%
%
\begin{equation}
\label{5lent}
\limsup_{n} \frac1n \operatorname{Ent} (Q^x_n | P_n ) \le I^\circ(x)
\end{equation}
for some function $I^\circ\dvtx \mc X^\circ\to[0,+\infty]$. Then $\{
P_n\}$
satisfies the large deviation lower bound with rate function
$I\dvtx\mc X \to[0,+\infty]$ given by
%
%
\begin{equation}
\label{5rilas}
I(x) = \sup_{\mc O \in\mc N_x} \inf_{y \in\mc O \cap\mc X^0}
I^\circ(y),
\end{equation}
where $\mc N_x$ denotes the collection of open neighborhoods of $x$.
\end{lemma}

Let $\tilde I\dvtx \mc X\to[0,+\infty]$ be the functional defined by
\[
\tilde I (x):=
\cases{
I^\circ(x), &\quad if $x\in\mc X^\circ$, \cr
+\infty, &\quad otherwise.}
\]
Then the functional $I$ in (\ref{5rilas}) is the lower
semicontinuous envelope of $\tilde I$, that is, the largest lower
semicontinuous functional below $\tilde I$. As is simple to show, the
condition that a large deviation rate function is lower
semicontinuous is not restrictive. More precisely, if a sequence of
probabilities satisfies the large deviation lower bound for some
rate function $\tilde I$, then the lower bound still holds with the
lower semicontinuous envelope of $\tilde I$. The previous lemma is
therefore stating that the entropy bound (\ref{5lent}) implies the
large deviation lower bound.

We are going to use Lemma~\ref{tglb} with $\mc X^\circ$ given by the
collection of some ``nice'' paths in $\mc M$. For such paths we can
prove the bound (\ref{5lent}) with $I^\circ$ given by the restriction
of the functional $I(\cdot|\gamma)$ defined in (\ref{2Ig}). To
conclude the proof of the lower bound (\ref{2lb}) we then need to
show the functional $I$ in (\ref{5rilas}) coincides with the
functional $I(\cdot|\gamma)$ on the whole space $\mc M$. We start by
defining precisely what we mean by ``nice'' paths. We basically
require that $\pi$ is a smooth function bounded away from zero and
one. However, as $I(\pi|\gamma)< +\infty$ implies $\pi_0=\gamma$
and $\gamma\in M$ is not necessary smooth and bounded away from
zero and one, we shall require that $\pi$ solves the hydrodynamic
equation (\ref{2he}) in some time interval $[0,\tau)$ and $\pi$ is
smooth only on $[\tau,T]\times\bb T^d$.
%
%
\begin{definition}
\label{tnice}
Given $T>0$ and $\gamma\in M$, let $\mc M_{\gamma}^\circ$ be
the collection of the paths $\pi\in\mc M$,
called \textit{nice} paths, satisfying the following conditions:
\begin{longlist}[(iii)]
\item[(i)]the map $(0,T]\times\bb T^d \ni(t,r)\mapsto\pi_t(r)$
is continuous;
\item[(ii)]for each $\delta\in(0,T]$ there exists $\epsilon>0$
such that $\epsilon\le\pi\le1-\epsilon$ in\break $[\delta,T]\times
\bb T^d$;
\item[(iii)]there exists $\tau=\tau_\pi\in(0,T]$ such that, in the
time interval $[0,\tau)$, the path $\pi$ satisfies the energy
estimate and solves (\ref{2he}) while in the time interval
$[\tau,T]$, the map $(t,r) \mapsto\pi_t(r)$ is in
$C^{1,2}([\tau,T]\times\bb T^d)$.
\end{longlist}
\end{definition}

Observe that if $\pi$ belongs to $\mc M_{\gamma}^\circ$, then
$\pi_t\to\gamma$ in $M$ as $t\downarrow0$. Moreover, nice paths
trivially satisfy the energy estimate (\ref{2ee}).

\subsubsection*{Lower bound for nice paths}

Fix $\gamma\in M$, a sequence $\{\eta^N\in\Omega_N\}$ associated
to $\gamma$ and a nice path $\pi\in\mc M_{\gamma}^\circ$. Given
$t\in[\tau_\pi,T]$, regard the first equation in (\ref{2hepsi}) as
a Poisson equation for $\Psi_{\gamma,\pi}$.
In view of Assumption~\ref{tregmob}, item (ii) in
Definition~\ref{tnice} and the bounds (\ref{2bkr}), (\ref{2bsr}),
the symmetric matrix $\sigma(\pi)$ is uniformly elliptic and
continuously differentiable.
Since $\pi$ belongs to $C^{1,2}([\tau_\pi,T]\times\bb T^d)$, by
elliptic regularity, we can solve this equation and get a function,
denoted by $H=H_\pi$, which belongs to $C^{1,2}([\tau_\pi,T]\times
\bb
T^d)$. We understand that for $t=\tau_\pi$ the time derivative
$\partial_t \pi$ stands for the right derivative. We finally extend
$H$ to a piecewise smooth function on $[0,T]\times\bb T^d$ by setting
$H=0$ on $[0,\tau_\pi)\times\bb T^d$. We remark that $H$ can be
discontinuous at $\tau_\pi$.
In any case, $H$ belongs to $\mc H^1(\sigma(\pi))$ and therefore, by
(\ref{2rI}),
%
%
\begin{equation}
\label{6I=}
I(\pi|\gamma)
= \int_{\tau}^T dt\,
\langle\nabla H_t, \sigma(\pi_t) \nabla H_t \rangle
.
\end{equation}

Recall the exponential martingale introduced in (\ref{5expM}) and
let, for the function $H=H_\pi$ constructed above, $\bb
P^{N,E,\pi}_{\eta^N}$ be the probability measures on
$D([0,T];\Omega_N)$ defined by
%
%
\begin{equation}
\label{5cm}
d \bb P^{N,E,\pi}_{\eta^N} = \mc E^N_{H_\pi,\ell, \mathbf
{g}[\delta]} (T)
\,d \bb P^{N,E}_{\eta^N},
\end{equation}
where $\mathbf{g}[\delta]$ is the family of good functions provided
by Lemma~\ref{salsiccia}. Observe that the measures $ \bb
P^{N,E,\pi}_{\eta^N}$ and $\bb P^{N,E}_{\eta^N}$ are equal if
restricted to the time interval $[0,\tau_\p)$.

As we next show, the sequence $\{\bb P^{N,E,\pi}_{\eta^N}\circ
(\pi^N)^{-1}\}$ fulfils the requirements in Lemma~\ref{tglb}.
By, for example,~\cite{KL}, Appendix 1, Proposition 7.3, the probability
$\bb
P^{N,E,\pi}_{\eta^N}$ restricted to the time interval $[\tau,T]$ is
the distribution of the perturbed Kawasaki dynamics with rates
$c_{x,y}^{E,2H,\mathbf{g}[\delta]}$ [see (\ref{delfinodino})].
The construction of the function $H$ and the hydrodynamic limit of the
perturbed Kawasaki dynamics stated in Theorem~\ref{dentini}
(applied with $H=0$, $\mathbf{g}=0$ in the time interval $[0,\tau_\pi)$
and with $H=H_\pi$, $\mathbf{g}= \mathbf{g}[\delta]$ in the time
interval $[\tau_\pi,T]$) then imply
that the sequence $\{\bb P^{N,E,\pi}_{\eta^N} \circ(\pi^N)^{-1} \}$
converges weakly to $\delta_\pi$. The entropic bound (\ref{5lent}) is
an immediate consequence of the next statement.\looseness=1
%
%
\begin{proposition}
\label{tentb}
Fix $T>0$, a vector field $E \in C^1(\bb T^d;\bbR^d)$, a profile
\mbox{$\gamma\in M$}, a sequence $\{\eta^N\in\Omega_N\}$, a nice path
$\pi\in\mc M_{\gamma}^\circ$ and let $\bb P^{N,E,\pi}_{\eta^N}$
be the probability measures on $D([0,T];\Omega_N)$ constructed
above. Then
\[
\limsup_{\delta\downarrow0, \ell\uparrow\infty, N \uparrow
\infty}
\frac1{N^d}
\operatorname{Ent}( \bb P^{N,E,\pi}_{\eta^N} | \bb P^{N,E}_{\eta
^N})
\le I(\pi|\gamma).
\]
\end{proposition}

We premise an elementary lemma on perturbations of Markov chains.
%
%
\begin{lemma}
\label{tpmc}
Let $X$ be a continuous time Markov chain on a finite state space
$E$ with generator $L f(i)=\sum_{j} c_{i,j} [f(j)-f(i)]$ and, given
$T>0$, denote by $\bb P_{i}$ its law in the time interval $[0,T]$
starting from $i\in E$. Fix a function $F\dvtx[0,T]\times E\to
\bb R$, consider the time inhomogeneous Markov chain with
generator
$L_t^F f(i)= \sum_{j} c_{i,j} \exp\{ F(t,j)-F(t,i)\}
[f(j)-f(i)]$ and denote by $\bb P^F_{i}$ its law in the
time interval $[0,T]$ starting from $i\in E$. Then
\[
\operatorname{Ent}({\bb P}^F_{i}| \bb P_{i} )
= {\bb E}^F_{i} \int_0^T dt\, S(t,X(t)),
\]
where ${\bb E}^F_{i}$ is the expectation with respect to
${\bb P}^F_{i}$ and
\[
S(t,i) = \sum_{j} c_{i,j} e^{ F(t,j)-F(t,i)}
\bigl\{ e^{- [ F(t,j)-F(t,i)]} - 1 + F(t,j) -F(t,i)\bigr\}.
\]
\end{lemma}
\begin{pf}
From the explicit expression of the Radon--Nikodym derivative
in~\cite{KL}, Appendix 1, Proposition 7.3, we deduce
\begin{eqnarray*}
&&\operatorname{Ent}({\bb P}^F_{i}| \bb P_{i} )
\\
&&\qquad =
{\bb E}^F_{i} \biggl[ F(T,X(T)) - F(0,X(0))
-\int_0^T dt\, e^{-F(t,X(t))} (\partial_t +L) e^{F(t,X(t))}
\biggr].
\end{eqnarray*}
By using that
$F(t,X(t)) - F(0,X(0)) -\int_0^t ds\, (\partial_s +L_s^F) F(s,X(s))$
is a ${\bb P}^F_{i}$ martingale, straightforward computations yield
the result.
\end{pf}
\begin{pf*}{Proof of Proposition~\ref{tentb}} Set $\t:=\t_\p$.
By definition (\ref{5cm}) [see also (\ref{5expM})] and
Lemma~\ref{tpmc}, a Taylor expansion of the exponential yields
\[
\limsup_{N \uparrow\infty}
\frac1{N^d}
\operatorname{Ent}( \bb P^{N,E,\pi}_{\eta^N} | \bb P^{N,E}_{\eta
^N})
=\limsup_{N \uparrow\infty}
\bb E^{N,E,\pi}_{\eta^N} \int_\tau^T dt\, J_3(t,\eta(t)),
\]
where $J_3$ is defined in Lemma~\ref{texpE}.
In the sequel we shall make use of the super-exponential estimates
in Proposition~\ref{cile} together with Remark~\ref{turn} keeping
the family $\mathbf{g}[\delta]$ fixed. In particular, the first
super-exponential equivalence in (\ref{carbonara25}) holds also with
respect to the measure $\bb P^{N,E,\pi}_{\eta^N}$. Since the function
$J_3$ is bounded uniformly in $N$ and $\ell$, we deduce
\begin{eqnarray*}
&& \limsup_{\ell\uparrow\infty, N \uparrow\infty}
\frac1{N^d}
\operatorname{Ent}( \bb P^{N,E,\pi}_{\eta^N} | \bb P^{N,E}_{\eta
^N})
\\
&&\qquad
=\limsup_{\ell\uparrow\infty, N \uparrow\infty}
\bb E^{N,E,\pi}_{\eta^N} \int_\tau^T dt
\Av_x
V_{\rho=\bar{\eta}_{x, \ell}(t)}\\
&&\hspace*{91pt}\qquad\quad{}\times\Biggl( \sum_{i=1}^d \partial_i H_t \biggl(\frac xN\biggr)
\\
&&\hspace*{124pt}\qquad\quad{} \times
\bigl[ j_{0,e_i}^0 +
L^1_0 g^{(i)}[\delta](\cdot,\upbar{\eta}_{x,\ell}(t))\bigr] \Biggr).
\end{eqnarray*}
In view of the two blocks estimate in Lemma~\ref{duedue},
we can replace above
$\bar{\eta}_{x, \ell}$ with $\tilde{\p}^{N}(\eta)^{\k, a} (x/N)$.
Recalling that the family $\mathbf{g}[\delta]$ is still kept fixed,
the hydrodynamic limit in Theorem~\ref{dentini} yields
\begin{eqnarray*}
&& \limsup_{\ell\uparrow\infty, N \uparrow\infty}
\frac1{N^d}
\operatorname{Ent}( \bb P^{N,E,\pi}_{\eta^N} | \bb P^{N,E}_{\eta
^N})
\\
&&\qquad
= \int_\tau^T dt \int_{\bb T^d} dr\,
V_{\rho=\tilde{\pi}^{\k, a}_t(r)}
\Biggl( \sum_{i=1}^d \partial_i H_t (r)
\bigl[ j_{0,e_i}^0 +
L^1_0 g^{(i)}[\delta](\cdot,\tilde{\pi}^{\k, a}_t(r))\bigr] \Biggr)\\
&&\qquad\quad{}+ \zeta_{\k,a},
\end{eqnarray*}
where $\limsup_{\k\downarrow0, a\downarrow0} \zeta_{\k,a}
=0$.
In view of the identities (\ref{pp1}), (\ref{pp2}), (\ref{pp5}) and
Lemma~\ref{salsiccia}, by taking the limits $a\downarrow0$,
$\k\downarrow0$ and finally $\delta\downarrow0$ we get
\[
\limsup_{\delta\downarrow0, \ell\uparrow\infty, N \uparrow
\infty}
\frac1{N^d}
\operatorname{Ent}( \bb P^{N,E,\pi}_{\eta^N} | \bb P^{N,E}_{\eta
^N})
= \int_\tau^T dt\,
\langle\nabla H_t, \sigma(\pi_t)\nabla H_t\rangle,
\]
which, recalling (\ref{6I=}), concludes the proof.
\end{pf*}

\subsubsection*{Conclusion}

We here conclude the proof of the lower bound (\ref{2lb}) by showing
how to approximate arbitrary paths in $\mc M$ by nice ones. To this end
we need a suitable a priori estimate. Let $\chi_0\dvtx[0,1]\to \bb R_+$
be defined\vspace*{1pt} by $\chi_0(\rho)=\rho(1-\rho)$ and recall the\vadjust{\goodbreak}
bound (\ref{2bchi}). Let $\wwidetilde{\mc Q}\dvtx \mc M \to[0,+\infty]$
be the functional defined by
\[
\wwidetilde{\mc Q} (\pi):= \sup\{\wwidetilde{ \mc Q}_{F}
(\pi), F\in C^{1} ([0,T]\times\bb T^d;\bb R^d) \},
\]
where, given $F\in C^{1}([0,T]\times\bb T^d;\bb R^d)$,
\[
\wwidetilde{\mc Q}_{F} (\pi):=
- 2 \int_0^T dt\, \langle\pi_t, \nabla\cdot F_t
\rangle
- \int_0^T dt\, \langle F_t,\chi_0(\pi_t) F_t \rangle.
\]
By the concavity of $\chi_0$, the functional $\wwidetilde{\mc Q}$ is
lower semicontinuous. Recalling (\ref{2QTF}), we note that $\mc{Q}
(\pi) \le\wwidetilde{\mc Q} (\pi)$.
We next show that the $ \wwidetilde{\mc Q}$ can be bounded by the rate
function $I(\cdot|\gamma)$.
%
%
\begin{lemma}
\label{tQ<I}
Fix $T>0$ and a vector field $E \in C^{1}([0,T]\times\bb
T^d;\bb R^d)$. There exists a constant $C_0=C_0(T,E)$ such that
for any $\gamma\in M$ and $\pi\in\mc M $
\[
\wwidetilde{\mc Q} (\pi) \le C_0 [ I (\pi|\gamma) + 1].
\]
\end{lemma}
\begin{pf}
We can assume $I(\pi|\gamma)<+\infty$. We first observe that in
such a case the linear functional $\ell_{\gamma,\pi}$ in (\ref{2lpH})
can be extended to a linear functional on $\mc H^1(\sigma(\pi))$ and
the supremum in (\ref{2Ig}) can be taken over all $H\in
\mc H^1(\sigma(\pi))$.
Pick a positive function $\phi\in C^2(\bb R)$ uniformly convex and
such that for any $\rho\in[0,1]$ we have $\phi''(\rho) \le
(1/2) \chi(\rho)^{-1}$. Since $\pi$ satisfies the energy estimate,
the function $H=\phi'(\pi)$ is a legal test function in
(\ref{2Ig}). We deduce
\begin{eqnarray*}
I(\pi|\gamma) & \ge & \ell_{\gamma,\pi}( \phi'(\pi)
)
- \int_0^T dt\,
\langle\nabla\phi'(\pi_t),\sigma(\pi_t)\nabla\phi'(\pi
_t)\rangle
\\
& = &\int dr\, [ \phi(\pi_T(r))- \phi(\pi_0(r)) ]
\\
&&{}
-\int_0^T dt\, [
\langle\sigma(\pi_t) E - D(\pi_t) \nabla\pi_t,
\phi''(\pi_t) \nabla\pi_t \rangle
\\
&&\hspace*{52.6pt}{}
+ \langle\phi''(\pi_t) \nabla\pi_t, \sigma(\pi_t)
\phi''(\pi_t) \nabla\pi_t\rangle].
\end{eqnarray*}
Whence, recalling $D=\sigma\chi^{-1} $ and the bounds
(\ref{2bchi}), (\ref{2bkr}), by Schwarz inequality we deduce there exists
$\alpha>0$ and a real $C_\alpha$ such that
\[
\alpha\int_0^T dt\, \langle\nabla\pi_t, \phi''(\pi_t) \nabla
\pi
_t\rangle
\le I (\pi|\gamma) + \int dr\, \phi(\pi_0(r))
+ C_\alpha\int_0^T dt\, \langle E, \sigma(\pi_t) E\rangle.
\]
Since
$\wwidetilde{\mc Q} (\pi)=\int_0^T dt\, \langle\chi_0(\p_t)
\nabla
\p_t,\nabla\p_t \rangle$, the proof is now completed optimizing over
$\phi$.
\end{pf}

In view of Lemma~\ref{tQ<I}, the following proposition can be proven
by adapting the arguments given in~\cite{QRV}, Section 6, or in
\cite{BLM}, Section 5.
%
%
\begin{proposition}
\label{tdens}
Fix $T>0$, a vector field $E \in C^1(\bb T^d;\bbR^d)$ and
$\gamma\in M$. The functional $I(\cdot|\gamma)\dvtx\mc
M \to[0,+\infty]$ has compact level sets, in particular is lower
semicontinuous.\vadjust{\goodbreak}
Moreover, for each $\pi\in\mc M $ such that $I(\pi|\gamma)<
+\infty$ there exists a sequence of nice paths $\{\pi^n\}\subset\mc
M_{\gamma}^\circ$ such that $\pi^n\to\pi$ in $\mc M$ and
$I(\pi^n|\gamma) \to I(\pi|\gamma)$.
\end{proposition}
\begin{pf*}{Proof of Theorem~\ref{tdldp} the lower bound}
Apply Lemma~\ref{tglb}
with $\mc X^\circ$ given by $\mc M_{\gamma}^\circ$ and choose the
perturbation as discussed above. In view of
Proposition~\ref{tentb}, the bound (\ref{5lent}) holds with
$I^\circ$ given by the restriction to $\mc M_{\gamma}^\circ$ of
$I(\cdot|\gamma)$. Finally, Proposition~\ref{tdens} implies that
the functional in (\ref{5rilas}) coincides with $I(\cdot|\gamma)$.
\end{pf*}

\section{The quasi-potential}
\label{s6}

In this section we analyze the variational problems (\ref{2qp}) and
(\ref{genqpdav}) defining the quasi-potential and prove
Theorem~\ref{tqp=f}. Throughout this section we assume that the
vector field $E$ is orthogonally decomposable (recall
Definition~\ref{tbello}) without further mention. We shall only
discuss the case in which assumption (iii) in Theorem~\ref{tqp=f}
holds; the other two cases are actually simpler and the corresponding
details are omitted. We will first consider the problem
(\ref{genqpdav}) and show that it admits a unique minimizer which is
explicitly characterized. From this we then deduce $\widehat
V{}^E_{\upbar{\rho}}=\mathcal F^U_{\upbar{\rho}}$. Finally, we prove the
identity $V^E_{\upbar{\rho}} =\widehat{V}{}^E_{\upbar{\rho}}$.
The characterization of the minimizer will be achieved by exploiting a
time reversal duality analogous to the one in~\cite{FW},
Theorem 4.3.1, and the
convergence, as $t\to+\infty$, of the solution to (\ref{2he*}) to a
stationary solution $\gamma_{\upbar{\rho}}$, $\upbar{\rho}\in[0,1]$.

\subsection*{Time reversal duality}

Given $T\in(0,+\infty]$, we introduce the time reversal $\theta\dvtx\mc
M_{[-T,0]}\to\mc M_{[0,T]}$ as follows. For $\pi\in\mathcal M_{[-T,0]}$
we set $(\theta\pi)_t:=\pi_{-t}$ for any $t\in[0,T]$ such that $-t$ is
a continuity point of $\pi$. This defines the values of $\theta\pi$
apart a countable subset of $[0,T]$ where the values of $\theta\pi$ are
determined by imposing that $\theta\pi\in \mathcal M_{[0,T]}$. For
the\vspace*{1pt} next result, recall (\ref{infdav}) and
(\ref{funinfdav}). Moreover, for $\p\in\mc M_{[0,+\infty)}$ set
$I^E_{[0,+\infty)}(\p):=\lim_{T\to+\infty} I^E_{[0,T]}(\p)$.
%
%
\begin{theorem}
\label{tonssiminf}
Fix $\upbar{\rho}\in[0,1]$.
For each $\pi\in\mathcal M_{(-\infty,0]}(\upbar{\rho})$ it holds
%
%
\begin{equation}
\label{onssiminf}
I^{-\nabla U+\wwidetilde{E}}_{(-\infty,0]}(\pi)=\mathcal F^U_{\upbar
{\rho}}(\pi_0)+I^{-\nabla
U-\wwidetilde{E}}_{[0,+\infty)}(\theta\pi).
\end{equation}
\end{theorem}

Of course, the identity (\ref{onssiminf}) means that either both sides
are infinite or both sides are finite and the respective values
coincide. In order\vspace*{1pt} to prove this result, we need to
introduce some more notation. The norms in $L^2(\bb T^d,dr)$ and in the
standard Sobolev space $W^{1,2}(\bb T^d,dr)$ are, respectively, denoted
by $\|\cdot\|_{L^2}$ and $\|\cdot\|_{W^{1,2}}$. Fix $T_1<T_2$. By
choosing a test function independent on the space variable, we easily
deduce that $I^E_{[T_1,T_2]}(\pi)<+\infty$ implies that the total mass
$\int dr\, \pi_t(r)$ is constant in time. Given $\upbar{\rho}\in[0,1]$,
we then set $\mc M_{[T_1,T_2]}(\upbar{\rho }):=
D([T_1,T_2];M(\upbar{\rho}))$ [recall (\ref{7Mbr})]. Also let $M^\circ
_{[T_1,T_2]}(\upbar{\rho})\subset\mathcal M_{[T_1,T_2]}(\upbar{\rho})$
be the collection\vspace*{1pt} of paths $\pi\in \mathcal M_{[T_1,T_2]}$ satisfying
the following conditions: (i) there exists $\epsilon>0$ such that
$\epsilon\leq\pi\leq1-\epsilon$, (ii) the map $[T_1,T_2]\times\mathbb
T^d\to\pi_t(r)$ belongs to $C^{1,2}([T_1,T_2]\times\bb T^d)$. Note that
if $\pi\in\mathcal M_{[-T,0]}^\circ(\upbar{\rho})$ then
$\theta\pi\in\mathcal M^\circ_{[0,T]}(\upbar{\rho})$. Given\vspace*{2pt} $\gamma\in
M$, we denote by $\mathcal M^\circ_{[T_1,T_2],\gamma}$ the collection
of nice paths, as in Definition~\ref{tnice}, in $\mathcal
M_{[T_1,T_2]}$. We observe\vspace*{1pt} that if $\pi$ belongs to $\mathcal
M_{[T_1,T_2]}^\circ(\upbar{\rho})$ for some $\upbar{\rho}\in[0,1]$ then
the linear functional $\ell _\pi$ in (\ref{ldav}) can be rewritten as
%
%
\begin{equation}\label{libia}
\ell^E_{\pi}(H)=\int_{T_1}^{T_2} dt\,
\langle\partial_t\pi_t+\nabla\cdot
[ \sigma(\pi_t) E - D(\pi_t) \nabla\pi_t ],
H_t\rangle,
\end{equation}
where we also included in the notation the dependence on the driving
field~$E$.

The next elementary result will be the key point in the proof of
Theorem~\ref{tonssiminf}. Recall (\ref{2frU}) and, given
$\upbar{\rho}\in(0,1)$, let $g_{\upbar{\rho}}\dvtx\bb T^d\times
[0,1]\to\bb
R$ be the function defined by
%
%
\begin{equation}
\label{7grr}
g_{\upbar{\rho}}(r,\rho):=
\frac{\partial}{\partial\rho} f^U_{\upbar{\rho}}(r,\rho)
= f'(\rho) - f'(\gamma_{\upbar{\rho}}(r)).
\end{equation}

%
\begin{lemma}
\label{thj}
Fix $\upbar{\rho}\in(0,1)$ and $\rho\in C^2(\bb T^d;(0,1))$. Let
$G=G_\rho\dvtx\bb T^d\to\bb R$ be the function defined by
$G(r):=g_{\upbar{\rho}}(r,\rho(r))$. Then
\[
\langle\nabla\cdot
[ \sigma(\rho) E - D(\rho) \nabla\rho], G
\rangle
-\langle\nabla G,\sigma(\rho) \nabla G\rangle= 0.
\]
\end{lemma}
%
%
\begin{remark}
\label{trhj}
Recall that the vector field $E$ satisfies (\ref{2E=}). The
statement of Lemma~\ref{thj} does not depend on the divergenceless
part $\wwidetilde E$; in particular, it holds also if $E$ is replaced
by the vector field $-\nabla U -\wwidetilde E$.
\end{remark}
\begin{pf*}{Proof of Lemma~\ref{thj}}
By the definition of $G$ and (\ref{2gbar}),
\[
\nabla G (r) = f''(\rho(r)) \nabla\rho(r)
-f''(\gamma_{\upbar{\rho}}(r)) \nabla\gamma_{\upbar{\rho}}(r)
= f''(\rho(r)) \nabla\rho(r) +\nabla U(r).
\]
Recalling (\ref{2chi}), (\ref{2D}) and that we assumed $\sigma$ to
be a multiple of the identity, the statement of the lemma is therefore
equivalent to
\[
\langle\sigma(\rho) E + \sigma(\rho) \nabla U,
f''(\rho) \nabla\rho+\nabla U \rangle= 0.
\]
Recall that $E =-\nabla U +\wwidetilde E$. Using again (\ref{2chi}) and
(\ref{2D}), the above equation holds if and only if
\[
\langle\wwidetilde E, D(\rho) \nabla\rho\rangle
+\langle\sigma(\rho) \wwidetilde E, \nabla U\rangle=0
.
\]
Since $D$ is also a multiple of the identity, the first term above
vanishes because $\wwidetilde E$ is divergenceless. Finally, as we
assumed $\wwidetilde E(r)\cdot\nabla U(r) =0$ for any $r\in\bb T^d$;
also the second term above vanishes.
\end{pf*}
%
%
\begin{lemma}
\label{strasim}
Fix $\upbar{\rho}\in(0,1)$ and $T>0$. For each $H\in
C^1([-T,0]\times
\mathbb T^d)$ and each $\pi\in\mathcal
M_{[-T,0]}^\circ(\upbar{\rho})$ it holds [recall (\ref{libia})]
%
%
\begin{eqnarray}
\label{corroasilo}
&& \ell^{-\nabla U+\wwidetilde{E}}_{\pi}(H)
-\int_{-T}^0 dt\, \langle\nabla H_t, \sigma(\pi_t)\nabla
H_t\rangle
\nonumber\\
&&\qquad
=\mathcal F^U_{\upbar{\rho}}(\pi_0)-\mathcal F^U_{\upbar{\rho
}}(\pi_{-T})
+ \ell^{-\nabla U-\wwidetilde{E}}_{\theta\pi}(-\theta\wwidetilde{H})
\\
&&\qquad\quad{}
-\int_{0}^T dt\,
\langle\nabla(\theta\wwidetilde{H})_t,
\sigma( (\theta\pi)_t)
\nabla(\theta\wwidetilde{H})_t \rangle,\nonumber
\end{eqnarray}
where $\wwidetilde{H}\equiv\wwidetilde{H}_t(r)$ is given by
%
%
\begin{equation}
\label{tildeHd}
\wwidetilde{H} = H
-[ f'(\pi)-f'(\gamma_{\upbar{\rho}}) ].
\end{equation}
\end{lemma}
\begin{pf}
The proof follows by a direct computation. As in Lemma~\ref{thj},
we call $G\dvtx[-T,0]\times\mathbb T^d\to\mathbb R$ the function
$G_t(r):=f'(\pi_t(r))-f'(\gamma_{\upbar{\rho}}(r))$. We start from
the left-hand side of (\ref{corroasilo}) and add and subtract
$\ell^{-\nabla U+\wwidetilde{E}}_{\pi} (G)$. We obtain the sum of three
terms: the first one is
%
%
\begin{eqnarray}
\label{primod}
&&
\int_{-T}^{0} dt\,
\langle\partial_t\pi_t +
\nabla\cdot[ \sigma(\pi_t)
(-\nabla U+\wwidetilde E) - D(\pi_t) \nabla\pi_t ]
, H_t-G_t\rangle
\nonumber\\[-8pt]\\[-8pt]
&&\qquad{}
-\int_{-T}^0 dt\,
\langle\nabla H_t,\sigma(\pi_t) \nabla H_t\rangle
+\int_{-T}^0 dt\,
\langle\nabla G_t,\sigma(\pi_t) \nabla G_t\rangle,\nonumber
\end{eqnarray}
the second one is
%
%
\begin{equation}
\label{secondod}
\int_{-T}^0 dt\,
\langle\partial_t\pi_t,G_t\rangle
= \mathcal{F}^U_{\upbar{\rho}}(\pi_0)
-\mathcal F^U_{\upbar{\rho}}(\pi_{-T})
\end{equation}
and the third one is
%
%
\begin{eqnarray}
&&\int_{-T}^0 dt\,
\langle\nabla\cdot[
\sigma(\pi_t) (-\nabla U+\wwidetilde E)
-D(\pi_t) \nabla\pi_t]
, G_t\rangle\nonumber\\[-8pt]\\[-8pt]
&&\qquad{}-\int_{-T}^0 dt\,
\langle\nabla G_t,\sigma(\pi_t) \nabla G_t\rangle.\nonumber
\end{eqnarray}
From Lemma~\ref{thj} it immediately follows that this last term vanishes.

We now elaborate the first term (\ref{primod}). Using
(\ref{tildeHd}), that is, $\wwidetilde H=H-G$, and performing an
integration by parts, it can be rewritten as
%
%
\begin{eqnarray}
\label{conmbu}
&&\int_{-T}^{0} dt\,
\langle\partial_t\pi_t
+\nabla\cdot[ \sigma(\pi_t)
(-\nabla U+\wwidetilde E) - D(\pi_t) \nabla\pi_t]
, \wwidetilde{H}_t\rangle
\nonumber\\[-8pt]\\[-8pt]
&&\qquad{}
-\int_{-T}^0 dt\,
\langle\nabla\wwidetilde{H}_t,
\sigma(\pi_t) \nabla\wwidetilde H_t \rangle
+2 \int_{-T}^0 dt\,
\langle\nabla\cdot[
\sigma(\pi_t)\nabla G_t ], \wwidetilde H_t \rangle.\nonumber
\end{eqnarray}
From the Einstein relation (\ref{2D}) and (\ref{2gbar}) we obtain
$\sigma(\pi) \nabla G = D(\pi) \nabla\pi+\sigma(\pi) \nabla U$
which, inserted into (\ref{conmbu}), gives
\begin{eqnarray*}
&& \int_{-T}^{0} dt\,
\langle\partial_t\pi_t+
\nabla\cdot[ \sigma(\pi_t) (\nabla U+\wwidetilde E)
+ D(\pi_t) \nabla\pi_t],
\wwidetilde{H}_t\rangle
\\
&&\qquad{}
-\int_{-T}^0 dt\,
\langle\nabla\wwidetilde H_t,
\sigma(\pi_t) \nabla\wwidetilde H_t\rangle.
\end{eqnarray*}
Performing a change of variable in the time integral and adding
(\ref{secondod}) we obtain the right-hand side of (\ref{corroasilo}).
\end{pf}

From Lemma~\ref{strasim} we deduce the time reversal duality for
bounded intervals.
%
%
\begin{lemma}
\label{tonssim}
Fix $\upbar{\rho}\in[0,1]$ and $T>0$.
For each $\pi\in\mathcal M_{[-T,0]}(\upbar{\rho})$ it holds
%
%
\begin{equation}
\label{onssim}
I^{-\nabla U+\wwidetilde{E}}_{[-T,0]}(\pi)=\mathcal F^U_{\upbar
{\rho}}(\pi_0)-\mathcal F^U_{\upbar{\rho}}(\pi_{-T})+I^{-\nabla
U-\wwidetilde{E}}_{[0,T]}(\theta\pi).
\end{equation}
\end{lemma}
\begin{pf}
Since the statement is trivial when $\upbar{\rho}=0$ or
$\upbar{\rho}=1$, we can assume $\upbar{\rho}\in(0,1)$. First consider
the case $\pi\in\mathcal M^0_{[-T,0]}(\upbar{\rho})$; then
the correspondence $H\leftrightarrow-\theta\wwidetilde H$ [see
(\ref{tildeHd})] define a bijection between $C^1([-T,0]\times\mathbb
T^d)$ and $C^1([0,T]\times\mathbb T^d)$. From (\ref{corroasilo})
we deduce
\begin{eqnarray*}
I^{-\nabla U+\wwidetilde{E}}_{[-T,0]}(\pi)
&=&\sup_H
\biggl\{ \ell^{-\nabla U+\wwidetilde{E}}_{\pi}(H)
-\int_{-T}^0 dt\, \langle\nabla H_t, \sigma(\pi_t)\nabla
H_t\rangle
\biggr\}
\\
&=& \mathcal F^U_{\upbar{\rho}}(\pi_0)
-\mathcal F^U_{\upbar{\rho}}(\pi_{-T})
\\
&&{}
+ \sup_H \biggl\{
\ell^{-\nabla U-\wwidetilde{E}}_{\theta\pi}(-\theta\wwidetilde{H})
-\int_{0}^T dt\,
\langle\nabla(\theta\wwidetilde{H})_t,
\sigma( (\theta\pi)_t) \nabla(\theta\wwidetilde{H})_t \rangle
\biggr\}
\\
&=& \mathcal F^U_{\upbar{\rho}}(\pi_0)
-\mathcal F^U_{\upbar{\rho}}(\pi_{-T})
+I^{-\nabla U-\wwidetilde{E}}_{[0,T]}(\theta\pi).
\end{eqnarray*}
Now consider\vspace*{1pt} an arbitrary path
$\pi\in\mathcal{M}_{[-T,0]}(\upbar {\rho})$ such that $I^{-\nabla
U+\wwidetilde{E}}_{[-T,0]}(\pi)<+\infty$. By Proposition~\ref{tdens},
there\vspace*{-1pt} exists a sequence $\{\pi^n\} \subset
\mathcal{M}^\circ_{[-T,0],\pi_{-T}}$ such that $\pi^n\to\pi$ in
$\mathcal M_{[-T,0]}$ and $I^{-\nabla
U+\wwidetilde{E}}_{[-T,0]}(\pi^n)\to I^{-\nabla
U+\wwidetilde{E}}_{[-T,0]}(\pi)$; in particular,\vspace*{1pt}
$\{\pi^n\}\subset\mathcal{M}_{[-T,0]}(\upbar{\rho})$. Let $\tau^n>0$ be
the time such that $\pi^n$ solves (\ref{2he}) in the time interval
$[-T,-T+\tau^n]$. From the result for nice paths we deduce that for
each $n$
%
%
\begin{eqnarray}
\label{onsn}
I^{-\nabla U+\wwidetilde{E}}_{[-T,0]}(\pi^n)&=&I^{-\nabla
U+\wwidetilde{E}}_{[-T+\tau^n,0]}(\pi^n)\nonumber\\[-8pt]\\[-8pt]
&=&\mathcal F^U_{\upbar{\rho}}(\pi^n_0)-\mathcal F^U_{\upbar
{\rho}}(\pi^n_{-T+\tau^n})+I^{-\nabla
U-\wwidetilde{E}}_{[0,T-\tau^n]}(\theta\pi^n),\nonumber
\end{eqnarray}
where the second identity follows from the fact that the restriction
of $\pi^n$ to the time interval $[-T+\tau^n,0]$ belongs to $\mathcal
M^0_{[-T+\tau^n,0]}(\upbar{\rho})$.\vadjust{\goodbreak}
It is easy to see that we can always choose $\pi^n$ in such a way that
$\lim_{n} \tau^n=0$. This implies that
$\lim_{n}\|\pi^n_{-T+\tau^n}-\pi_{-T}\|_{L^2}=0$.
Since $\mc{F}^U_{\upbar{\rho}}$ is continuous with respect to the $L^2$
topology, we get
$\lim_{n}\mathcal F^U_{\upbar{\rho}}(\pi^n_{-T+\tau^n})=
\mathcal{F}^U_{\upbar{\rho}}(\pi_{-T})$.
By using the lower semicontinuity of $\mc{F}^U_{\upbar{\rho}}$ on $M$
and of $I_{[0,T]}^{-\nabla U-\wwidetilde{E}}$ on $\mc{M}_{[0,T]}$, from
(\ref{onsn}) we then deduce that for each $S\in(0,T)$ it holds
\begin{eqnarray*}
I^{-\nabla U+\wwidetilde{E}}_{[-T,0]}(\pi)
&=& \lim_{n\to+\infty} I^{-\nabla U+\wwidetilde{E}}_{[-T,0]}(\pi^n)
\\
&\ge& \liminf_{n\to+\infty}
\bigl\{ \mathcal F^U_{\upbar{\rho}}(\pi^n_0)
-\mathcal F^U_{\upbar{\rho}}(\pi^n_{-T+\tau^n})
+I^{-\nabla U-\wwidetilde{E}}_{[0,T-\tau^n]}(\theta\pi^n)\bigr\}
\\
&\ge& \mathcal
F^U_{\upbar{\rho}}(\pi_0)-\mathcal F^U_{\upbar{\rho}}(\pi_{-T})
+I^{-\nabla U-\wwidetilde{E}}_{[0,S]}(\theta\pi).
\end{eqnarray*}
Observing that $\theta\pi$ is necessarily continuous, we can take the
limit $S\uparrow T$ and deduce
\[
I^{-\nabla U+\wwidetilde{E}}_{[-T,0]}(\pi)\geq\mathcal F^U_{\upbar
{\rho}}(\pi_0)-\mathcal F^U_{\upbar{\rho}}(\pi_{-T})+I^{-\nabla
U-\wwidetilde{E}}_{[0,T]}(\theta\pi).
\]
The proof is now completed by exchanging the roles of $\pi$ and
$\theta\pi$.
\end{pf}

Recall that the set $\mc M_{(-\infty,0]}(\upbar{\rho})$ has been
defined in (\ref{infdav}) by requiring that $\pi_t\to\gamma_{\upbar
{\rho}}$ in $M$ as $t\to-\infty$. The next lemma states that if
$I^E_{(-\infty,0]}(\pi)<+\infty$, the above convergence actually
takes place also with respect to the $L^2$ topology. The proof,
which is omitted, is achieved by repeating the arguments of
\cite{BDGJLn+1}, Lemma 5.2, in the present setting.
%
%
\begin{lemma}
\label{liam}
Fix $\upbar{\rho}\in[0,1]$ and a path
$\pi\in\mathcal{M}_{(-\infty,0]}(\upbar{\rho})$ with finite rate,
that is, satisfying
$I^E_{(-\infty,0]}(\pi)<+\infty$. Then $\lim_{t\to-\infty}
\|\pi_t-\gamma_{\upbar{\rho}}\|_{L^2}=0$. Moreover, there exists a
sequence $T_n\to-\infty$ such that $\lim_{n\to\infty}
\|\pi_{T_n}-\gamma_{\upbar{\rho}}\|_{W^{1,2}}=0$.
\end{lemma}
\begin{pf*}{Proof of Theorem~\ref{tonssiminf}}
Consider the case in which $\pi\in\mathcal
M_{(-\infty,0]}(\upbar{\rho})$ is such that $I^{-\nabla U+\wwidetilde
E}_{(-\infty,0]}(\pi)<+\infty$. From Lemma~\ref{liam} and the
continuity of $\mathcal F^U_{\upbar{\rho}}$ in $L^2$ we deduce
$\lim_{T\to+\infty}\mathcal{F}^U_{\upbar{\rho}}(\pi_{-T})=0$.
Therefore, (\ref{onssiminf}) follows from (\ref{onssim}) by taking
the limit $T\to+\infty$. In particular, if $I^{-\nabla
U+\wwidetilde{E}}_{(-\infty,0]}(\pi)<+\infty$ then also $I^{-\nabla
U-\wwidetilde{E}}_{[0,+\infty)}(\theta\pi)<+\infty$. The proof is
now completed by exchanging the roles of $\pi$ and $\theta\pi$.
\end{pf*}

\subsection*{Convergence to a stationary solution}

We next discuss the asymptotic behavior of the solutions to the
equation (\ref{2he*}). Observe that, since $\nabla U(r)\cdot
\wwidetilde E(r)=0$ for any $r\in\bb T^d$, for each $\upbar{\rho}\in
[0,1]$ the function $\gamma_{\upbar{\rho}}$ defined in (\ref{2gbar})
is also a stationary solution to (\ref{2he*}). While the following
result is stated for the equation (\ref{2he*}), it holds also for
the hydrodynamic equation (\ref{2he}). As we need to emphasize the
dependence on the initial condition, given $\rho\in M$, we denote by
$v_t(\rho)\equiv v_t(r;\rho)$ the solution to (\ref{2he*}) with
initial condition $\rho$.\vadjust{\goodbreak}

%
\begin{theorem}
\label{tcfhe*}
Fix $\upbar{\rho}\in[0,1]$ and let $v_t(\rho)$ be the solution to
(\ref{2he*}). Then,
\[
\lim_{t\to+\infty}
\sup_{\rho\in M(\upbar{\rho})}
\| v_t(\rho) -\gamma_{\upbar{\rho}}\|_{L^2} = 0.
\]
Moreover, for each $\rho\in M(\upbar{\rho})$ there exists a sequence
$T_n\to+\infty$ such that
$\lim_{n\to\infty} \| v_{T_n}(\rho) -\gamma_{\upbar{\rho}}\|
_{W^{1,2}}=0$.
\end{theorem}

The proof of this result will be achieved by showing that
$\mc{F}^U_{\upbar{\rho}}$ is a Lyapunov functional for the flow
defined by
(\ref{2he*}) and using comparison arguments.

%
\begin{lemma}
\label{tmgr}
If $0<\upbar{\rho}_1<\upbar{\rho}_2<1$ then $0< \gamma_{\upbar
{\rho}_1}
<\gamma_{\upbar{\rho}_2}<1$. Moreover,
if $\upbar{\rho}\uparrow1$ or $\upbar{\rho}\downarrow0$ then
$\gamma_{\upbar{\rho}}\uparrow1$ or $\gamma_{\upbar{\rho
}}\downarrow
0$, respectively.
\end{lemma}
\begin{pf}
Recall that $f'\dvtx(0,1)\to\bb R$ is strictly increasing and denote by
$(f')^{-1}\dvtx\bb R \to(0,1)$ its inverse. Then the map
$\upbar{\rho}\mapsto\alpha(\upbar{\rho})$ in (\ref{2gbar}) is defined
by requiring
\[
\int_{\bb T^d} dr\, (f')^{-1} \bigl( -U(r) +\alpha(\upbar{\rho
})\bigr)
=\upbar{\rho}.
\]
In particular, since $(f')^{-1}$ is strictly increasing,
the map $\upbar{\rho}\mapsto\alpha(\upbar{\rho})$ is strictly increasing.
Again by the strict monotonicity of $(f')^{-1}$, the first statement
follows. To prove the second, it is enough to notice that if
$\upbar{\rho}\uparrow1$, respectively, $\upbar{\rho}\downarrow0$, then
$\alpha(\upbar{\rho})\uparrow+\infty$, respectively,
$\alpha(\upbar{\rho})\downarrow-\infty$.
\end{pf}
%
%
\begin{lemma}
\label{tvbg}
Let $v\dvtx [0,+\infty) \times\bb T^d\to[0,1]$ be the solution to
(\ref{2he*}) and assume there exist $0<\upbar{\rho}_1<\upbar{\rho}_2<1$
such that $\gamma_{\upbar{\rho}_1}\le\rho\le\gamma_{\upbar{\rho}_2}$.
Then for any $t\ge0$ we have $\gamma_{\upbar{\rho}_1} \le v_t(\rho
) \le
\gamma_{\upbar{\rho}_2}$.
\end{lemma}
\begin{pf}
By classical results for uniformly parabolic equation, $v$ is smooth
on $(0,+\infty)\times\bb T^d$. Let $w\dvtx[0,+\infty)\times\bb T^d\to
[0,1]$ be defined by $w_t(r):= \gamma_{\upbar{\rho}_1}(r) -
v_t(r;\rho)$ and observe that, by hypotheses, $w_0 < 0$. Recall the
bounds (\ref{2bchi}), (\ref{2bkr}), (\ref{2bsr}), definition
(\ref{2D}) and that $\sigma$ is a multiple of the identity. Since
$\gamma_{\upbar{\rho}_1}$ is a stationary solution to (\ref{2he*}),
it is simple to check that $w$ solves the linear parabolic equation
\[
\partial_t w = a \Delta w + b \cdot\nabla w + c w
\]
for some continuous functions $a,b,c$ on $[0,+\infty)\times\bb T^d$.
Moreover, $a$ is uniformly positive on $[0,+\infty)\times\bb T^d$.
By Theorem 3.7 and the remark (ii) following it in~\cite{PW}, we then
deduce $w_t\le0$ for any $t\ge0$.
The inequality $v_t(\rho)\le\gamma_{\upbar{\rho}_2}$ is proven by
the same
argument.
\end{pf}
%
%
\begin{lemma}
\label{tvbd}
Fix $\upbar{\rho}\in(0,1)$.
For each $t_0>0$ there exists $\delta\in(0,1/2)$ such that for any
$t\ge t_0$ and any $\rho\in M(\upbar{\rho})$ it holds
$\delta\le v_t(\rho) \le1-\delta$.
\end{lemma}
\begin{pf}
Let $\rho\in M$ and consider a sequence $\{\rho^n\}\subset
M$ converging to $\rho$ in $M$.
By standard parabolic regularity, for each $t>0$ the sequence of
functions on $\bb T^d$ given by $v_t(\cdot;\rho^n)$ converges
uniformly to $v_t(\cdot;\rho)$. Set
\[
\delta_0:= \inf\{v_{t_0}(r;\rho), r\in\bb T^d,
\rho\in M(\upbar{\rho}) \}.
\]
By the compactness of $M(\upbar{\rho})$ and the above continuity,
there exists $\rho^*\in M(\upbar{\rho})$ such that
$\delta_0 = \inf\{v_{t_0}(r;\rho^*), r\in\bb T^d\}$.
Since $\rho^*$ is not identically equal to zero, by applying
Theorem 3.7 and the remark (ii) following it in~\cite{PW}, we deduce
$\delta_0>0$.
By Lemma~\ref{tmgr}, there exists $\upbar{\rho}_1\in(0,1)$ such
that $\gamma_{\upbar{\rho}_1}\le\delta_0$. Setting
$\delta:=\min\{ \gamma_{\upbar{\rho}_1}(r), r\in\bb T^d\}
$ and
using Lemma~\ref{tvbg} we deduce that for any $t\ge t_0$ we have
$v_t(\rho) \ge\gamma_{\upbar{\rho}_1} \ge\delta$.

The uniform upper bound is proven by the same argument.
\end{pf}
\begin{pf*}{Proof of Theorem~\ref{tcfhe*}}
Since the statement is trivial when $\upbar{\rho}=0$ or
$\upbar{\rho}=1$, we assume $\upbar{\rho}\in(0,1)$. Recall that the
functional $\mc F^U_{\upbar{\rho}}\dvtx M\to\bb[0,+\infty)$ has been
defined in (\ref{2Fbr}). In view of the uniform convexity of the free
energy $f$, it is\vspace*{1pt} simple to show that for each
$\upbar{\rho}\in (0,1)$ the functional $\mc F^U_{\upbar{\rho}}(\cdot)$
is equivalent to $|\cdot-\gamma_{\upbar{\rho}}|_{L^2}^2$. Namely, there
exists a constant $C_0=C_0(\upbar{\rho})>0 $ such that for any
$\gamma\in M(\upbar{\rho})$ we have
%
%
\begin{equation}
\label{7F2}
\frac1{C_0} \|\gamma-\gamma_{\upbar{\rho}}\|_{L^2}^2
\le\mc F^U_{\upbar{\rho}}(\gamma)
\le{C_0} \|\gamma-\gamma_{\upbar{\rho}}\|_{L^2}^2.
\end{equation}

By parabolic regularity, the function $v(\rho)$ is smooth on
$(0,+\infty)\times\bb T^d$.
Using Remark~\ref{trhj} we then deduce that for $t>0$ it holds
%
%
\begin{equation}
\label{7lf}
\frac{d}{dt} \mc F^U_{\upbar{\rho}} ( v_t(\rho))
= - \langle\nabla G_t, \sigma(v_t(\rho)) \nabla
G_t \rangle,
\end{equation}
where, recalling (\ref{7grr}), $G$ is the function defined by
$G_t(r)=g_{\upbar{\rho}}(r;v_t(r;\rho))$. In particular,
$\mc F^U_{\upbar{\rho}}$ is a Lyapunov functional for both the
flows defined by (\ref{2he*}) and (\ref{2he}).
Given $\epsilon>0$ set
\[
A_\epsilon:= \{\gamma\in M(\upbar{\rho})\dvtx
\mc F^U_{\upbar{\rho}}(\gamma) <\epsilon\}
\]
and let $\tau_\epsilon(\rho):=\inf\{t>0\dvtx v_t(\rho) \in
A_\epsilon\}\in[0,+\infty]$. In view of (\ref{7F2}) and
(\ref{7lf}), the proof of the theorem is completed once we show
that for each $\epsilon>0$ the hitting time $\tau_\epsilon(\rho)$
is bounded
uniformly for $\rho\in M(\upbar{\rho})$.

Given $\upbar{\rho}\in(0,1)$ and $\delta\in(0,1/2)$ set
\[
\widehat{M}_\delta(\upbar{\rho}):= \biggl\{\gamma\in L^2(\bb
T^d,dr), \delta\le\gamma\le1-\delta, \int dr\,
\gamma(r)=\upbar{\rho}\biggr\},
\]
which is a closed subset of $L^2(\bbT^d,dr)$ that we consider endowed
with the relative topology. Fix $t_0>0$ and observe that if we choose
$\delta$ as in Lemma~\ref{tvbd} then this lemma implies that
$v_t(\rho)\in\widehat{M}_\delta(\upbar{\rho})$ for any $t\ge t_0$ and
$\rho\in M(\upbar{\rho})$. Moreover,\vspace*{1pt} the functional $\mc
F^U_{\upbar{\rho}}$ is continuous on
$\widehat{M}_\delta(\upbar{\rho})$. Given\vspace*{-1pt}
$\gamma\in\widehat{M}_\delta(\upbar{\rho})$ let $G_\gamma\dvtx\bb
T^d\to\bb R$ be the function defined by $G_\gamma(r) =
g_{\upbar{\rho}}(r,\gamma(r))$. Let also $\mc R_{\upbar{\rho}}\dvtx
\widehat{M}_\delta(\upbar{\rho}) \to[0,+\infty]$ be the lower
semicontinuous functional defined by
\[
\mc R_{\upbar{\rho}}(\gamma):=\sup_{F}
\{ - 2 \langle\nabla\cdot F, G_\gamma\rangle- \langle
F,F\rangle\},
\]
where the supremum is over all $F\in C^1(\bb
T^d;\bb R^d)$. If $\mc R_{\upbar{\rho}}(\gamma) <+\infty$ then
$G_\gamma$ belongs to the Sobolev space $W^{1,2}(\bb T^d,dr)$ and
$\mc R_{\upbar{\rho}}(\gamma) = \langle\nabla G_\gamma, \nabla
G_\gamma\rangle$.
In particular, by Sobolev embedding and elementary estimates, the
functional $\mc R_{\upbar{\rho}}$ has compact level sets.
It is finally straightforward to check that $\mc
R_{\upbar{\rho}}(\gamma)=0$ if and only if
$\gamma=\gamma_{\upbar{\rho}}$. Recalling (\ref{7F2}), we deduce
that for each $\epsilon>0$ and $\delta>0$
\[
c_\epsilon:= \inf\{\mc R_{\upbar{\rho}}(\gamma),
\gamma\in\widehat{M}_\delta(\upbar{\rho})\setminus A_\epsilon
\} >0.
\]

Given $t_0>0$, let $\delta\in(0,1/2)$ be as in Lemma~\ref{tvbd}
and set $m =\min\{\sigma(u), \delta\le u\le1-\delta\} >0$. Set
also $K =\sup\{\mc F^U_{\upbar{\rho}}(\gamma), \gamma\in
M(\upbar{\rho})\}<+\infty$. We are now ready to conclude the proof.
If $\tau_\epsilon(\rho)<t_0$ there is nothing to prove, otherwise,
in view of Lemma~\ref{tvbd} and (\ref{7lf}), we deduce that for
each $\epsilon>0$, $\rho\in M(\upbar{\rho})$ and $t\ge t_0$
\begin{eqnarray*}
K &\ge& \mc F^U_{\upbar{\rho}}(v_{t_0}(\rho))
= \mc F^U_{\upbar{\rho}}\bigl( v_{t\wedge\tau_\epsilon(\rho
)}(\rho)\bigr)
+\int_{t_0}^{t\wedge\tau_\epsilon(\rho)} ds\,
\langle\nabla G_s, \sigma(v_s(\rho)) \nabla G_s
\rangle
\\
&\ge& m \int_{t_0}^{t\wedge\tau_\epsilon(\rho)} ds\,
\mc R_{\upbar{\rho}}( v_s(\rho) )
\ge m c_\epsilon[ t\wedge\tau_\epsilon(\rho) -t_0].
\end{eqnarray*}
By taking the limit $t\uparrow+\infty$, the previous bound yields
$\sup\{\tau_\epsilon(\rho), \rho\in M(\upbar{\rho})\} <
+\infty$.

It remains to prove the second statement.
By the regularity and uniform convexity of the free energy $f$, it
is simple to check that for each $\upbar{\rho}\in(0,1)$ and
$\delta\in(0,1/2)$ there exists a real $C_1=C_1(\upbar{\rho},\delta)$
such that for any $\gamma\in\widehat{M}_\delta(\upbar{\rho})$
\[
\| \gamma-\gamma_{\upbar{\rho}} \|_{W^{1,2}}^2
\le C_1 [ \mc R_{\upbar{\rho}}(\gamma) +
\| \gamma-\gamma_{\upbar{\rho}} \|_{L^{2}}^2 ].
\]
Fix $t_0>0$ and let $\delta$ be as in Lemma~\ref{tvbd}. From
(\ref{7lf}) we deduce that for any $\rho\in M(\upbar{\rho})$ and
any $t\ge t_0$
\[
\mc F^U_{\upbar{\rho}}( v_t(\rho))
+ m \int_{t_0}^t ds\, \mc R_{\upbar{\rho}}( v_s(\rho))
\le\mc F^U_{\upbar{\rho}}(v_{t_0}(\rho)) \le K.
\]
In particular, there exists a sequence $T_n\to+\infty$ such that
$\mc{R}_{\upbar{\rho}}( v_{T_n}(\rho))\to0$.
\end{pf*}

\subsection*{Conclusion}
We next conclude the proof of the identity between the quasi-potential
and the functional $\mc F^U_{\upbar{\rho}}$ and the characterization of
the minimizer for (\ref{genqpdav}).\vadjust{\goodbreak}
\begin{pf*}{Proof of Theorem~\ref{tqp=f} (the identity
$\widehat{V}{}^E_{\upbar{\rho}} =\mathcal F^U_{\upbar{\rho}}$)}
For $\upbar{\rho}\in[0,1]$ and $\rho\in M(\upbar{\rho})$, let
$\pi\in\mathcal M_{(-\infty,0]}(\upbar{\rho})$ be such that $\pi
_0=\rho$.
From Theorem~\ref{tonssiminf} we get
%
%
\begin{equation}\label{zuppa}
I^{-\nabla U+\wwidetilde{E}}_{(-\infty,0]}(\pi)=\mathcal F^U_{\upbar
{\rho}}(\rho)+I^{-\nabla U-\wwidetilde{E}}_{[0,+\infty)}(\theta\pi).
\end{equation}
Since $I^{-\nabla U-\wwidetilde{E}}_{[0,+\infty)}\geq0$, we deduce
$I^{-\nabla U+\wwidetilde{E}}_{(-\infty,0]}(\pi)\geq
\mathcal{F}^U_{\upbar{\rho}}(\rho)$. The lower bound
$\widehat V{}^E_{\upbar{\rho}}(\rho)\geq\mathcal F^U_{\upbar{\rho
}}(\rho)$ follows.

Now let $v\equiv v(\rho)\dvtx[0,+\infty)\times\bb T^d\to[0,1]$ be the
solution to (\ref{2he*}). Theorem~\ref{tcfhe*} implies that $v\in
\mathcal M_{[0,+\infty)}(\upbar{\rho})$ and therefore $\theta v \in
\mathcal M_{(-\infty,0]}(\upbar{\rho})$. Since $I^{-\nabla
U-\wwidetilde{E}}_{[0,T]}(v)=0$ for every $T>0$, it holds $I^{-\nabla
U-\wwidetilde{E}}_{[0,+\infty)}(v)=0$. Considering (\ref{zuppa}) when
$\pi=\theta v$ we get $I^{-\nabla
U+\wwidetilde{E}}_{(-\infty,0]}(\theta v)=
\mathcal{F}^U_{\upbar{\rho}}(\rho)$. Whence
$\widehat{V}{}^E_{\upbar{\rho}}(\rho) \le\mathcal
F^U_{\upbar{\rho}}(\rho)$.
\end{pf*}
\begin{pf*}{Proof of Theorem~\ref{tqp=f} (characterization of the
minimizer)} As the previous argument implies that $\theta v$ is a
minimizer for (\ref{genqpdav}), it remains only to prove uniqueness.
Suppose that $\pi^*$ is a minimizer for (\ref{genqpdav}). By
(\ref{zuppa}), it necessarily holds $I_{[0,+\infty)}^{-\nabla
U-\wwidetilde{E}}(\theta\pi^*)=0$ and, by monotonicity, this is
possible if and only if $I_{[0,T]}^{-\nabla U-\wwidetilde{E}}(\theta
\pi^*)=0$ for any $T>0$. This is\vspace*{2pt} equivalent to say that $\theta
\pi^*$ is a weak solution to (\ref{2he*}) in any time interval
$[0,T]$. Whence $\pi^*=\theta v$.
\end{pf*}
%
%
\begin{lemma}
\label{tinter}
Fix $\upbar{\rho}\in(0,1)$ and let $\gamma\in M(\upbar{\rho})$ be
such that $\delta\le\gamma\le1- \delta$ for some $\delta\in
(0,1/2)$. Then there exist a constant $C=C(\delta)>0$, a time $T_0>0$
and a path $\pi^0\in\mc M_{[0,T_0]}$ such that
$\pi^0_{0}=\gamma_{\upbar{\rho}}$, $\pi^0_{T_0}=\gamma$ and
\[
I_{[0,T_0]}^E (\pi^0|\gamma_{\upbar{\rho}}) \le C
\|\gamma- \gamma_{\upbar{\rho}} \|_{W^{1,2}}^2.
\]
\end{lemma}
\begin{pf}
Elementary computations (see, e.g.,~\cite{BDGJLn+1}, Lemma 4.3) show
that, by taking $T_0=1$, the ``straight'' path
$\pi_t = \gamma t+ \gamma_{\upbar{\rho}} (1-t)$ fulfils the
requirements.
\end{pf}
\begin{pf*}{Proof of Theorem~\ref{tqp=f} (the identity
${V}^E_{\upbar{\rho}} =\widehat{V}{}^E_{\upbar{\rho}}$)} Fix
$\upbar{\rho}\in[0,1]$ and $\rho\in M(\upbar{\rho})$.
Recall\vspace*{1pt} that any path $\pi\in\mathcal M_{[-T,0]}$ such that
$I^E_{[-T,0]}(\pi|\gamma_{\upbar{\rho}})<+\infty$ satisfies necessarily
the condition $\pi_{-T}=\gamma_{\upbar{\rho}}$. This means that if we
extend $\pi$ to an element $\widehat{\pi}\in \mathcal
M_{(-\infty,0]}(\upbar{\rho})$ by setting $\widehat{\pi}_t
=\gamma_{\upbar{\rho}}$ for $t\in(-\infty,-T)$, we then have
$I^E_{(-\infty,0]}(\hat\pi)=I^E_{[-T,0]}( \pi|\gamma_{\upbar{\rho}})$.
This readily implies the inequality
$\widehat{V}{}^E_{\upbar{\rho}}(\rho)\leq V^E_{\upbar{\rho}}(\rho)$.

Since we have already proven that $\widehat{V}{}^E_{\upbar{\rho}}
=\mathcal{F}^U_{\upbar{\rho}}(\rho)$, it is enough to show
$V^E_{\upbar{\rho}}\leq\mathcal F^U_{\upbar{\rho}}$.
Fix $\upbar{\rho}\in(0,1)$. We need to prove the following
statement. For each $\rho\in M(\upbar{\rho})$ and $\epsilon>0$
there exist a time $T>0$ and a path $\pi\in\mc M_{[-T,0]}$ such
that $\pi_{-T}=\gamma_{\upbar{\rho}}$, $\pi_0=\rho$ and
$I^E_{[-T,0]}(\pi| \gamma_{\upbar{\rho}})\le\mc
F_{\upbar{\rho}}^U(\rho)+\epsilon$.\vadjust{\goodbreak}

Let $v(\rho)$ be the solution to (\ref{2he*}). Given $\epsilon_1>0$
to be chosen later, by Theorem~\ref{tcfhe*}, there exists a time
$T_1$ such that $\| v_{T_1}(\rho)-
\gamma_{\upbar{\rho}}\|_{W^{1,2}}\le\epsilon_1$.
Set\vspace*{1pt} $\gamma:= v_{T_1}(\rho)$; by Lemmas~\ref{tvbd} and
\ref{tinter} there exists a time $T_0$ and a path $\pi^0\in\mc
M_{[-T_1-T_0,-T_1]}$ such that
$\pi^0_{-T_1-T_0}=\gamma_{\upbar{\rho}}$, $\pi^0_{-T_1}=\gamma$ and
$I^E_{[-T_1-T_0,-T_1]}(\pi^0|\gamma_{\upbar{\rho}}) \le C
\epsilon_1^2$. We claim the path $\pi\in\mc M_{[-T_1-T_0,0]}$
defined by
\[
\pi_t:=
\cases{
\pi^0_{t}, &\quad if $t\in[-T_0-T_1,-T_1)$, \vspace*{1pt}\cr
(\theta v (\rho))_t, &\quad if $t\in[-T_1,0]$,}
\]
fulfils the above requirement with $T=T_0+T_1$. Since $\pi$ is
continuous, we indeed have
\begin{eqnarray*}
I^{-\nabla U+\wwidetilde E}_{[-T,0]}(\pi|\gamma_{\upbar{\rho}})
& = &
I^{-\nabla U+\wwidetilde E}_{[-T_1-T_0,-T_1]}(\pi^0|\gamma_{\upbar
{\rho}})+
I^{-\nabla U+\wwidetilde E}_{[-T_1,0]}(\theta v)
\\
&\le& C
\epsilon_1^2+\mathcal F^U_{\upbar{\rho}}(\rho)
-\mathcal F^U_{\upbar{\rho}}(\gamma)
+I^{-\nabla U-\wwidetilde{E}}_{[0,T_1]}(v)
\leq C \epsilon_1^2+\mathcal F^U_{\upbar{\rho}}(\rho).
\end{eqnarray*}
We conclude the proof choosing $\epsilon_1$ small enough.
\end{pf*}

\section*{Acknowledgments}

We thank G. Jona-Lasinio for useful discussions.\break
D.~Gabrielli acknowledges the Department of Physics of the
University ``La Sapienza'' for the kind hospitality.


%

\printaddresses

\end{document}